\documentclass{article}
\usepackage[utf8]{inputenc}
\usepackage[top = 1in, bottom = 1in, left = 1in, right =1in]{geometry}
\usepackage{latexsym,amsmath,amssymb,amsthm, graphicx}
\usepackage{algorithm,algorithmicx,algpseudocode}
\usepackage{enumitem}
\usepackage{hyperref}
\usepackage{mathtools}
\usepackage{xcolor}
\usepackage{subcaption}
\usepackage{graphicx}
\usepackage{soul}
\usepackage{bbold}
\usepackage{dsfont}
\usepackage{appendix}
\usepackage{mathtools}
\usepackage[makeroom]{cancel}
\usepackage[normalem]{ulem}

\allowdisplaybreaks
\makeatletter
\newcommand{\pushright}[1]{\ifmeasuring@#1\else\omit\hfill$\displaystyle#1$\fi\ignorespaces}
\newcommand{\pushleft}[1]{\ifmeasuring@#1\else\omit$\displaystyle#1$\hfill\fi\ignorespaces}
\makeatother

\newcommand{\SB}[1]{#1^+} 
\newcommand{\norm}[1]{\left\lVert#1\right\rVert}

\setlist[itemize]{leftmargin=*}

\newcounter{thm_count}
\newcounter{claim_count}
\theoremstyle{remark} % theorem, lemma, and corollary environments
\newtheorem*{theorem*}{Theorem} %usage: \begin{theorem*} ...\end{theorem*}
\newtheorem{theorem}[thm_count]{\bf Theorem} %usage: \begin{theorem} ...\end{theorem}
\newtheorem{lemma}[thm_count]{\bf Lemma} %usage: \begin{lemma} ...\end{lemma}
\newtheorem{claim}[claim_count]{\bf Claim} %usage: \begin{corollary} ...\end{corollary}

\newcounter{assump}
\newcounter{defn}
\newcounter{prbl}
\theoremstyle{remark}
\newtheorem{problem}[prbl]{\bf Problem}
\newtheorem{assumption}[assump]{\bf Assumption}

\newtheorem{remark}[defn]{\bf Remark}

\renewenvironment{proof}[1][Proof:]{\begin{trivlist}
\item[\hskip \labelsep {\bfseries #1}]}{\end{trivlist}}

\title{Distributed Optimization for Client-Server Architecture \\with Negative Gradient Weights \footnote{This research is supported in part by National Science Foundation awards 1421918 and 1610543, and Toyota InfoTechnology Center. Any opinions, findings, and conclusions or recommendations expressed here are those of the authors and do not necessarily reflect the views of the funding agencies or the U.S. government. }}

\author{Shripad Gade $\qquad$ Nitin H. Vaidya \\ \\ Department of Electrical and Computer Engineering, and \\ Coordinated Science Laboratory, \\ University of Illinois at Urbana-Champaign. \\ Email: \{gade3, nhv\}@illinois.edu \\ \\ Technical Report\footnote{Submitted: 12 August, 2016, Revised: 16 December, 2016. Added Section~3.1, added additional discussion to Section 5, added references.}}

\date{}

\begin{document}

\maketitle
\vspace*{-0.25in}

\begin{abstract}
Availability of both massive datasets and computing resources have made machine learning and predictive analytics extremely pervasive. In this work we present a synchronous algorithm and architecture for distributed optimization motivated by privacy requirements posed by applications in machine learning. We present an algorithm for the recently proposed multi-parameter-server architecture. We consider a group of parameter servers that learn a model based on randomized gradients received from clients. Clients are computational entities with private datasets (inducing a private objective function), that evaluate and upload randomized gradients to the parameter servers. The parameter servers perform model updates based on received gradients and share the model parameters with other servers. We prove that the proposed algorithm can optimize the overall objective function for a very general architecture involving $C$ clients connected to $S$ parameter servers in an arbitrary time varying topology and the parameter servers forming a connected network. 
\end{abstract}

\section{Introduction}
We study a system with $C$ clients. Each client ($C_i$) maintains a private, convex function $f_i(x)$ and we are interested in solving the optimization problem,
\begin{align}
    \underset{x \in \mathcal{X}}{\min} \; f(x) := \sum_{i=1}^C f_i(x).
\end{align}
Each $f_i: \mathbb{R}^D \rightarrow \mathbb{R}$ is a differentiable convex function, with the gradients being Lipschitz continuous. The decision variable $x \in \mathcal{X} \subseteq \mathbb{R}^D$ is a $D$ dimensional real vector. We present a novel distributed optimization algorithm for the recently proposed multiple parameter server architecture for machine learning \cite{abadi2016tensorflow}. This algorithm and architecture are motivated by the need to preserve privacy in the learning process. We discuss the privacy enhancement due to our algorithm later in the report. 

Optimization is at the heart of Machine Learning algorithms. Recent work has focused on distributed variants of optimization algorithms \cite{nesterov2012efficiency,liu2015asynchronous,agarwal2011distributed,Singh2014, Nedic2007}. Several solutions to distributed optimization of convex functions are proposed for myriad scenarios involving directed graphs \cite{Nedi2015}, link failures and losses \cite{hadjicostis2016robust}, asynchronous communication models \cite{nedic2011asynchronous,wei20131,zhang2014asynchronous}, stochastic functions  \cite{agarwal2011distributed,ram2010distributed,ram2009incremental}, fault tolerance \cite{su2015fault,Su2016podc} and differential privacy \cite{huang2015differentially}. Popular methods in convex optimization show guaranteed learning under the assumption that the learning rate being a positive, monotonically non-increasing sequence with constraints on its sum and sum of squares (see Eq.~\ref{Eq:LearnStepCond}). In this work we present an interleaved consensus and projected gradient descent algorithm allowing the possibility of negative weight being multiplied to the learning rate. We show that this algorithm introduces privacy preserving properties in the learning algorithm while conserving correctness and convergence properties.  

Parameter Server framework, presented in \cite{Li2013,NIPS2014_5597}, is shown to improve efficiency, elastic scalability and fault tolerance in distributed machine learning tasks. Asymptotic convergence of delayed proximal gradient method, in parameter server framework is proved by Li \textit{et al.} in \cite{Li2013}. Abadi \textit{et al.} demonstrate via experiments that an architecture with multiple parameter servers can efficiently perform distributed machine learning \cite{abadi2016tensorflow}. However, no analysis exists for proving convergence in multi-parameter server architectures. We address this in our work and prove convergence in learning tasks performed on multiple parameter server architecture.  

Privacy enhancing distributed optimization algorithms are critical while learning models from sensitive data. Datasets involving medical or financial information are extremely sensitive, and outmost care has to be taken to prevent any leakage of information. Experimental results of privacy preserving deep learning were presented in \cite{shokri2015privacy} and demonstrate  $\epsilon$-differential privacy in experiments. However, Shokri \textit{et al.}, in \cite{shokri2015privacy}, do not provide convergence analysis or correctness guarantees for their algorithm. Abadi \textit{et al.} in \cite{abadi2016deep} present a learning algorithm and present analysis of privacy costs under the differential privacy framework. In \cite{gade16PLN}, we show that arbitrary function sharing approach provides privacy. We present precise privacy definitions and prove privacy claims\footnote{Interested readers are referred to \cite{gade16PLN} and references therein for a discussion on existing privacy preserving optimization algorithms.}. In this work, we show accurate synchronous learning in multiple parameter server architecture (\cite{abadi2016tensorflow}).

\subsection{Organization}
Learning problem formulation, assumptions, learning system architecture and notation are presented in Section~\ref{Sec:ProbF}. Interleaved consensus and descent algorithm is proposed in Section~\ref{Sec:Algorithm}. Correctness and convergence results for the iterative algorithm are presented in Section~\ref{Sec:ConvergenceResults}. Short discussion on privacy enhancing aspect of our algorithm is tabled in Section~\ref{Sec:PrivacyDiscussion}. Simulation results validating correctness and convergence claims, and related explanations and discussion are presented in Section~\ref{Sec:SimulationResults}. 

\section{Problem Formulation and Assumptions} \label{Sec:ProbF}
We consider a problem setting where a group of entities collaboratively learn a model by minimizing an additive cost function. Clients ($C_i$) have access to a private cost function ($f_i(x)$). Each cost function ($f_i(x)$) is induced by a dataset (privately stored with client $C_i$). These datasets may be formed by independently collected data or may be partitioned from a much larger central dataset resulting in disjoint or overlapping partitions. The underlying model can be parametrized by a vector $x \ (\in \mathbb{R}^D)$ of dimension $D$. The set of all feasible parameter vectors is called the decision set, $\mathcal{X} \ (\subseteq \mathbb{R}^D)$. The learning process involves finding an optimal decision vector that optimizes sum of functions ($f_i(x)$) over the set $\mathcal{X} $. The clients are connected to one or more parameter servers in an arbitrary time-varying topology (albeit with some connectivity constraint). 

Figure~\ref{Fig:Schematic} shows client-server architecture with 4 parameter servers and 7 clients. The clients are connected to the servers in an arbitrary time-varying topology. We assume that the clients are connected to one or more servers at least once every $\Delta$ iterations (Assumption~\ref{Asmp:QConn}). The parameter servers are connected between themselves in an arbitrary possibly time-varying topology (with the only constraint being that the servers form a connected component, Assumption~\ref{Asmp:SerConn}). The clients download latest parameter vectors (also referred to as server-states) from the servers and compute gradients for their local, private cost function. At every iteration, the clients upload randomized gradients to one or more servers. Every server collects all received gradients, and updates its parameter vector using projected gradient descent rule. The above steps occur synchronously for $\Delta$ iterations. The projected gradient descent steps are followed by a consensus step among all servers. The algorithm is described in detail in Section~\ref{Sec:Algorithm}.  

\begin{figure}[!t]
    \centering
    \includegraphics[height=3in]{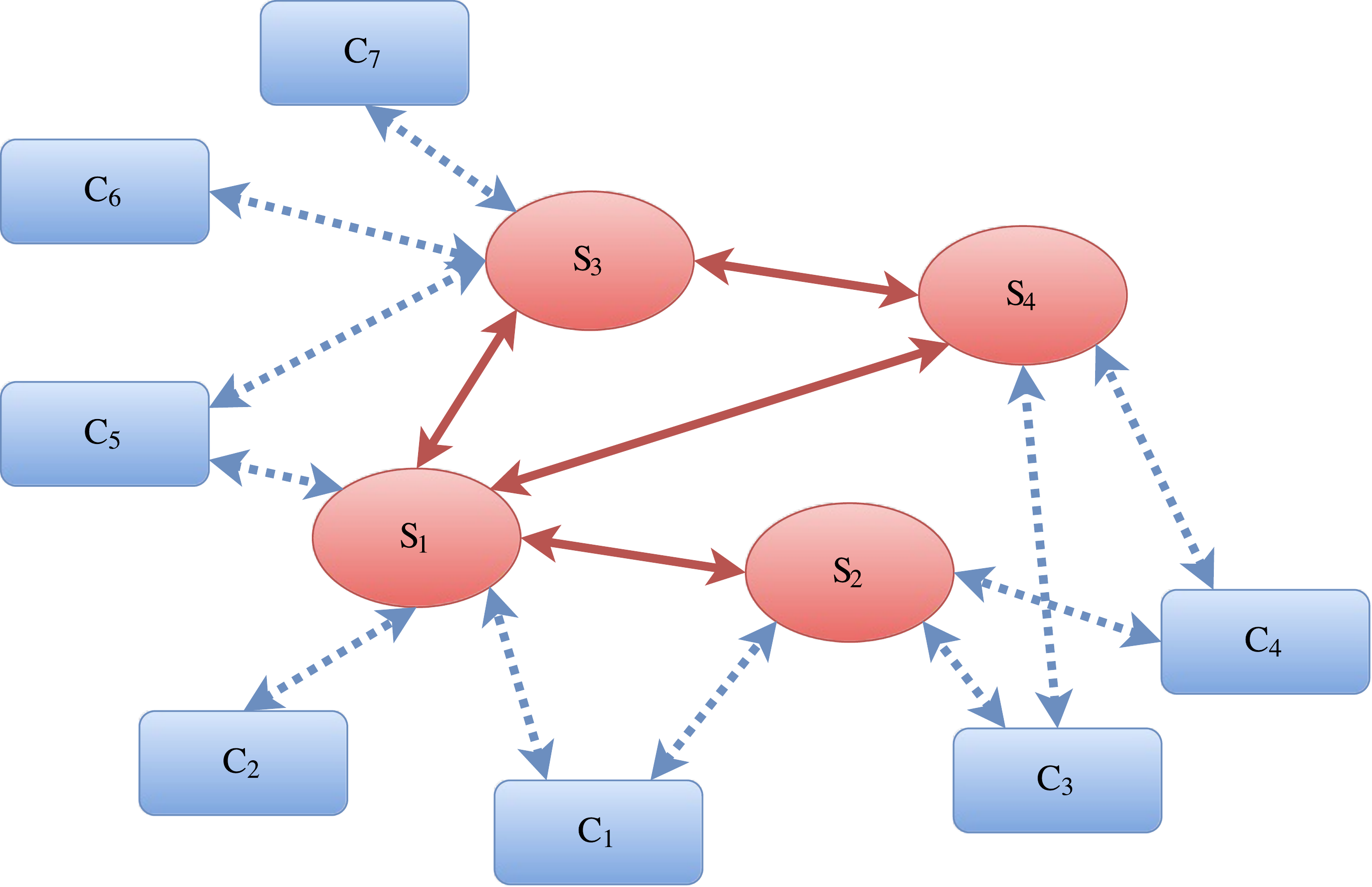}
    \caption{A schematic for Client-Server architecture: $S_i$, (red ovals) are parameter servers and $C_i$, (blue squares) are clients.}
    \label{Fig:Schematic}
\end{figure}

We impose following assumptions on the functions $f_i(x)$ and on the decision set, $\mathcal{X}$:  
\begin{assumption} [Objective Functions] \label{Asmp:Function} 
The objective functions $f_i : \mathbb{R}^D \rightarrow \mathbb{R}, \; \forall \; i = 1, 2,{ }\ldots, C$ are continuously differentiable convex functions of parameter vector $x \in \mathcal{X}$.  %twice continuously differentiable ($f_i \in \mathcal{C}^2$)
\end{assumption}
\begin{assumption} [Decision Set]
The feasible parameter vector set, $\mathcal{X}$, is a convex-compact subset of $\mathbb{R}^D$. \label{Asmp:Set}
\end{assumption}

We make a boundedness assumption on the gradient ($g_h(x)$) of function $f_h(x)$ in Assumption~\ref{Asmp:SubBound}. We also make an additional assumption on the Lipschitz continuity of gradients $g_h(x)$.  
\begin{assumption} [Gradient Boundedness] \label{Asmp:SubBound}
Let $g_i(x)$ denote the gradient of the function $f_i(x)$. There exist scalars $L_1, L_2, \ldots, L_C $ such that, $\| g_i(x) \| \leq L_i; \; \forall \; i = 1, 2, \cdots, C \; \text{and} \; \forall x \in \mathcal{X}$.
\end{assumption}
\begin{assumption}[Gradient Lipschitzness]
\label{Asmp:GradLip}
Each function gradient ($g_h(x)$) is assumed to be Lipschitz continuous, i.e. there exist scalars $N_h > 0$ such that, $\|g_h(x) - g_h(y) \| \leq N_h \| x -y\|$ for all $x \neq y; \ \forall x, y \in \mathcal{X}$ and for all $h = 1, 2, \ldots, C$. 
\end{assumption}

%\begin{assumption} [Synchrony]\label{Asmp:Synchrony}
%All communications are assumed to be synchronous. 
%\end{assumption}

%\begin{assumption} [Communication Loss] \label{Asmp:LossLess}
%All communications are assumed to be loss-less (no packet drops). 
%\end{assumption}

%\begin{assumption} [Non Faulty Agents] \label{Asmp:NoFault}
%Servers or clients do not experience any faults.
%\end{assumption}

For the purpose of this report we will assume without explicitly stating that all communication links are synchronous and loss less. All agents are assumed to operate perfectly and do not experience any faults. As discussed above in the distributed optimization architecture, the server-client and server-server topologies satisfy the following assumptions, 
\begin{assumption} [Server-Client $\Delta$-Connectedness] \label{Asmp:QConn}
Over a period of $\Delta$ steps, every client has uploaded gradients to at least one of the servers. 
\end{assumption}

\begin{assumption} [Server-Server Connectedness] \label{Asmp:SerConn}
The server graph forms a connected component at every consensus step.  
\end{assumption}

The distributed learning problem is formally articulated in Problem~\ref{Prob:GlobalOpt}.  

\begin{problem} [Distributed Learning] \label{Prob:GlobalOpt}
Learning a model implies finding a minimizer to the objective function ($f(x)$) defined as the sum of functions ($f_i(x)$), subject to Assumptions ~\ref{Asmp:Function} and \ref{Asmp:Set}, i.e.,
\begin{align*}
\text{find} \; x^* \in \underset{x \in \mathcal{X}}{\text{argmin}} \; f(x), 
\end{align*}
where $f(x) = \sum_{i = 1}^{C} \; f_i (x)$.
\end{problem}

\subsection{Notation} \label{Sec:Notation}
Let us denote the number of servers by $S$ and the number of clients by $C$. Upper case alphabets ($I, J, K$ etc.) are used to index servers and lower case alphabets ($j, h$) are used to index clients. We use the symbol ``$\sim$" to denote communication link and information sharing between entities. As an example, $I \sim G$ denotes that servers $I$ and $G$ have a communication link between them, and correspondingly $I \cancel{\sim} J$ denotes that servers $I$ and $J$ do not share information (cannot communicate) with each other. The dimension of the problem (number of parameters in the decision vector) is denoted by $D$.

%The learning algorithm involves performing projected gradient update steps followed by consensus step (every $\Delta$ steps). 
Iterations number is denoted by the 2-tuple $\{i,k\}$. The first element in the iteration number denotes the number of steps after elapsed after a consensus step and the second element denotes the number of consensus steps elapsed. In this algorithm we perform consensus every $\Delta$ steps. Consequently, the total iterations elapsed at $\{i,k\}$ can be calculated to be $k \ \Delta + i$. The decision vector (also referred to as iterate from now on) stored in server $I$ at time $\{i,k\}$ is denoted by $x^I_{i,k}$, where the superscript denotes the server-id, the subscript denotes the time index. $x^J_{i,k}[p]$ ($p = 1, 2, \ldots, D$) denotes the $p^{th}$ dimension in decision vector $x^J_{i,k}$. %The overall number of projected gradient steps applied can be found by $(\Delta)k + i$, while the total number of steps can be evaluated by $( \Delta + 1)k + i$. 

The average of iterates at time instant $\{i,k\}$ is denoted by $\bar{x}_{i,k}$. 
\begin{equation}
\bar{x}_{i,k} = \frac{1}{S} \sum_{J=1}^S x^J_{i,k}. \label{Eq:deltaDef}
\end{equation}
We denote the disagreement of an iterate ($x^J_{i,k}$) with the iterate average ($\bar{x}_{i,k}$) by $\delta^J_{i,k}$.
\begin{equation}
\delta^J_{i,k} = x^J_{i,k} - \bar{x}_{i,k}.
\label{Eq:deltaDef2}
\end{equation}
At time instants $\{0,k\}$ we denote $\delta^J_{0,k}$ by $\delta^J_k$ for simplicity. 
\begin{equation}
\delta^J_{k} = \delta^J_{0,k}.
\label{Eq:deltaDef3}
\end{equation}

\noindent We use $\tilde{.}$ to denote a vector that is stacked by its coordinates. As an example, consider three vectors in $\mathbb{R}^3$ given by ${a_1} = [a_x, \ a_y, \ a_z]^T$, ${a_2} = [b_x, \ b_y, \ b_z]^T$, ${a_3} = [c_x, \ c_y, \ c_z]^T$. Let us represent ${a} = [{a_1}, \ {a_2}, \ {a_3}]^T$, then $\tilde{{a}} = [a_x, \ b_x, \ c_x, \ a_y, \ b_y, \ c_y, \ a_z, \ b_z, \ c_z]^T$. Similarly we can write stacked model parameter vector as, 
\begin{equation}
\tilde{x}_{0,k} = [x^1_{0,k}[1], x^2_{0,k}[1], \ldots, x^S_{0,k}[1], x^1_{0,k}[2], x^2_{0,k}[2], \ldots, x^S_{0,k}[2], \ldots, x^1_{0,k}[D], \ldots, x^S_{0,k}[D]]^T. \label{Eq:TildeVecDef}
\end{equation}

We use $g_h(x_{i,k})$ to denote the gradient of function $f_h(x)$ evaluated at $x_{i,k}$. Remember, $L_h$ are bounds on gradients from Assumption~\ref{Asmp:SubBound} and $N_h$ are Lipschitz constants from Assumption~\ref{Asmp:GradLip}, for all $h = 1, 2, \ldots, C$. We define constants $\SB{L} = \sum_{h=1}^C L_h$ and $\SB{N} = \sum_{h=1}^C N_h$ to be used later in the analysis. 

$\|.\|$ denotes standard Euclidean norm for vectors, and matrix 2-norm for matrices ($\|A\| = \sqrt{\lambda_{\max}(A^\dagger A)} = \sigma_{\max}(A)$, where $A^\dagger$ denotes conjugate transpose of matrix A, and $\lambda, \ \sigma$  are eigenvalues and singular values respectively). 

Throughout this report, we will use the following definitions and notation regarding the optimal solution ($x^*$), the set of all optima ($\mathcal{X}^*$) and the function value at optima ($f^*$),
$$ f^* = \inf_{x \in \mathcal{X}} f(x), \qquad \mathcal{X}^* = \{x \in \mathcal{X} | f(x) = f^*\}, \qquad dist(x, \mathcal{X}^*) = \inf_{x^* \in \mathcal{X}^*} \|x - x^*\|.$$ 
\noindent The optimal function value, at the solution of the optimization problem or the minimizing state vector is denoted by $x^*$, is denoted by $f^*$.

\section{Interleaved Consensus and Descent Algorithm} \label{Sec:Algorithm}
We proposed an interleaved consensus and projected gradient descent algorithm to find a distributed solution to Problem~\ref{Prob:GlobalOpt}. We consider a scenario with multiple parameter servers (called servers) and clients. As described in Section~\ref{Sec:ProbF}, the servers are connected to each other in an arbitrary topology (Assumption~\ref{Asmp:SerConn}), and the clients are connected to the servers in a time varying topology (Assumption~\ref{Asmp:QConn}). Projected gradient descent and consensus steps are performed in interleaved fashion. 

%\noindent \underline{\textbf{Projected Gradient Descent}}: \\
\subsubsection*{Projected Gradient Descent}
In the projected gradient descent phase, the clients download the current states from the parameter server(s) (that are connected to the client) and upload corresponding weighted gradients to the respective parameter servers. The projected gradient step is given by the following update rule, 
\begin{equation}
x^I_{i,k} = \mathcal{P}_{\mathcal{X}} \left[ x^I_{i-1,k} - \alpha_k \sum_{h=1}^C W_{i-1,k}[I,h] \nabla f_h(x^I_{i-1,k}) \right], \label{Eq:PGDE}
\end{equation}
where $W_{i-1,k}[I,h]$ is the weight assigned (by client $h$) to an gradient update sent to server $I$. These computations and gradient uploads happen synchronously, at every time step. Projected Descent without weights is a well known iterative gradient based method that guarantees convergence to optimum under reducing learning rate ($\alpha_k$) \cite{bertsekas1976goldstein}. We assume that the monotonically non-increasing learning rate/step possesses the following properties, 
\begin{equation}
    \alpha_k > 0, \ \forall k \geq 0; \quad \alpha_{k+1} \leq \alpha_k, \ \forall k \geq 0; \quad  \sum_{k=0}^\infty \alpha_k = \infty; \ \text{and} \quad \sum_{k=0}^\infty \alpha_k^2 < \infty. \label{Eq:LearnStepCond}
\end{equation}

The weights used by clients at iteration $i$ after $k^{th}$ consensus steps are expressed in a ${S \times C}$ matrix, $W_{i,k}$. The $(I,j)^{th}$ entry of matrix $W_{i,k}[I,j]$ denotes the weight that client $j$ assigns to server $I$. The  weights are further characterized by the Symmetric Learning Condition (SLC) and the Bounded Update Condition (BUC). The SLC ensures that over a period of $\Delta$ steps, every client $h$ (objective function $f_h(x)$) is assigned equal weight ($=M$).
\begin{assumption} [SLC] \label{Asmp:SLC}
Equal weights are assigned to updates from every client over a period of $\Delta$ steps. More formally, if the sum of weight matrices over a period of $\Delta$ steps is denoted by $W = \sum_{i = 1}^{\Delta} W_{i-1,k}$, then the sum of all entries in a column (of matrix $W$) is a constant $M$.  
\begin{align}
\mathbb{1}_{[1 \times S]} \ W = \mathbb{1}_{[1 \times S]} \left[ \sum_{i = 1}^{\Delta} W_{i-1,k} \right] = M \ \mathbb{1}_{[1 \times C]}. && \ldots \text{for some scalar $M>0$}
\end{align}
\end{assumption}

The BUC ensures that the sum of absolute values of weights is upper bounded. BUC is needed because, if negative weights are permitted, the weights may oscillate with arbitrary large amplitudes (while SLC is still satisfied). BUC restricts the clients from selecting arbitrarily large weights.   
\begin{assumption} [BUC] \label{Asmp:SLC-BUC}
The sum of absolute value of weights (over $\Delta$ steps)is upper bounded by a constant ($\bar{M} > 0$). Mathematically, if we denote the sum of absolute values of weight matrices as $\bar{W} = \sum_{i = 1}^{\Delta} | W_{i-1,k} | $, then the sum of entries in every column of $\bar{W}$ is upper bounded by a constant ($\bar{M} > 0$).  
\begin{align}
\mathbb{1}_{[1 \times S]} \ \bar{W} = \mathbb{1}_{[1 \times S]} \left[ \sum_{i = 1}^{\Delta} | W_{i-1,k} |  \right] \leq \bar{M} \ \mathbb{1}_{[1 \times C]}. && \ldots \text{for some scalar $\bar{M}>0$}
\end{align}
\end{assumption}

%\noindent \underline{\textbf{Consensus}}: \\
\subsubsection*{Consensus}
The projected gradient descent steps are followed by a consensus step between parameter servers. A consensus step is performed after every $\Delta$ iterations. For an arbitrary incomplete graph we use a doubly stochastic weight matrix $B_k$ to write the update rule. The parameters at server $I$ are updated by,
\begin{equation}
x^I_{0,k+1} = \sum_{J=1}^S B_k[I,J] x^J_{\Delta,k}. \label{Eq:ConsensusUpdateRelation1}
\end{equation}

\noindent If we assume a fully connected topology for servers, the update rule can be written as, 
\begin{equation}
x^I_{0,k+1} = \frac{1}{S} \sum_{J=1}^S x^J_{\Delta,k}.
\end{equation}

The weight matrix $B_k$ has to be a row stochastic matrix for the decision vectors to enter the synchronization manifold and $B_k$ has to be column stochastic for the update to be average of parameter vectors at servers \cite{tsi1984phdthesis}. Methods for constructing $B_k$ matrix based on the graph topology are well known and we point the reader to Refs.~\cite{nedic2009distributed,blondel2005convergence,xiao2007distributed} for a more thorough treatment. We construct a doubly stochastic weight matrix for a bidirectional communication graph between servers (denoted as $B_{k}$ with entries $B_{k}[I,J] \geq 0$) using the following rules. 

\noindent For all $I, J = 1, 2, \ldots, S$ and all $k \geq 0$, the following hold:
\begin{enumerate}[label=(\alph*)]
\item $B_{k}[I,J] \geq \kappa$ for all $J \sim I$ at $\{i,k\}$ and $B_{k}[I,J] = 0$ for $J \cancel{\sim} I$, for some $\kappa>0$;
\item $\sum_{J=1}^N B_{k}[I,J] = 1 \quad \forall \; I$;
\item $\sum_{I=1}^N B_{k}[I,J] = 1 \quad \forall \; J$.
\end{enumerate}
The matrix so constructed assigns positive weights (lower bounded by $\kappa$) to all neighbors. The last two conditions enforce row and column stochasticity respectively.

The Interleaved Consensus and Descent Algorithm is formally presented in Algorithm~\ref{Algo:ASLearn} and \ref{Algo:ASLearn2}. Server side algorithm for information fusion, state update and consensus operation is described in Algorithm~\ref{Algo:ASLearn} while Client side algorithm for gradient upload is presented in Algorithm~\ref{Algo:ASLearn2}.

\begin{algorithm}[!t]
\caption{Server - Interleaved Consensus and Descent Algorithm}
\begin{algorithmic}[1]
\State Input: $x^J_{0,k}$, $\alpha_k$, $\Delta$, NSteps \Comment{NSteps - Termination Criteria}     
\State Result: $x^* = \underset{x \in \mathcal{X}} {argmin} \sum_{i=1}^{C} f_i(x) $ 
\For {k = 1 to NSteps} 
    \For {i = 0 to $ \Delta$-1}     \Comment{$\Delta$ Projected Gradient Steps}
        \For {J = 1 to $S$}   \Comment{Server updates states based on gradients received}
            \State $x^J_{i,k} = \mathcal{P}_\mathcal{X} \left[ x^J_{i-1,k} - \alpha_k \sum_{h=1}^C W_{i-1,k}[J,h] \; g_h(x^J_{i-1,k})\right]$
        \EndFor
    \EndFor
    \For {J = 1 to $S$} \Comment {Server performs consensus step with neighbors}
        \State $x^J_{0,k+1} = \sum_{I=1}^S B_k[J,I] \ x^I_{ \Delta,k}$ \Comment{$B_k\text{ is } S\times S $ doubly stochastic matrix}
    \EndFor
\EndFor
\end{algorithmic}
\label{Algo:ASLearn}
\end{algorithm}
\begin{algorithm}[!t]
\caption{Client - Interleaved Consensus and Descent Algorithm}
\begin{algorithmic}[1]
\State Input: $x^J_{i,k}$, $\Delta$, NSteps \Comment{NSteps - Termination Criteria}     
\State Result: $W[J,h]g_h(x^J_{i,k})$ for all $J = 1, 2, \ldots, S$
\For {k = 1 to NSteps} 
    \For {i = 0 to $ \Delta$-1}     \Comment{$\Delta$ Projected Gradient Steps}
        \For {J = 1 to $S$}   
            \State Download: $x^J_{i,k}$ 
            \State Compute: $g_h(x^J_{i,k})$ \Comment{Compute gradient}
            \State Select: $W_{i,k}[J,h]$ \Comment{Decide on random weight (SLC \& BUC)}
            \State Upload: $W_{i,k}[J,h]\; g_h(x^J_{i,k})$ \Comment{Client uploads gradient update to Server $J$}
        \EndFor
    \EndFor
\EndFor
\end{algorithmic}
\label{Algo:ASLearn2}
\end{algorithm}

\subsection{Coordinate-wise Independently Weighted Updates}

In the algorithm above, we assume that weights $W_{i,k}[J,h]$ are applied to the entire gradient vectors $\nabla f_h (x^J_{i,k})$. However, we can easily show that we can choose different weights for different coordinates of the gradient vector, i.e. $W^{p}_{i,k}[J,h]$ is used as weight for $\nabla f_h (x^J_{i,k})[p]$ for $p = 1, 2, \ldots, D$. If we use different weights for every coordinates the projected gradient descent update rule (Eq.~\ref{Eq:PGDE}) gets modified to, 
\begin{align}
x^I_{i,k}[p] = \mathcal{P}_{\mathcal{X}} \left[ x^I_{i-1,k}[p] - \alpha_k \sum_{h=1}^C W^p_{i-1,k}[I,h] \nabla f_h(x^I_{i-1,k})[p] \right]. && (p = 1, 2, \ldots, D)
\end{align}
To account for the coordinate-wise independently weighted updates, the symmetric learning condition will have to be satisfied coordinate wise, 
\begin{align}
\mathbb{1}_{[1 \times S]} \ W^p = \mathbb{1}_{[1 \times S]} \left[ \sum_{i = 1}^{\Delta} W^p_{i-1,k} \right] = M \ \mathbb{1}_{[1 \times C]}. && \ldots \text{for scalar $M>0$, and $p = 1, 2, \ldots, D$}
\end{align}
The bounded update condition is also satisfied coordinate wise,
\begin{align}
\mathbb{1}_{[1 \times S]} \ \bar{W}^p = \mathbb{1}_{[1 \times S]} \left[ \sum_{i = 1}^{\Delta} | W^p_{i-1,k} |  \right] \leq \bar{M} \ \mathbb{1}_{[1 \times C]}. && \ldots \text{for scalar $\bar{M}>0$, and $p = 1, 2, \ldots, D$}
\end{align}
The proofs developed in this report will hold even for coordinate wise weighted updates, since we estimate bounds for a summation of weights over $\Delta$ steps in the proof, rather than the weights themselves. If the SLC and BUC hold coordinate wise, the coordinate wise selection of weights does not alter the convergence proofs (albeit a few constants in the bound will be slightly different).

\section{Convergence Results} \label{Sec:ConvergenceResults}

We prove convergence result for the most general case first. Elements of $W$ matrix can be negative numbers, however, the symmetric learning (Assumption~\ref{Asmp:SLC}) and bounded update (Assumption~\ref{Asmp:SLC-BUC}) conditions hold for arbitrary time varying server-client communication graph and an incomplete server-server communication graph. It is followed by convergence results for more restrictive cases involving complete server graph (instead of Assumption~\ref{Asmp:SerConn}, see Section~\ref{Sec:4.1}) and complete server graph with non-negative weight matrix $W_{i,k}$ (see Section~\ref{Sec:4.2}).

We state two important results that will be useful in convergence analysis, the first being on convergence of non-negative almost supermartingales by Robbins and Siegmund (Theorem 1, \cite{robbins1985convergence}) followed by Lemma 3.1 (b) by Ram \textit{et al.}, \cite{ram2010distributed}.
\begin{lemma}  \label{Lem:RobSiegConv}
Let ($\Omega, \mathcal{F}, \mathcal{P}$) be a probability space and let $\mathcal{F}_0 \subset \mathcal{F}_1 \subset \ldots$ be a sequence of sub $\sigma-$fields of $\mathcal{F}$. Let $u_k, v_k$ and $w_k$, $k = 0, 1, 2, \ldots$ be non-negative $\mathcal{F}_k-$ measurable random variables and let \{$\gamma_k$\} be a deterministic sequence. Assume that $\sum_{k=0}^\infty \gamma_k < \infty$, and $\sum_{k=0}^\infty w_k < \infty$ and $$E[u_{k+1}|\mathcal{F}_k] \leq  (1+\gamma_k) u_k -v_k + w_k,$$ holds with probability 1. Then, the sequence \{$u_k$\} converges to a non-negative random variable and $\sum_{k=0}^\infty v_k < \infty$.
\end{lemma}

\begin{lemma} \label{Lem:Ram}
Let $\{\zeta_k\}$ be a non-negative sequence scalar sequence. If $\sum_{k=0}^\infty \zeta_k < \infty$ and $0 < \beta < 1$, then $\sum_{k=0}^\infty \left( \sum_{j=0}^k \beta^{k-j} \zeta_j \right) < \infty$.
\end{lemma}

\noindent The well known non-expansive property (cf. \cite{bertsekas2003convex}) of Euclidean projection onto a non-empty, closed, convex set $\mathcal{X}$, is represented by the following inequality, $\forall \; x, y \in \mathbb{R}^D$, 
\begin{align}
    \| \mathcal{P}_\mathcal{X}[x] - \mathcal{P}_\mathcal{X}[y] \| \leq \|x - y\|. \label{Eq:EUCPROJINEQUAL}
\end{align}

We now present the relationship of server iterates between two consensus steps (at time instants 
$\{0,k\}$ and $\{0,k+1\}$) in the following Lemma.
\begin{lemma} \label{Lem:IterateConvAnyW}
Under Assumptions~\ref{Asmp:Function}, \ref{Asmp:Set}, \ref{Asmp:SubBound}, \ref{Asmp:GradLip}, \ref{Asmp:QConn} and \ref{Asmp:SerConn} a sequence {$x^J_{0,k}$} generated by an interleaved consensus and projected gradient algorithm following symmetric learning condition and bounded update condition (Assumptions~\ref{Asmp:SLC}, \ref{Asmp:SLC-BUC} respectively) satisfies, for all $y \in \mathcal{X}$,
\begin{align}
\eta_{k+1}^2 &\leq \left(1 + \gamma_k \right) \eta_{k}^2 - 2\alpha_k M \left(f(\bar{x}_{0,k}) - f(y)\right) + {\alpha_k^2}(C_1^2 + C_2^2 + C_3^2) + \alpha_k C_4 \max_{J}\{\|\delta^J_k\|\}, \label{Eq:Interim5}
\end{align}
where, $\delta^J_k$ (Eq.~\ref{Eq:deltaDef3}) is the disagreement between the iterate and its average (Eq.~\ref{Eq:deltaDef}, \ref{Eq:deltaDef2}),
\begin{align}
\eta_k^2 &= \sum_{J=1}^S \|x^J_{0,k}-y\|^2, \quad 
\gamma_k = \frac{1}{S} (2 \alpha_k^2 \bar{M}^2 \SB{N} \SB{L} + \alpha_k \bar{M}\SB{N} \max_J \|\delta^J_{k}\|), \nonumber \\
C_1^2 &= \sum_{i=1}^{ \Delta}\sum_{J=1}^S \left( \sum_{h=1}^C | W_{i-1,k}[J,h] | \  L_h \right)^2, \quad C_2^2 = 4(\bar{M}\SB{L})^2, \quad C_3^2 = 2 \bar{M}^2 \SB{N} \SB{L}, \quad \text{and} \nonumber \\
C_4 &= \bar{M} (2 \SB{L} + \SB{N}). \nonumber 
\end{align}
\end{lemma}

\begin{proof}
Every server accepts gradients from its clients and updates states based on projected gradient descent algorithm,  Eq.~\eqref{Eq:PGDE}, where $\mathcal{P}_\mathcal{X}$ is the projection operator,
\begin{align}
    x^J_{i,k} = \mathcal{P}_{\mathcal{X}}\left[ x^J_{i-1,k} - \alpha_k \sum_{h=1}^C W_{i-1,k}[J,h] \; g_h(x^J_{i-1,k})\right]. && \ldots  J = \{1, 2, \ldots, S\} \label{Eq:E1}
\end{align}

\noindent Using the non-expansive property of the projection operator ($\mathcal{X}$ is a non-empty closed, convex set) on the update equation Eq.~\eqref{Eq:E1}, $\forall \ J = 1, 2, \ldots, S \text{ and } y \in \mathcal{X}$,
\begin{align}
    \|x^J_{i,k} - y\|^2 &= \|\mathcal{P}_{\mathcal{X}}\left[ x^J_{i-1,k} - \alpha_k \sum_{h=1}^C W_{i-1,k}[J,h] \; g_h(x^J_{i-1,k})\right] - y\|^2, && \ldots  \mathcal{P}_{\mathcal{X}}[y] = y \\
    &\leq \| x^J_{i-1,k} - \alpha_k \sum_{h=1}^C W_{i-1,k}[J,h] \; g_h(x^J_{i-1,k}) - y\|^2. \label{Eq:E2} && \ldots \text{Eq.~\ref{Eq:EUCPROJINEQUAL}}
\end{align}
\noindent Note that the norm used for all vectors and matrices in this report are Euclidean norms or 2-norms (see Section~\ref{Sec:Notation}). We add the inequality in Eq.~\eqref{Eq:E2} for all servers (all $J$),
\begin{align}
    \sum_{J=1}^S \|x^J_{i,k} - y\|^2 &\leq \sum_{J=1}^S \| x^J_{i-1,k} - \alpha_k \sum_{h=1}^C W_{i-1,k}[J,h] \; g_h(x^J_{i-1,k}) - y\|^2, \\
    &\leq \sum_{J=1}^S \|x^J_{i-1,k} - y\|^2 + \alpha_k^2 \sum_{J=1}^S \norm{\sum_{h=1}^C W_{i-1,k}[J,h] \ g_h(x^J_{i-1,k})}^2 \nonumber \\
    & \qquad \qquad \qquad - 2 \alpha_k \sum_{J=1}^S \left(\sum_{h=1}^C W_{i-1,k}[J,h] \ g_h(x^J_{i-1,k})\right)^T \left(x^J_{i-1,k} - y \right).\label{Eq:E3}
\end{align}

\noindent Adding inequality in Eq.~\eqref{Eq:E3} for iterations between consensus steps, i.e. adding for $i = 1, 2, \ldots, \Delta$, and cancelling telescoping terms we get, 
\begin{align}
    \sum_{J=1}^S \|x^J_{\Delta,k} - y\|^2 &\leq \sum_{J=1}^S \|x^J_{0,k} - y\|^2 + \alpha_k^2 C_1^2 \nonumber \\
    & \qquad - 2 \alpha_k \sum_{i=1}^{\Delta} \sum_{J=1}^S \left( \left(\sum_{h=1}^C W_{i-1,k}[J,h] \ g_h(x^J_{i-1,k})\right)^T \left(x^J_{i-1,k} - y \right) \right), \label{Eq:E4}
\end{align}
where, $C_1^2 = \sum_{i=1}^{ \Delta}\sum_{J=1}^S \left( \sum_{h=1}^C | W_{i-1,k}[J,h] | \  L_h \right)^2$. \\

We now use consensus relationship between iterates at $\{\Delta,k\}$ and $\{0,k+1\}$. We start by stacking the state vector for all servers component/coordinate wise. We use $\tilde{.}$ notation to denote a vector that is stacked by its coordinates (see Eq.~\ref{Eq:TildeVecDef} for definition). $\tilde{y}$ denotes $S$ copies of $y$ vector stacked coordinate wise. We know that in D-dimension the consensus step can be written using Kronecker products \cite{fax2004information}. 
\begin{align}
    \tilde{x}_{0,k+1} &= (I_D \otimes B_k) \tilde{x}_{\Delta,k}. \label{Eq:E5} && \ldots \text{Consensus Step} \\
    \tilde{x}_{0,k+1} - \tilde{y} &= (I_D \otimes B_k) (\tilde{x}_{\Delta,k} - \tilde{y}).  && \ldots \text{Subtracting $\tilde{y}$ in Eq.\eqref{Eq:E5} and } (I_D \otimes B_k)\tilde{y} = \tilde{y}
\end{align}

\noindent Taking norm of both sides of the equality (2-norm),
\begin{align}
    \|\tilde{x}_{0,k+1} - \tilde{y}\|_2^2 &= \|(I_D \otimes B_k) (\tilde{x}_{\Delta,k} - \tilde{y})\|_2^2, && \ldots \text{Norms of equal vectors are equal}  \\
    &\leq \|(I_D \otimes B_k)\|_2^2 \|(\tilde{x}_{\Delta,k} - \tilde{y})\|_2^2. && \ldots \text{$\|Ax\|_2 \leq \|A\|_2 \|x\|_2$} \label{Eq:Eigenvallt1}
\end{align}

\noindent We use the following property of eigenvalues of Kronecker product of matrices. If $A$ ($m$ eigenvalues given by $\lambda_i$, with $i = 1, 2, \ldots, m$) and $B$ ($n$ eigenvalues given by $\mu_j$, with $j = 1, 2, \ldots, n$) are two matrices then the eigenvalues of the Kronecker product $A \otimes B$ are given by $\lambda_i \mu_j$ for all $i$ and $j$ ($mn$ eigenvalues). Hence, the eigenvalues of $I_D \otimes B_k$ are essentially $D$ copies of eigenvalues of $B_k$. Since $B_k$ is a doubly stochastic matrix, its eigenvalues are upper bounded by 1. Clearly, $\|(I_D \otimes B_k)\|_2^2 = \lambda_{\max}((I_D \otimes B_k)^\dagger(I_D \otimes B_k)) \leq 1$. This follows from the fact that $(I_D \otimes B_k)^\dagger(I_D \otimes B_k)$ is a doubly stochastic matrix since product of two doubly stochastic matrices is also doubly stochastic. Note that if $B_k$ is only row stochastic, $\lambda_{\max}((I_D \otimes B_k)^\dagger(I_D \otimes B_k))$ may be greater than 1. Hence, we need the transition matrix $B_k$ to be doubly stochastic. From, Eq.~\ref{Eq:Eigenvallt1} and the fact that $\|(I_D \otimes B_k)\|_2^2 \leq 1$ we get \footnote{Inequality in Eq.~\ref{Eq:ConsBoundktok1} can also be proved using first principles and doubly stochastic nature of $B_k$, without using eigenvalue arguments, see Appendix~\ref{sec:Apx0} for proof.},
\begin{align}
    \|\tilde{x}_{0,k+1} - \tilde{y}\|_2^2 &\leq \|(\tilde{x}_{\Delta,k} - \tilde{y})\|_2^2. \label{Eq:ConsBoundktok1}
\end{align} 

\noindent Furthermore, the square of the norm of a stacked vector is equal to sum of the square of the norms of all servers. 
\begin{align}
    \sum_{J=1}^S \|x^J_{0,k+1} - y\|^2 = \|\tilde{x}_{0,k+1} - \tilde{y}\|_2^2 \leq \|(\tilde{x}_{\Delta,k} - \tilde{y})\|_2^2 = \sum_{J=1}^S \|x^J_{\Delta,k} - y\|^2. \label{Eq:E6}
\end{align}

\noindent Combining Eq.~\eqref{Eq:E4} and \eqref{Eq:E6}, we get the following relationship, 
\begin{align}
    \sum_{J=1}^S \|x^J_{0,k+1} - y\|^2  &\leq \sum_{J=1}^S \|x^J_{0,k} - y\|^2 + \alpha_k^2 C_1^2 \nonumber \\
    & \qquad \qquad \underbrace{- 2 \alpha_k \sum_{i=1}^{\Delta} \sum_{J=1}^S \left[ \left(\sum_{h=1}^C W_{i-1,k}[J,h] \ g_h(x^J_{i-1,k})\right)^T \left(x^J_{i-1,k} - y \right) \right]}_{\Lambda}. \label{Eq:E9}
\end{align}

\noindent We use the term $\eta_k^2$ to represents the RMS like error term. 
\begin{equation}
\eta^2_{k} = \sum_{J=1}^S \|x^J_{0,k} - y\|^2.
\label{Eq:RMSErrorDef}
\end{equation}

We will construct a few reltions and establish a few bounds before we present a bound on the term $\Lambda$. We start by writing an expression for $x^J_{i,k}$,
\begin{equation}
x^J_{i,k} = x^J_{i-1,k} - \alpha_k \left[ \sum_{h=1}^C W_{i-1,k} [J,h] g_h(x^J_{i-1,k})\right] + e^J_{i-1}, 
\label{Eq:ProjErrBound}
\end{equation}
where $e^J_{i-1}$ is the difference term between a gradient descent and its projection at $i^{th}$ iteration. Recursively performing this unrolling operation to relate the iterates at time $\{i,k\}$ to $\{0,k\}$ we get,
$$ x^J_{i,k} = x^J_{0,k} - \alpha_k \sum_{t = 1}^{i} \left[ \sum_{h=1}^C W_{t-1,k} [J,h] g_h(x^J_{t-1,k})\right] + \sum_{t = 1}^{i} e^J_{t-1}.$$

The sequence of iterates belongs to the decision set ($\mathcal{X}$) i.e. $x^J_{i,k} \in \mathcal{X}$. Hence the distance of the iterates from the decision set is zero, i.e. $dist(x^J_{i,k},\mathcal{X}) = dist(x^J_{i-1,k},\mathcal{X}) = 0$. Clearly from Eq.~\ref{Eq:ProjErrBound}, the projection error is upper bounded by the deviation due to the gradient descent step, i.e. \begin{equation}
\|e^J_{i-1}\| \leq \alpha_k \|\sum_{h=1}^C W_{i-1,k} [J,h] g_h(x^J_{i-1,k}) \|.
\label{Eq:ProjErrBoundNew}
\end{equation}

We hence arrive at the following expression, that will be used later in the analysis. Note the expressions for terms $\alpha_k D^J_{i-1,k}$ and $\alpha_k E^J_{i-1,k}$ and their bounds, they will be used later.
\begin{align}
 y - x^J_{i-1,k} &= y - x^J_{0,k} + \underbrace{\alpha_k \sum_{t = 1}^{i-1} \left[ \sum_{h=1}^C W_{t-1,k} [J,h] g_h(x^J_{t-1,k})\right]}_{\alpha_k D^J_{i-1,k}} - \underbrace{\sum_{t = 1}^{i-1} e^J_{t-1}}_{\alpha_k E^J_{i-1,k}} \nonumber \\
 &\triangleq y - x^J_{0,k} + \alpha_k ( D^J_{i-1,k} - E^J_{i-1,k}). \label{Eq:EstErrEqn-fc}
 \end{align}
We obtain bounds on both $\sum_{J=1}^S D^J_{i-1,k}$ and $\sum_{J=1}^S E^J_{i-1,k}$ using the property $\| \sum (.) \| \leq \sum \| (.) \|$ successively,
\begin{align}
\sum_{J=1}^S \|E^J_{i-1,k} \| &= \sum_{J=1}^S \| \frac{1}{\alpha_k} \sum_{t = 1}^{i-1} e^J_{t-1}\| \leq  \frac{1}{\alpha_k} \sum_{J=1}^S \sum_{t = 1}^{i-1} \|e^J_{t-1}\| \leq \frac{1}{\alpha_k} \sum_{J=1}^S \sum_{t = 1}^{i-1} \alpha_k\| \sum_{h = 1}^C W_{i-1,k}[J, h] g_h(x^J_{i-1,k}) \| \nonumber \\
&\leq \sum_{J=1}^S \sum_{t = 1}^{i-1}  \sum_{h = 1}^C | W_{i-1,k}[J, h] | \ \|g_h(x^J_{i-1,k}) \| \leq \sum_{J=1}^S \sum_{t = 1}^{ \Delta}  \sum_{h = 1}^C | W_{i-1,k}[J, h] | \ \|g_h(x^J_{i-1,k}) \| \nonumber \\
&\leq \bar{M} \SB{L}. \quad \quad \ldots \text{Replacing $i-1$ with $ \Delta$, adding positive terms to the bound.}  \label{Eq:EBound} \\
\sum_{J=1}^S\|D^J_{i-1,k}\| & = \sum_{J=1}^S \| \sum_{t = 1}^{i-1} \left[ \sum_{h=1}^C W_{t-1,k} [J,h] g_h(x^J_{t-1,k})\right] \| \leq \sum_{J=1}^S \sum_{t = 1}^{i-1} \left[ \sum_{h=1}^C | W_{t-1,k} [J,h] | \ \|g_h(x^J_{t-1,k})\|\right] \nonumber \\
& \leq \sum_{J=1}^S \sum_{t = 1}^{ \Delta} \left[ \sum_{h=1}^C | W_{t-1,k} [J,h] | \ \|g_h(x^J_{t-1,k})\|\right] \leq \bar{M} \SB{L}. \label{Eq:DBound}
\end{align} 
Quite naturally we also get, $\|D^K_{i-1,k}\| \leq \sum_{J=1}^S \|D^J_{i-1,k}\|$ and $\|E^K_{i-1,k}\| \leq \sum_{J=1}^S \|E^J_{i-1,k}\|$ for all $K$. \\

The gradient is Lipschitz continuous, hence from Assumption~\ref{Asmp:GradLip} we have, $\| g_h(\bar{x}_{0,k}) - g_h(x^J_{i-1,k})\| \leq N_h \| \bar{x}_{0,k} - x^J_{i-1,k}\|$. Hence we can express the gradient at point $x^J_{i-1,k}$ as the gradient at point $\bar{x}_{0,k}$ plus a norm bounded vector,
\begin{align} 
g_h(x^J_{i-1,k}) =  g_h(\bar{x}_{0,k}) + l^J_{h,i-1}, 
\label{Eq:UnrollGradFirst}
\end{align}
where, the vector $l^J_{h,i-1}$ is such that $\|l^J_{h,i-1}\| \leq N_h \| x^J_{i-1,k} - \bar{x}_{0,k}\|$. We rewrite this bound using triangle inequality, $\|l^J_{h,i-1}\| \leq N_h \| x^J_{i-1,k} - x^J_{0,k} + x^J_{0,k} - \bar{x}_{0,k}\| \leq N_h \|x^J_{i-1,k} - x^J_{0,k}\| + N_h \|x^J_{0,k} - \bar{x}_{0,k}\|$. We further bound $\| x^J_{i-1,k} - x^J_{0,k} \|$ using Eq.~\ref{Eq:EstErrEqn-fc} and further use triangle inequality to obtain bound on $l^J_{h,i-1}$ itself. 
\begin{align*}
    \| x^J_{i-1,k} - x^J_{0,k} \| &= \alpha_k \| D^J_{i-1,k} - E^J_{i-1,k}\| \leq \alpha_k (\| D^J_{i-1,k}\| + \|E_{i-1,k}\|) \leq 2 \alpha_k \bar{M} \SB{L}, \\
    \| l^J_{h,i-1} \| &\leq N_h \| x^J_{i-1,k} - x^J_{0,k}\| + N_h \|x^J_{0,k} - \bar{x}_{0,k}\| \leq 2 \alpha_k \bar{M} N_h \SB{L} + N_h \| \delta^J_k\|,
\end{align*}
where $\delta^J_k \triangleq x^J_{0,k} - \bar{x}_{0,k}$ defines (Eq.~\ref{Eq:deltaDef2}, \ref{Eq:deltaDef3}) the misalignment of iterate (of server $J$) with the average of the iterates. We now bound the largest magnitude of $l^J_{h,i}$ vector,
\begin{equation}
\max_{J,h} \{ \|l^J_{h,i}\| \} \leq \max_{J,h} \{2 \alpha_k \bar{M} N_h \SB{L} + N_h \| \delta^J_k\|\} \leq 2 \alpha_k \bar{M} \SB{N} \SB{L} + \SB{N} \max_J \{\| \delta^J_k\|\}. \label{Eq:MaxLBound}
\end{equation}

We now begin to construct a bound on the term $\Lambda$,
\begin{align}
    \Lambda &= - 2 \alpha_k \sum_{i=1}^{\Delta} \sum_{J=1}^S \left[ \left(\sum_{h=1}^C W_{i-1,k}[J,h] \ g_h(x^J_{i-1,k})\right)^T \left(x^J_{i-1,k} - y \right) \right] \nonumber \\
    & \text{Absorbing the negative sign in $\left(x^J_{i-1,k} - y \right)$ and using Eq.~\eqref{Eq:EstErrEqn-fc} and adding-subtracting $\bar{x}_{0,k}$,} \nonumber \\
    &= 2 \alpha_k \sum_{i=1}^{\Delta} \sum_{J=1}^S \left[ \left(\sum_{h=1}^C W_{i-1,k}[J,h] \ g_h(x^J_{i-1,k})\right)^T \left( y - \bar{x}_{0,k} + \bar{x}_{0,k} - x^J_{0,k} + \alpha_k (D^J_{i-1,k} - E^J_{i-1,k}) \right) \right] \\
    & \text{We now unroll the gradient (using Gradient Lipschitzness) from Eq.~\eqref{Eq:UnrollGradFirst}, to get,} \nonumber \\
    &= 2 \alpha_k \sum_{i=1}^{\Delta} \sum_{J=1}^S \left[ \left(\sum_{h=1}^C W_{i-1,k}[J,h] \ \left(g_h(\bar{x}_{0,k}) + l^J_{h,i-1} \right) \right)^T \left( y - \bar{x}_{0,k} + \bar{x}_{0,k} - x^J_{0,k} + \alpha_k (D^J_{i-1,k} - E^J_{i-1,k}) \right) \right] \\
    & \text{Rearranging terms we get,} \nonumber \\
    &\leq \underbrace{2 \alpha_k \sum_{i=1}^{\Delta} \sum_{J=1}^S \left[ \left( \sum_{h=1}^C W_{i-1,k}[J,h] \ g_h(\bar{x}_{0,k})\right)^T (y - \bar{x}_{0,k}) \right]}_{T1} \nonumber \\
    & \quad + \underbrace{2 \alpha_k \sum_{i=1}^{\Delta} \sum_{J=1}^S \left[ \left( \sum_{h=1}^C W_{i-1,k}[J,h] \ l^J_{h,i-1}\right)^T (y - \bar{x}_{0,k}) \right]}_{T2} \nonumber \\
    & \quad + \underbrace{2 \alpha_k \sum_{i=1}^{\Delta} \sum_{J=1}^S \left[ \left(\sum_{h=1}^C W_{i-1,k}[J,h] \ \left(g_h(\bar{x}_{0,k}) + l^J_{h,i-1} \right) \right)^T (\bar{x}_{0,k} - x^J_{0,k}) \right]}_{T3} \nonumber \\
    & \quad + \underbrace{2 \alpha_k^2 \sum_{i=1}^{\Delta} \sum_{J=1}^S \left[ \left(\sum_{h=1}^C W_{i-1,k}[J,h] \ \left(g_h(\bar{x}_{0,k}) + l^J_{h,i-1} \right) \right)^T (D^J_{i-1,k} - E^J_{i-1,k}) \right]}_{T4}. \label{Eq:BoundonLambdaInterim1}
\end{align}

\noindent The first term ($T_1$), in the above expression, is upper bounded by using the gradient/subgradient inequality for convex functions. We use the independence of $g_h(\bar{x}_{0,k})$ term to $i,\ J$ and independence of $(y - \bar{x}_{0,k})$ to $i,\ J,\ h$, to rearrange the order of summation in $T_1$.
\begin{align}
T_1 &= 2 \alpha_k  \sum_{J=1}^S \left[ \left( \sum_{h=1}^C \left[\sum_{i=1}^{\Delta} W_{i-1,k}[J,h] \right]\ g_h(\bar{x}_{0,k})\right)^T (y - \bar{x}_{0,k}) \right] \nonumber \\
&= 2 \alpha_k  \sum_{h=1}^C\left[ \left( \left[\sum_{J=1}^S \sum_{i=1}^{\Delta} W_{i-1,k}[J,h] \right]\ g_h(\bar{x}_{0,k})\right)^T (y - \bar{x}_{0,k}) \right] \nonumber \\
&= 2 \alpha_k \left[ M \sum_{h=1}^C g_h(\bar{x}_{0,k}) \right]^T (y - \bar{x}_{0,k}) = 2 \alpha_k \left(g(\bar{x}_{0,k})\right)^T (y - \bar{x}_{0,k}) \nonumber \\
&\leq - 2 \alpha_k M (f(\bar{x}_{0,k}) - f(y)). \nonumber 
\end{align}

\noindent We bound the second term ($T_2$), third term ($T_3$) and the fourth term ($T_4$) by replacing it with its norm (strengthening the inequality). We also extensively use the property $\norm{\sum_i a_i} \leq \sum_i \norm{a_i}$ for some $a_i$.
\begin{align}
T_2 &\leq \|T_2\| = 2 \alpha_k \norm{\sum_{i=1}^{\Delta} \sum_{J=1}^S \left[ \left( \sum_{h=1}^C W_{i-1,k}[J,h] \ l^J_{h,i-1}\right)^T (y - \bar{x}_{0,k}) \right]} \nonumber \\
&\leq 2 \alpha_k \sum_{h=1}^C \left[ \left( \sum_{J=1}^S \sum_{i=1}^{\Delta} | W_{i-1,k}[J,h] | \right) \right] \max_{J,h} \|l^J_{h,i-1}\| \|(y - \bar{x}_{0,k})\| \nonumber \\
& \leq 2 \alpha_k \bar{M} (\max_{J,h} \|l^J_{h,i-1}\|) \|\bar{x}_{0,k} - y\| \nonumber \\
&\leq 2 \alpha_k \bar{M} (2 \alpha_k \SB{N} \bar{M} \SB{L} + \SB{N} \max_J \{\| \delta^J_k\|\}) \|\bar{x}_{0,k}-y\|. && \ldots \text{Eq.~\ref{Eq:MaxLBound}} \nonumber \\
%%%%%%%%%%%%%%%%%%%%%%%%%%%%%%%%%%%%%%%%%%%%%%%%
T_3 &\leq \|T_3\| = 2 \alpha_k \norm{\sum_{i=1}^{\Delta} \sum_{J=1}^S \left[ \left(\sum_{h=1}^C W_{i-1,k}[J,h] \ \left(g_h(\bar{x}_{0,k}) + l^J_{h,i} \right) \right)^T (\bar{x}_{0,k} - x^J_{0,k}) \right]} \nonumber \\
&\leq 2 \alpha_k \left(\sum_{h=1}^C \left[ \left( \sum_{J=1}^S \sum_{i=1}^{\Delta} | W_{i-1,k}[J,h] | \right) \right] \right) \SB{L} \|\bar{x}_{0,k} - x^J_{0,k}\|  \nonumber \\
&\leq 2 \alpha_k \bar{M} \SB{L} \max_J{\|\delta^J_{0,k}\|}. \nonumber \\
%%%%%%%%%%%%%%%%%%%%%%%%%%%%%%%%%%%%%%%%%%%%%%%
T_4 &\leq \|T_4\| = 2 \alpha_k^2 \norm{\sum_{i=1}^{\Delta} \sum_{J=1}^S \left[ \left(\sum_{h=1}^C W_{i-1,k}[J,h] \ \left(g_h(\bar{x}_{0,k}) + l^J_{h,i-1} \right) \right)^T (D^J_{i-1,k} - E^J_{i-1,k}) \right]} \nonumber \\
&\leq 2 \alpha_k^2 \sum_{h=1}^C \left[ \sum_{J=1}^S \sum_{i=1}^\Delta |W_{i-1,k}[J,h]| \right] \|g_h(x^J_{i-1,k})\| (\|D^J_{i-1,k}\| + \|E^J_{i-1,k}\|)\nonumber \\
&\leq 4 \alpha_k^2 (\bar{M}\SB{L})^2. \nonumber
\end{align} \\
Combining the above bounds (on $T_1$, $T_2$, $T_3$ and $T_4$) with the bound on $\Lambda$ (Eq.~\ref{Eq:BoundonLambdaInterim1}), 
\begin{align}
\Lambda &\leq - 2 \alpha_k M (f(\bar{x}_{0,k}) - f(y)) + 2 \alpha_k \bar{M} (2 \alpha_k \SB{N} \bar{M} \SB{L} + \SB{N} \max_J \{\| \delta^J_k\|\}) \|\bar{x}_{0,k}-y\| \nonumber \\
& \qquad + 2 \alpha_k \bar{M} \SB{L} \max_J{\|\delta^J_{0,k}\|} + 4 \alpha_k^2 (\bar{M}\SB{L})^2. \label{Eq:E7}
\end{align} 
%& \text{,} \nonumber \\
    %& \text{third and fourth term by adding respective norms to the expression (i.e. strengthening the inequality),} \nonumber \\
    %&\leq  +  +  +  \label{Eq:E7}
    
\noindent We now try to rewrite (new bound) the second term in Eq.~\eqref{Eq:E7} (bound on $T_2$). We begin with definition of $\bar{x}_{0,k}$, followed by application of Jensen's inequality,\footnote{If $x_1, x_2, \ldots, x_n$ are $n$ non-negative real numbers, and $m, n$ are positive integers then, $\left(\frac{x_1+x_2+\ldots+x_n}{n}\right)^m \leq \left(\frac{x_1^m+x_2^m+\ldots+x_n^m}{n}\right)$. In this case, $m = 2$ and $n = S$.}
\begin{align}
    \bar{x}_{0,k} - y &= \left( \frac{1}{S} \sum_{J=1}^S x^J_{0,k} \right) - y = \frac{1}{S} \sum_{J=1}^S (x^J_{0,k} - y), \label{Eq:DEFAVGITER} && \ldots \text{Eq.~\ref{Eq:deltaDef}} \\
    \|\bar{x}_{0,k} - y\|^2 &= \norm{\frac{1}{S} \sum_{J=1}^S (x^J_{0,k} - y)}^2 \leq \left(\frac{1}{S} \sum_{J=1}^S \|(x^J_{0,k} - y)\|\right)^2 && \ldots \text{Eq.~\ref{Eq:DEFAVGITER} and $\triangle$ Inequality} \nonumber \\
    &\leq \frac{1}{S} \sum_{J=1}^S \|(x^J_{0,k} - y)\|^2 = \frac{1}{S} \eta_k^2. && \ldots \text{Jensen's Inequality} \label{Eq:BoundonNormSquarexy}
\end{align}
Using the property that $2\|a\| \leq 1 + \|a\|^2$ where we consider $\|a\| = \|\bar{x}_{0,k} - y\|$, followed by using the bound derived in Eq.~\ref{Eq:BoundonNormSquarexy} we can rewrite the bound in Eq.~\eqref{Eq:E7} as,
\begin{align}
    \Lambda &\leq - 2 \alpha_k M (f(\bar{x}_{0,k}) - f(y)) + \alpha_k \bar{M} (2 \alpha_k \SB{N} \bar{M} \SB{L} + \SB{N} \max_J \{\| \delta^J_k\|\}) (1 + \frac{1}{S}\eta_k^2) \nonumber \\
    & \qquad + 2 \alpha_k \bar{M} \SB{L} \max_J{\|\delta^J_{0,k}\|} + 4 \alpha_k^2 (\bar{M}\SB{L})^2.
\end{align}\\

\noindent Merging the above bound on $\Lambda$ with Eq.~\eqref{Eq:E9}, we get,
\begin{align}
    \sum_{J=1}^S \|x^J_{0,k+1} - y\|^2  &\leq \left( 1+ \underbrace{2 \alpha_k^2 \frac{1}{S} \bar{M}^2 \SB{N} \SB{L} + \alpha_k \frac{1}{S} \bar{M}\SB{N} \max_J \|\delta^J_{k}\|}_{\gamma_k} \right) \underbrace{\sum_{J=1}^S \|x^J_{0,k} - y\|^2}_{\eta_k^2} \nonumber \\
    & \quad - 2 \alpha_k M (f(\bar{x}_{0,k}) - f(y)) + \alpha_k \underbrace{\left(2 \bar{M} \SB{L} + \bar{M}\SB{N} \right) }_{C_4}\max_J{\|\delta^J_{k}\|} \nonumber \\
    & \quad + \alpha_k^2 \underbrace{\left(C_1^2 + 4(\bar{M}\SB{L})^2 + 2 \bar{M}^2 \SB{N} \SB{L}\right)}_{C_1^2 + C_2^2 + C_3^2}.
\end{align}
$\hfill \blacksquare$
\end{proof}

\begin{lemma} \label{Lem:deltaJbound}
Let iterates be generated by interleaved consensus and projected gradient descent algorithm (Algorithms~\ref{Algo:ASLearn} and \ref{Algo:ASLearn2}), and Assumptions~\ref{Asmp:Function}, \ref{Asmp:Set}, \ref{Asmp:SubBound}, \ref{Asmp:GradLip}, \ref{Asmp:QConn} and  \ref{Asmp:SerConn} hold, then there exists constant $\nu<1$, such that the following bound on the maximum (over $J$) disagreement between iterates at servers and the average iterate given by $\delta^J_k$ (Eq.~\ref{Eq:deltaDef}, \ref{Eq:deltaDef2} and \ref{Eq:deltaDef3}) holds,
%$$\max_J \{ \| \delta^J_{k+1} \| \} \leq \left(\frac{S-1}{2}\right) \left(\nu^k \max_{P,Q} \left( \|x^P_{0,0} - x^Q_{0,0}\| \right) + 4 \bar{M}\SB{L} \sum_{i = 1}^k \left( \alpha_i \nu^{k-i+1} \right) \right), $$
$$\max_J\{\|\delta^J_{k+1}\|\} \leq \frac{S-1}{S} \left(\nu^{k+1} \max_{P,Q} \left( \|x^P_{0,0} - x^Q_{0,0}\| \right) + 4 \bar{M}\SB{L} \sum_{i = 1}^k \left( \alpha_i \nu^{k-i+1} \right) \right).$$
\end{lemma}

\begin{proof}
We use Kronecker product to write consensus step as shown in Eq.~\ref{Eq:E5}. This step is equivalent to the following form of representing the consensus step (Eq.~\ref{Eq:ConsensusUpdateRelation1}),
\begin{equation}
x^I_{0,k+1} = \sum_{J=1}^S B_k [I,J] x^J_{\Delta,k},
\end{equation}
where $x^I_{0,k+1}$ and $x^J_{\Delta,k}$ represent the parameter vector at servers $I$ and $J$ at time instant $\{0,k+1\}$ and $\{\Delta,k\}$ respectively while $B_k [I,J]$ represents the $I^{th}$ row and $J^{th}$ column entry of transition matrix $B_k$. We can then write the difference between parameter vectors at server $I$ and $G$ as,
\begin{equation}
x^I_{0,k+1} - x^G_{0,k+1} = \sum_{J=1}^S \left( B_k[I,J] - B_k[G,J] \right) x^J_{\Delta,k}. \label{Eq:ROWSTOC}
\end{equation}
Since $B_k$ is doubly stochastic, clearly the coefficients of states in Eq.~\ref{Eq:ROWSTOC} add up to zero (i.e. $\sum_{J=1}^S ( B_k[I,J] - B_k[G,J] ) = 0$). Collecting all positive coefficients and negative coefficients and rearranging we get the following equation,
\begin{align}
x^I_{0,k+1} - x^G_{0,k+1} = \sum_{P,Q} \eta_{P,Q} (x^P_{\Delta,k} - x^Q_{\Delta,k}),  && \ldots \forall \; I, G %\footnote{Remember that $x \sim y$ notation denotes that $x$ and $y$ are neigbors and can communicate with each other.} 
\label{Eq:REARRANG}
\end{align}
where, $\eta_{P,Q} \geq 0$ is the weight associated to servers $P$ and $Q$ and $\eta_{P,Q} \geq 0$. Note that all coefficients $\eta_{P,Q}$ refer to some $I, G$ pair at time $k$. For simplicity in notation we will ignore $I, G$ and $k$ without any loss of generality or correctness.

\begin{assumption} \label{Asmp:Scrambling}
The transition matrix $B_k$ is assumed to be a scrambling matrix. \cite{seneta2006non}
\end{assumption}
From Assumption~\ref{Asmp:Scrambling}, any two rows of  $B_k$ matrix have a common non-zero column entry. Also, since any entry of $B_k$ is less than 1, the difference is also strictly less than 1. Hence, $\sum \eta_{PQ} < 1$. By taking norm on both sides of equality in \eqref{Eq:REARRANG}, recalling  Assumption~\ref{Asmp:Scrambling} and using property of norm of sum, we get for any $I, G$ pair, 
\begin{equation}
\|x^I_{0,k+1} - x^G_{0,k+1}\| \leq \left( \sum_{P,Q} \eta_{P,Q} \right) \max_{P,Q} \|x^P_{\Delta,k} - x^Q_{\Delta,k}\|. 
\end{equation}
Since the above inequality is valid for all $I, G$, we can rewrite the above relation as,
\begin{align}
 \max_{I, G} \|x^I_{0,k+1} - x^G_{0,k+1}\| \leq \max_{I,G}\left( \sum_{P,Q} \eta_{P,Q} \right) \max_{P,Q} \|x^P_{\Delta,k} - x^Q_{\Delta,k}\|. \label{Eq:LAB1}
\end{align}
Note that, $\max_{I,G} \left( \sum \eta_{P,Q} \right)$ is dependent only on the topology at time $k$ (i.e. the doubly stochastic weight matrix given by $B_k$). Due to the countable nature of possible topologies for $S$ servers, we can define a new quantity $\nu = \max_k \{ \max_{I,G} \{ \sum \eta_{P,Q} \} \}$ \footnote{Note that $\eta_{P,Q}$ is dependent on $I,G$ pair and $k$.}. By definition, $\max_{I,G} \{ \sum \eta_{P,Q} \} \leq \nu, \; \forall \; k \geq 0$ and since $\max_{I,G} \{ \sum \eta_{P,Q} \} < 1 \; \forall \; k \geq 0$, we have $\nu < 1$. 

We further use Eq.~\ref{Eq:EstErrEqn-fc} to relate iterate at time $\{\Delta,k\}$ to $\{0,k\}$, followed by triangle inequality along with Eq.~\ref{Eq:LAB1} and definition of $\nu$, 
\begin{align}
    \max_{I , G} \|x^I_{0,k+1} - x^G_{0,k+1}\| &\leq \nu \max_{P,Q} \|x^P_{0,k} - x^Q_{0,k} + \alpha_k (D^Q_{\Delta,k}-E^Q_{\Delta,k}-D^P_{\Delta,k}+E^P_{\Delta,k})\| \\
    &\leq \nu\left( \max_{P,Q} \{ \|x^P_{0,k} - x^Q_{0,k}\| + \alpha_k \left( \|D^Q_{\Delta,k}\| + \|E^Q_{\Delta,k}\|+ \|D^P_{\Delta,k}\|+ \|E^P_{\Delta,k}\|\right) \} \right) \nonumber \\ 
    &\leq \nu\left( \max_{P,Q} \left( \|x^P_{0,k} - x^Q_{0,k}\| + 4 \alpha_k \bar{M} \SB{L} \right) \right) \nonumber \\
    & \leq \nu \max_{P,Q} \left( \|x^P_{0,k} - x^Q_{0,k}\| \right) + 4 \alpha_k \nu\bar{M} \SB{L}.
\end{align}

Now we perform an unrolling operation, and relate the maximum disagreement between servers to the initial disagreement between agents (at step $\{0,0\}$).
\begin{align}
\max_{I , G} \|x^I_{0,k+1} - x^G_{0,k+1}\| &\leq  \left( \nu \max_{P,Q} \left( \|x^P_{0,k} - x^Q_{0,k}\| \right) + 4 \alpha_k \nu \bar{M} \SB{L} \right) \nonumber \\
&\leq \left( \nu \left( \nu \max_{P,Q} \left( \|x^P_{0,k-1} - x^Q_{0,k-1}\| \right) + 4 \alpha_{k-1} \nu \bar{M} \SB{L} \right) + 4 \alpha_k \nu \bar{M} \SB{L} \right) \nonumber \\
&\leq \qquad \ldots \nonumber \\ 
&\leq \left(\nu^{k+1} \max_{P,Q} \left( \|x^P_{0,0} - x^Q_{0,0}\| \right) + 4 \bar{M}\SB{L} \sum_{i = 1}^k \left( \alpha_i \nu^{k-i+1} \right) \right). \label{Eq:ASYMPCONVPARVEC}
\end{align}

We start with the definition of $\delta^J_{k+1}$ (see  Eq.~\ref{Eq:deltaDef2}, \ref{Eq:deltaDef3}) and consider the maximum deviation ($\delta^J_{k+1}$) over all servers,
\begin{align}
\max_J\{\|\delta^J_{k+1}\|\} &= \max_J\{ \|x^J_{0,{k+1}} - \bar{x}_{0,{k+1}}\| \} = \max_J\{ \|x^J_{0,{k+1}} - \frac{1}{S} \sum_{I=1}^S x^I_{0,{k+1}}\| \} \nonumber \\
&= \max_J\{ \|\frac{1}{S} \sum_{I \neq J}^S (x^J_{0,{k+1}} - x^I_{0,{k+1}})\| \} \leq \frac{S-1}{S} \max_{I,G} \|x^I_{0,{k+1}} - x^G_{0,{k+1}}\|. \label{Eq:ScramRes1}
\end{align}
Together with Eq.~\ref{Eq:ASYMPCONVPARVEC}, we arrive at the desired expression from the statement of lemma, 
\begin{align}
\max_J\{\|\delta^J_{k+1}\|\} \leq \frac{S-1}{S} \left(\nu^{k+1} \max_{P,Q} \left( \|x^P_{0,0} - x^Q_{0,0}\| \right) + 4 \bar{M}\SB{L} \sum_{i = 1}^k \left( \alpha_i \nu^{k-i+1} \right) \right).
\end{align} 
$\hfill \blacksquare $
\end{proof}

Note that we can do away with Assumption~\ref{Asmp:Scrambling}, and prove similar bound on the maximum disagreement between server iterates and its average for any connected graph. For simplicity, we assume that the transition matrix is scrambling.

\begin{claim} [Server Consensus] \label{Cl:Consensus}
The server parameter vectors achieve consensus asymptotically. $$\lim_{k \rightarrow \infty} \max_{I,G} \|x^I_{0,k+1} - x^G_{0,k+1}\| = 0.$$
\end{claim}
\begin{proof}
We know from Eq.~\ref{Eq:ASYMPCONVPARVEC} that the maximum disagreement between any two servers ($I$ and $G$) is given by,
\begin{align}
    \max_{I,G} \|x^I_{0,k+1} - x^G_{0,k+1}\| \leq \left( \nu^{k+1} \max_{P,Q} \left( \|x^P_{0,0} - x^Q_{0,0}\| \right) + 4 \bar{M}\SB{L} \sum_{i = 1}^k \left( \alpha_i \nu^{k-i+1} \right) \right). \label{Eq:SerConTemp1} 
\end{align}

The first term on the right hand side of above expression decreases to zero as $k \rightarrow \infty$, since $\nu < 1$ and $\nu^{k+1} \rightarrow 0$ as $k \rightarrow 0$. 

We now show that the second term tends to zero too. Let us consider $\epsilon_0 > 0$. Since $\nu < 1$, we can define $0 < \epsilon < \frac{\epsilon_0}{4 \bar{M} \SB{L}} \frac{1 - \nu}{2\nu}$. Since $\alpha_k$ is non-increasing sequence, $\exists \ K = K(\epsilon) \in \mathbb{N}$ such that $\alpha_i < \epsilon$ for all $i \geq K$. Hence we can rewrite the second term for $k > K$ as,
$$ 4 \bar{M}\SB{L} \sum_{i = 1}^k  \left( \alpha_i \nu^{k-i+1} \right) = 4 \bar{M}\SB{L} \left[  \underbrace{\left( \alpha_1 \nu^k + \ldots + \alpha_{K-1} \nu^{k-K+2}\right)}_{A} + \underbrace{ \left(\alpha_K \nu^{k-K+1} + \ldots + \alpha_k \nu \right)}_{B}\right].$$

\noindent We can bound the individual terms A and B by using the monotonically non-increasing property of $\alpha_i$ and bounds on the sum of a geometric series.
\begin{align}
A &=  \alpha_1 \nu^k + \alpha_2 \nu^{k-1} + \ldots + \alpha_{K-1} \nu^{k-K+2} \nonumber \\
&\leq \alpha_1 (\nu^k + \nu^{k-1} + \ldots + \nu^{k-K+2}) && \ldots \alpha_1 \geq \alpha_i, \ \forall \ i \geq 1 \nonumber \\
&\leq \alpha_1 \nu^{k-K+2} \left(\frac{1 - \nu^{K-1}}{1-\nu} \right) \leq \frac{\alpha_1 \nu^{k-K+2}}{1-\nu}. &&  \ldots \nu < 1 \implies (1 - \nu^{K-1}) \leq 1\label{Eq:ABOUNDS1}\\
B &= \alpha_K \nu^{k-K+1} + \ldots + \alpha_k \nu \nonumber \\
& \leq \epsilon \nu \left( \frac{1-\nu^{k-K+1}}{1-\nu}\right) \leq \frac{\epsilon \nu}{1-\nu}.   \label{Eq:BBOUNDS1} && \ldots \nu < 1 \implies (1 - \nu^{k-K+1}) \leq 1
\end{align}

\noindent Since the right side of inequality in Eq.~\ref{Eq:ABOUNDS1} is monotonically decreasing in $k$ ($\nu < 1$), we can conclude, $\exists K_{0_1} > K$, such that, $A < \frac{\epsilon_0}{8 \bar{M} \SB{L}}$. Substituting the upper bound for $\epsilon$ in right side of inequality in Eq.~\ref{Eq:BBOUNDS1}, we get $\exists K_{0_2} > K$ such that $B < \frac{\epsilon_0}{8 \bar{M}\SB{L}}$. 

Using the individual bounds obtained above (on $A$ and $B$), we conclude, $\forall \; \epsilon_0 > 0$, $\exists \; K_0 = \max\{K_{0_1}, K_{0_2}\}$ such that $4\bar{M}\SB{L} \sum_{i=1}^k\left( \alpha_i \nu^{k-i+1} \right) < \epsilon_0$, $\; \forall \; k > K_0$. Clearly (from the $\epsilon-\delta$ definition of limit), we get, $\lim_{k \rightarrow \infty} 4 \bar{M} \SB{L} \sum_{i=1}^k \left( \alpha_i \nu^{k-i+1} \right) = 0$. This limit together with the limit of first term on the right side of Eq.~\ref{Eq:SerConTemp1} being zero, implies, $\lim_{k \rightarrow \infty} \max_{I,G} \|x^I_{0,k+1} - x^G_{0,k+1}\| = 0$. Thus we have asymptotic consensus of the server parameter vectors.

$\hfill \blacksquare$
%Now $\lim_{k \rightarrow \infty} A = 0$ and $\lim_{k \rightarrow \infty} B = \epsilon \frac{\nu}{1-\nu}$, the second term can be upper bounded by,
%$$\lim_{k \rightarrow \infty} 4 \bar{M} \SB{L} \sum_{i = 1}^k  \left( \alpha_i \nu^{k-i+1} \right) \leq 4 \bar{M} \SB{L} \epsilon \left(\frac{\nu}{1-\nu}\right)$$
%We can hence say that the second term can arbitrarily decrease as $\epsilon$ decreases. As both terms decrease arbitrarily close to zero, we can show that the maximum disagreement between any two communicating servers decreases, and the server-parameters enter the consensus manifold. 
\end{proof}

\begin{theorem} \label{Th:ConvMain}
Let Assumptions~\ref{Asmp:Function}, \ref{Asmp:Set}, \ref{Asmp:SubBound}, \ref{Asmp:QConn}, \ref{Asmp:SerConn}, and \ref{Asmp:GradLip} hold with $\mathcal{X}^*$ being a nonempty bounded set. Also assume a diminishing step size rule presented in Eq.~\ref{Eq:LearnStepCond}. Then, for a sequence of iterates $\{x_{0,k}\}$ generated by an interleaved consensus projected gradient descent algorithm following symmetric learning condition and bounded update condition, the  iterate average converge to an optimum in $\mathcal{X}^*$.
\end{theorem}
\begin{proof}
We intend to prove convergence using deterministic version of Lemma~\ref{Lem:RobSiegConv}. We begin by using the relation between iterates given in Lemma~\ref{Lem:IterateConvAnyW} with $y = x^* \in \mathcal{X}^*$,
\begin{align}
\eta_{k+1}^2 &\leq \left(1 + \gamma_k \right) \eta_{k}^2 - 2\alpha_k M \left(f(\bar{x}_{0,k}) - f(y)\right) + {\alpha_k^2}(C_1^2 + C_2^2 + C_3^2) + \alpha_k C_4 \max_{J}\{\|\delta^J_k\|\}. \label{Eq:PROOF0}
%&\| \bar{x}_{0,k+1} - x^*\|^2 \leq \left(1 +\underbrace{\frac{\alpha_k^2}{S^2}F + \frac{\alpha_k}{S} \bar{M} \SB{N} \max_{J}\{\|\delta^J_k\|\}}_{q_k} \right) \| \bar{x}_{0,k} - x^* \|^2 - \underbrace{\frac{2\alpha_k M}{S} (f(\bar{x}_{0,k}) - f(x^*))}_{v_k} \nonumber \\ 
%& + \underbrace{\frac{\alpha_k^2}{S^2} \left( 2 \bar{M}^2 \SB{L}^2 +  \sum_{i=1}^{ \Delta}\left(\sum_{J=1}^S \sum_{hx=1}^C W_{i-1,k}[J,h] \  L_h \right)^2  + F + 8 \alpha_k \bar{M}^3 \SB{L}^2 \SB{N} \right) + \left(\frac{2 \alpha_k^2 \bar{M}^2 \SB{L} \SB{N}}{S^2} +\frac{\alpha_k}{S} \bar{M} \SB{N} \right) \max_{J}\{\|\delta^J_k\|\}}_{w_k} \nonumber % $+ \max_{J}\{\|\delta^J_k\|\}$
\end{align}
We check if the above inequality satisfies the conditions in Lemma~\ref{Lem:RobSiegConv} viz. $\sum_{k=0}^\infty \gamma_k < \infty$ and $\sum_{k=0}^\infty w_k < \infty$, where $\gamma_k = (1/S) (2 \alpha_k^2 \bar{M}^2 \SB{N} \SB{L} + \alpha_k \bar{M}\SB{N} \max_J \|\delta^J_{k}\|)$ and $w_k = {\alpha_k^2}(C_1^2 + C_2^2 + C_3^2) + \alpha_k C_4 \max_{J}\{\|\delta^J_k\|\}$. \\

We first show that $\sum_{k = 0}^\infty \alpha_k \max_J \|\delta^J_k\| < \infty$.
\begin{align}
    &\sum_{k = 0}^\infty \alpha_k \max_J \|\delta^J_k\| \leq \frac{S-1}{S} \sum_{k = 0}^\infty \alpha_k  \left( \nu^{k+1} \max_{P,Q} \left( \|x^P_{0,0} - x^Q_{0,0}\| \right) + 4 \bar{M}\SB{L} \sum_{i = 1}^k \left( \alpha_i \nu^{k-i+1} \right) \right) \nonumber \\
    &\leq \frac{S-1}{S} \left( \sum_{k = 0}^\infty \alpha_k  \nu^{k+1} \max_{P,Q} \left( \|x^P_{0,0} - x^Q_{0,0}\| \right) + \sum_{k = 0}^\infty 4 \bar{M}\SB{L} \alpha_k \sum_{i = 1}^k \left( \alpha_i \nu^{k-i+1} \right) \right) \nonumber \\
    &\leq \frac{S-1}{S} \left( \max_{P,Q} \left( \|x^P_{0,0} - x^Q_{0,0}\| \right) \sum_{k = 0}^\infty \alpha_k  \nu^{k+1}  + \sum_{k = 0}^\infty 4 \bar{M}\SB{L} \sum_{i = 1}^k \left( \alpha_i^2 \nu^{k-i+1} \right) \right) \quad \ldots \alpha_k \leq \alpha_i, \forall i \leq k \nonumber \\
    &\leq \frac{S-1}{S} \left( \max_{P,Q} \left( \|x^P_{0,0} - x^Q_{0,0}\| \right) \sum_{k = 0}^\infty \alpha_k  \nu^{k+1}  + 4 \bar{M}\SB{L} \sum_{k = 0}^\infty \sum_{i = 1}^k \left( \alpha_i^2 \nu^{k-i+1} \right) \right). \label{Eq:PROOF1}
\end{align}
In the above expression, we can show that the first term is convergent by using the ratio test. We observe that, $$ \limsup_{k \rightarrow \infty} \frac{\alpha_{k+1} \nu^{k+2}}{\alpha_{k} \nu^{k+1}} = \limsup_{k \rightarrow \infty} \frac{\alpha_{k+1} \nu}{\alpha_k} < 1 \implies \sum_{k=0}^\infty \alpha_k \nu^k < \infty,$$
since, $\alpha_{k+1} \leq \alpha_k$ and $\nu < 1$. Arriving at the second term involves using the non-increasing property of $\alpha_i$, i.e. $\alpha_k \leq \alpha_i \forall i \leq k $. Now, we use Lemma~\ref{Lem:Ram}, with $\zeta_j = \alpha_j^2$ where $\sum_{j=1}^\infty \zeta_j < \infty$, and show that the second term in the above expression is finite i.e., $4 \bar{M}\SB{L} \sum_{k = 0}^\infty \sum_{i = 1}^k \left( \alpha_i^2 \nu^{k-i+1} \right) < \infty.$ 
Together using finiteness of both parts on the right side of inequality in Eq.~\ref{Eq:PROOF1} we have proved, 
\begin{align}
\sum_{k = 0}^\infty \alpha_k \max_J \|\delta^J_k\| < \infty.
\label{Eq:PROOF2}
\end{align}

We now begin to prove the finiteness of sum of $\gamma_k$ sequence, i.e. $\sum_{k=0}^\infty \gamma_k < \infty$. 
\begin{align}
\sum_{k=0}^\infty \gamma_k &= \sum_{k=0}^\infty \left( \frac{1}{S} (2 \alpha_k^2 \bar{M}^2 \SB{N} \SB{L} + \alpha_k \bar{M}\SB{N} \max_J \|\delta^J_{k}\|) \right) \nonumber \\
%&\leq \sum_{k=0}^\infty \left( \alpha_k^2 \frac{\bar{M}^2 \SB{N} \SB{L}}{S} + \frac{\alpha_k}{S} \bar{M} \SB{N} \left( \prod_{i=1}^{k-1} (1-\beta_i) \|\tilde{\delta}_0\| + \left( 4 \bar{M} \SB{L} \left(1 + \frac{1}{S}\right) \right) \sum_{i = 1}^{k-1} (1 - \bar{\beta})^{k-1-i} \alpha_i \right) \right) \nonumber \\
&=\frac{2 \bar{M}^2 \SB{N} \SB{L}}{S} \sum_{k=0}^\infty  \alpha_k^2  +  \bar{M}\SB{N} \sum_{k=0}^\infty \alpha_k \max_J \|\delta^J_{k}\|
%\left(\frac{S-1}{S}\right) \left( \max_{P,Q} \left( \|x^P_{0,0} - x^Q_{0,0}\| \right) \sum_{k = 0}^\infty \alpha_k  \nu^k  + 4 \bar{M}\SB{L} \sum_{k = 0}^\infty \sum_{i = 1}^k \left( \alpha_i^2 \nu^{k-i+1} \right) \right) \nonumber 
\end{align}
Clearly the first term is finite (due to the assumptions on step size, $\sum_{k=1}^\infty \alpha_k^2 < \infty$). The second term is finite by the analysis given above (see Eq.~\ref{Eq:PROOF2}). We have hereby proved, $\sum_{k=0}^\infty \gamma_k < \infty$. We can similarly prove $\sum_{k=0}^\infty w_k < \infty$.
\begin{align}
    \sum_{k=0}^\infty w_k  = (C_1^2 +C_2^2+C_3^2)\sum_{k=0}^\infty \alpha_k^2 + C_4 \sum_{k=0}^\infty \alpha_k \max_J \|\delta^J_k\|.
\end{align}
The first term is finite due the assumption on the step size and the second term is proved in analysis above (see Eq.~\ref{Eq:PROOF2}).

We can now use the deterministic version of Lemma~\ref{Lem:RobSiegConv} to show convergence of the iterate average to the optimum. We know from the analysis above that $\sum_{k=0}^\infty \gamma_k < \infty$ and $\sum_{k=0}^\infty w_k < \infty$.  As a consequence of Lemma~\ref{Lem:RobSiegConv} we get that the sequence $\eta_k^2$ converges to some point and $\sum_{k=0}^\infty 2 \frac{\alpha_k M}{S} (f(\bar{x}_{0,k}) - f(x^*)) < \infty$.

We use $\sum_{k=0}^\infty 2 \frac{\alpha_k M}{S} (f(\bar{x}_{0,k}) - f(x^*)) < \infty$ to show the convergence of the iterate-average to the optimum. Since we know $\sum_{k=0}^\infty \alpha_k = \infty$, it follows directly that $\lim \inf_{k \rightarrow \infty} f(\bar{x}_{0,k}) = f(x^*)$. Due to the continuity of $f(x)$, we know that the sequence of iterate average must enter the optimal set $ \mathcal{X}^*$ (i.e. $\bar{x}_{0,k} \in \mathcal{X}^*$). Since $\mathcal{X}$ is bounded (compactness in $\mathbb{R}^D$) we know that there exists a iterate-average subsequence $\bar{x}_{0,k_l} \subseteq \bar{x}_{0,k}$ that converges to some $x^* \in \mathcal{X}^*$ (i.e. $\lim_{l \rightarrow \infty} \bar{x}_{0,k_l} = x^*$). 

We know from Claim~\ref{Cl:Consensus} that the servers agree to a parameter vector asymptotically (i.e. $x^J_{0,k} \rightarrow x^I_{0,k}, \ \forall I \neq J$ as $k \rightarrow \infty$). Hence, all servers agree to the iterate average. This along with the convergence of iterate-average to the optimal set gives us that all server iterates must enter the optimal set $\mathcal{X}^*$ (i.e. $x^J_{0,k} \in \mathcal{X}^*, \ \forall J$). 

$\hfill \blacksquare $
\end{proof}

\begin{remark} \label{Rem:RMSErrorZero}
Claim~\ref{Cl:Consensus} states that the server parameter vectors asymptotically converge, implying $x^J_{0,k} = x^I_{0,k}$ for any $I, J$ as $k \rightarrow \infty$. It also directly follows that $\bar{x}_{0,k} = x^J_{0,k}$ for any $J$ as $k \rightarrow \infty$. From the definition of $\eta_k^2$ in Eq.~\ref{Eq:RMSErrorDef}, we get,  $\eta_k^2 = \sum_{J = 1}^S \|x^J_{0,k} - x^*\|^2 = S \|x^J_{0,k} - x^*\|^2 = S \|\bar{x}_{0,k} - x^*\|^2$ as $k \rightarrow \infty$ for some $x^* \in \mathcal{X}^*$. Theorem~\ref{Th:ConvMain} states that $\bar{x}_{0,k}$ enters the optimal set $\mathcal{X}^*$ as $k \rightarrow \infty$. From the above statements it is clear,  $\liminf_{k \rightarrow \infty} \eta_k^2 = 0$. 
\end{remark}

\begin{remark}
If the optimal set $\mathcal{X}^*$ is a singleton set, then we know from Theorem~\ref{Th:ConvMain}, Claim~\ref{Cl:Consensus} and Remark~\ref{Rem:RMSErrorZero} that all server states will converge to the only element of $\mathcal{X}^*$, and the RMS Error is driven to zero. If the optimal set has many points, then the server states should maintain the same distance from all optimal points (solutions). In a space of $D$ dimensions, $D+1$ linearly independent optimal points, if we find a server state that is at equal distance from all optimal points, the iterates reach a unique solution. 
\end{remark}

\subsection{Complete Server Graph} \label{Sec:4.1}
We can prove that under the added condition of the server graph being complete, the interleaved algorithm solves Problem~\ref{Prob:GlobalOpt}.
\begin{lemma} \label{Lem:IterateConvCSG}
Let Assumptions~\ref{Asmp:Function}, \ref{Asmp:Set} and \ref{Asmp:SubBound} hold and let the server network form a complete graph. Then a sequence \{$x_{0,k}$\} generated by an interleaved consensus and projected gradient descent algorithm following symmetric learning condition, satisfies, for all $y \in \mathcal{X}$,
\begin{equation}
\| x_{0,k+1} - y \|^2 \leq (1 + \frac{1}{S} \alpha_k^2 \bar{M} F) \| x_{0,k} - y \|^2 - 2  \frac{1}{S}  \alpha_k M (f(x_{0,k}) - f(y)) + \frac{1}{S} \alpha_k^2 C_0^2,
\end{equation}
where, \begin{align*}
&F = \left[  \SB{N}  + 4 C  \Delta \SB{L} \right] \quad \text{and} \\
&C_0^2 = \sum_{J=1}^S \sum_{i=1}^{ \Delta}\left(\sum_{h=1}^C W_{i-1,k}[J,h] \  L_h \right)^2 + 8  \Delta (\bar{M} \SB{L})^2 + 2 \alpha_k \left[ \bar{M}\SB{N}  + C  \Delta (4 \bar{M} \SB{L})^2\right] + \bar{M}F.
\end{align*}
\end{lemma}
\begin{proof}
Proof can be found in Appendix~\ref{sec:ApxB}
$\hfill \blacksquare$
\end{proof}

\begin{theorem} \label{Th:ConvCSG}
Let Assumptions~\ref{Asmp:Function}, \ref{Asmp:Set},  \ref{Asmp:SubBound} hold, and let $\mathcal{X}^*$ be a non-empty bounded set. Also consider a diminishing step size condition, i.e. $$\alpha_k > 0, \qquad \lim_{k\rightarrow \infty} \alpha_k = 0, \qquad \sum_{k = 0}^{\infty} \alpha_k = \infty. $$  Then for a sequence {$\{x_{0,k}\}$} generated by an interleaved consensus and projected gradient descent algorithm following symmetric learning condition and non-negative $W$ matrix condition, we have,
$$\lim_{k \rightarrow \infty} dist(x_{0,k}, \mathcal{X}^*) = 0, \qquad \lim_{k \rightarrow \infty} f(x_{0,k}) = f^*,$$
where $dist(x_{0,k}) = \| x_{0,k} - x^* \|^2$.
\end{theorem}
\begin{proof}
The proof is formally stated in Appendix~\ref{sec:ApxB}
$\hfill \blacksquare$
\end{proof}

\subsection{Non-Negative Weight Matrix with Complete Server Graph} \label{Sec:4.2}
We can prove that under the added conditions of non-negative weight matrix ($W_{i,k}$) and the server graph being complete, the interleaved algorithm solves Problem~\ref{Prob:GlobalOpt}.
\begin{lemma} \label{Lem:IterateConvNNWCSG}
Let Assumptions~\ref{Asmp:Function}, \ref{Asmp:Set} and \ref{Asmp:SubBound} hold and let the server network form a complete graph. Then a sequence \{$x_{0,k}$\} generated by an interleaved consensus and projected gradient descent algorithm following symmetric learning condition and non-negativity of $W$ matrix, satisfies, for all $y \in \mathcal{X}$,
\begin{equation}
\| x_{0,k+1} - y\|^2 \leq  \| x_{0,k} - y \|^2 - 2 \frac{M}{S} \alpha_k \left( f(x_{0,k}) - f(y)\right) + \frac{1}{S} \alpha_k^2 C_0^2, \label{Eq:Lemma1}
\end{equation}
where,
\begin{align*}
    C_0^2 = \sum_{J = 1}^S \sum_{i=1}^{ \Delta}\left( \left(\sum_{h=1}^C W_{i-1,k}[J,h] \  L_h \right)^2  + 2 \left(\sum_{h=1}^C W_{i-1,k}[J,h] L_h \left( \sum_{t=1}^{i-1}\sum_{l = 1}^{C} W_{t-1,k}[J,l] L_l \right) \right) \right).
\end{align*}
\end{lemma}
\begin{proof}
Proof can be found in Appendix~\ref{sec:ApxA}
$\hfill \blacksquare$
\end{proof}

\begin{theorem} \label{Th:ConvNNWCSG}
Let Assumptions~\ref{Asmp:Function}, \ref{Asmp:Set},  \ref{Asmp:SubBound} hold, and let $\mathcal{X}^*$ be a non-empty bounded set. Also consider a diminishing step size condition, i.e. $$\alpha_k > 0, \qquad \lim_{k\rightarrow \infty} \alpha_k = 0, \qquad \sum_{k = 0}^{\infty} \alpha_k = \infty. $$  Then for a sequence {$\{x_{0,k}\}$} generated by an interleaved consensus and projected gradient algorithm following symmetric learning condition and non-negative $W$ matrix condition, we have,
$$\lim_{k \rightarrow \infty} dist(x_{0,k}, \mathcal{X}^*) = 0, \qquad \lim_{k \rightarrow \infty} f(x_{0,k}) = f^*,$$
where $dist(x_{0,k}) = \| x_{0,k} - x^* \|^2$.
\end{theorem}
\begin{proof}
The proof is formally stated in Appendix~\ref{sec:ApxA}
$\hfill \blacksquare$
\end{proof}

%%%%%%%%%%%%%%%%%%%%%%%%%%%%%%%%%%%%%%%%%%%%%%%%%%%%%%%%%%%%%%%%%%%%%%%%%%%%%%%%%%%%%%%%%%%%%%%%%%%%%%%%%%%%%%%%%%%%%%%%%%%%%%%%%%%%%%%%%%%%%%%%%%%%%%%%%%%%%%%%%%%%%%%%%%%%%%%%%%%%%%%%%%%%%%%%%%%

\section{Privacy} \label{Sec:PrivacyDiscussion}
We claim that, because we multiply gradients by random weights (both positive and negative), no parameter server has accurate information about the exact state, gradient pair ($x, g_i(x)$), thus enhancing privacy. 

As a first take on privacy, let us consider a malicious entity (or an adversary, $Z$) that intends to extract  information about the individual objective functions ($f_i(x)$), based on information from gradients uploaded to the servers ($W_{i-1,k}[J,h] g_h(x)$) and the parameter vectors ($x$). One may choose to look at the malicious entity as a third party that has gained access to gradient and parameter information and now intends to know datasets/information stored with clients that induce the objective function. We can also consider multiple servers to form a malicious cohort, however, we will restrict to just one malicious server for the purpose of this discussion. The following information and constants are accessible to server $Z$, $\forall \; \{i,k\}$ -  
\begin{itemize}
    \item ($x^Z_{i-1,k}$,$W_{i-1,k}[Z,h] g_h(x^Z_{i-1,k})$) - (State, Randomized Gradient) pair ($\forall \; h = 1, 2, \ldots, C$)
    \item ($x^J_{\Delta,k}$) - Parameter State from neighbors at consensus step, ($\forall \; J \sim Z$)
    \item $\Delta$, $S$ and $C$.
\end{itemize}

Information about individual functions ($L_i$, $N_i$ etc.), the algorithm and other constants, (the strategy of selecting $W[J,h]$, $M$ and $\bar{M}$) is not known to the adversary. 

The random weights can be drawn from a uniform distribution $\mathcal{U}[0,\Gamma_h[k]]$ or $\mathcal{U}[-\Gamma_h[k],0]$, where $\Gamma_h[k]$ is private to client $h$ and changes in every cycle (dependent on $k$). It is not too difficult to construct an algorithm for generating random weights ($W$ matrix) that satisfies SLC and BUC. Correctness and optimality of the result obtained with this algorithm is proved by the analysis in Section~\ref{Sec:ConvergenceResults}. Such an algorithm will ensure that no adversary $Z$, can accurately estimate gradient values ($g_h(x)$) at state $x$ (and thereby function values ($f_h(x)$)) thereby enhancing privacy. 

The concept of weights (positive and negative), however, opens up a lot of avenues to achieve privacy in distributed optimization. Dividing gradients into smaller fractions and supplying these fractions to different servers guarantees that the model is accurate and also provides greater privacy (since the servers are oblivious to the true gradients). 
%\red{\sout{Further analysis, quantification and privacy guarantees will be explored in a future technical report.}} 

\subsection*{Privacy of Function Partitioning Approach}
Function partitioning approach is presented in our companion work \cite{gade16convsum}. Clients construct partitions of their own objective functions, and each of such partition is associated to one of the servers. This indicates that whenever a client sends an update to a specific server, it employs the function partition associated with that server as its objective function. Hence any server observes gradients from a specific partition of the client (and not the whole client objective function). The function splitting step is followed by a consensus - projected gradient step with the weights $W_{i,k}$ being fixed. Note that since the function-partitions can be added to form the original client function, all partitions together add up to $f(x)$. We can, using the analysis developed in this report, claim that this problem, with partitioned functions, can be correctly optimized by our algorithm (with constant weights).

As an example consider an optimization problem posed in client-server framework as shown in Figure~\ref{Fig:SV2P1}. The servers intend to optimize $f(x) = f_1(x)+f_2(x)+f_3(x)$ where, $f_1(x)$ corresponds to $C_1$, $f_2(x)$ corresponds to $C_2$, and $f_3(x)$ corresponds to $C_3$. In the function partitioning approach, each client $C_i$ splits its objective function into partitions $f_{i,J}(x)$ (corresponding to $C_{i,J}$), such that $f_i(x) = \sum_{J } f_{i,J}(x)$. Each fraction $C_{i,J}$ is used as client $C_i$'s objective function while sending updates to server $J$. This can be seen in Figure~\ref{Fig:SV2P2}.

\begin{figure}[t]
\begin{subfigure}{.43\textwidth}
  \centering
  \includegraphics[width=1.2\linewidth]{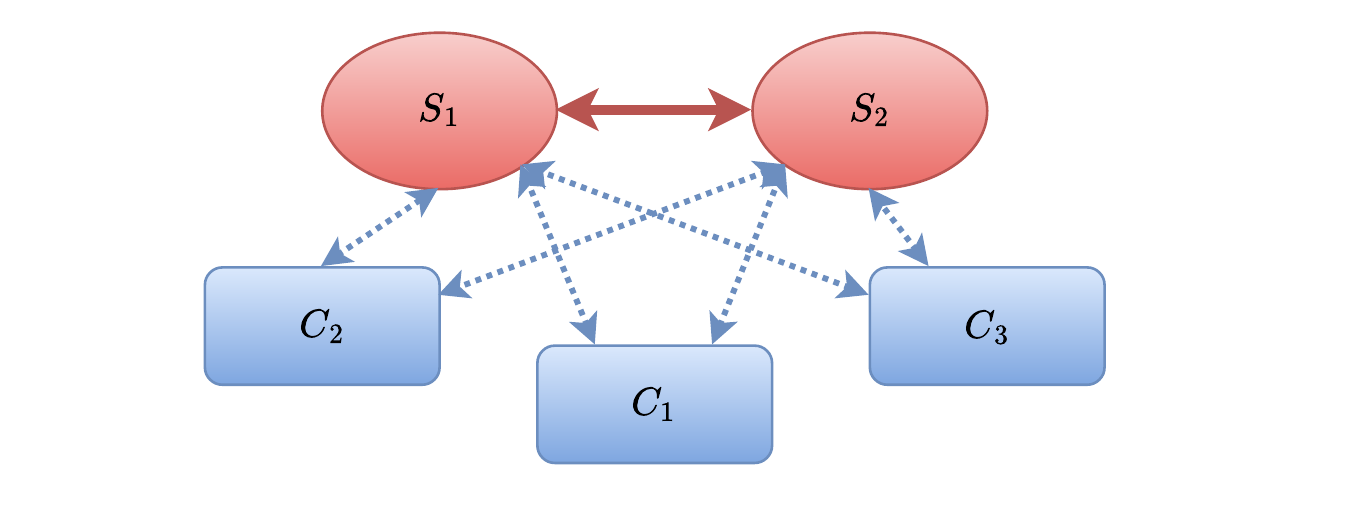}
  \caption{Distributed optimization problem posed in client-server architecture.}
  \label{Fig:SV2P1}
\end{subfigure} \hfill
\begin{subfigure}{.55\textwidth}
  \centering
  \includegraphics[width=.95\linewidth]{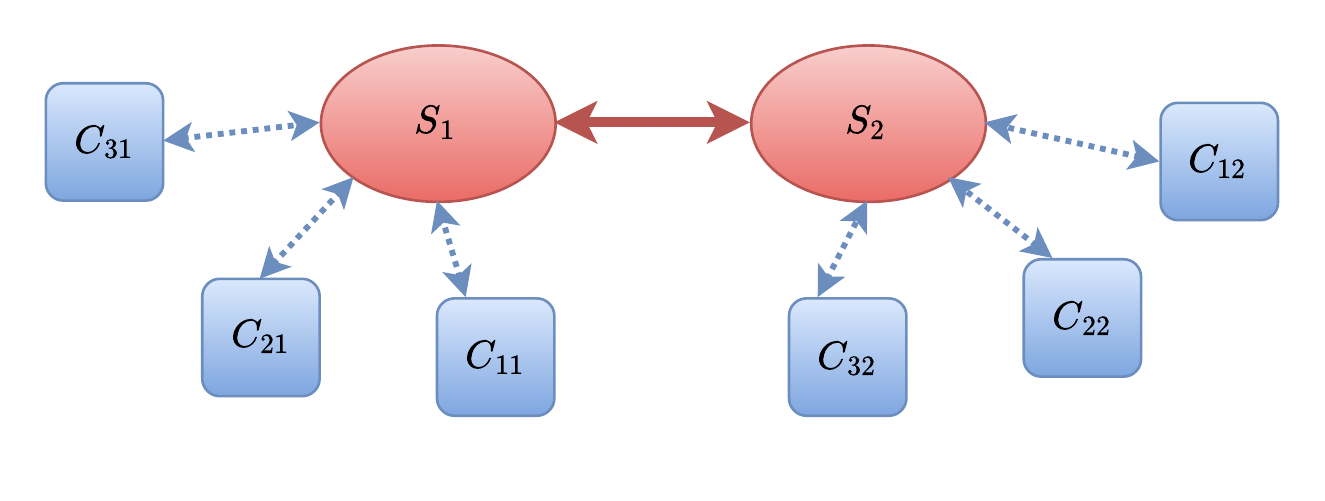}
  \caption{Evolution of maximum deviation of iterate from its average over time.}
  \label{Fig:SV2P2}
\end{subfigure} 
\caption{Function partitioning approach to privacy preserving optimization.}
\label{Fig:StatR-New}
\end{figure}

\begin{claim} \label{Cl:Priv1}
If every client sends updates to more than one server, and every server receives updates from more than one client, then no server can ever uncover the private objective function of any client.
\end{claim}
\begin{proof}
Servers form a connected component (Assumption~\ref{Asmp:SerConn}) and share states with other servers. A malicious server can observe the state evolution of all its neighbors. Using the state evolution and step-sizes, a malicious server may estimate the total gradient being used at neighboring server. This observed gradient is a cumulative effect of several function partitions connected to a server (greater than 1 by assumption). Hence, any malicious server may observe only sums of certain partitions.

A malicious agent can observe private objective function if and only if it can estimate gradients from all the partitions of that functions individually. This is, quite clearly, impossible since servers use aggregated gradient of all the partitions that are connected to them and one cannot estimate gradients from individual partitions based on aggregated gradients. Hence, no server can uncover private objective function of any client under given connectivity conditions of the partitions. $\hfill \blacksquare$
\end{proof}

\subsection*{Privacy of Interleaved Consensus and Descent Algorithm}

We can draw parallels between function splitting strategy and Algorithms~\ref{Algo:ASLearn}, \ref{Algo:ASLearn2}. We can consider every partition created by agents (function partition strategy) to be functions belonging to virtual agents (so we have as many virtual agents as we have function partitions). Recall that each partition (virtual agent) communicates with one server and hence, we set the weight for that server to be one and zero everywhere else. We have showed in Claim~\ref{Cl:Priv1} that this method makes individual objective functions private to adversarial servers. Constant weight matrix provides privacy of objective functions, hence we conjecture that random weights will be at least as private. This follows from the fact that with randomly weighted gradients (with coordinate-wise weights) observed by any server it will be difficult for any server to guess actual gradients, thereby, providing privacy to the functions (data) stored with clients.

Although use of random multiplicative weights provides privacy there could be inadvertent leakage of information, e.g. if $\nabla f_i(x) = 0$, then $W[J,i] \ \nabla f_i(x) = 0$ for any random weight $W[J,i]$. Hence, if a server receives zero gradient, then, it is clear that the true gradient $\nabla f_i(x)$ is zero. We can easily guarantee privacy even under this scenario by tweaking the uploaded gradient. If $\nabla f_i(x) = 0$, then, instead of uploading $W[J,i] \ \nabla f_i(x)$ (= 0), any client $i$ uploads $\tilde{g}^J_i$ to server $J$, such that the gradient uploads add up to $0$ (i.e. $\sum_{J} \tilde{g}^J_i = 0$) and $\|\tilde{g}^J_i\|$ is bounded. This condition is similar to SLC and ensures that the overall gradient from client $i$ to all servers is zero. In this scenario, servers cannot guess if the true gradient is zero, based solely on the observed gradient upload.

%%%%%%%%%%%%%%%%%%%%%%%%%%%%%%%%%%%%%%%%%%%%%%%%%%%%%%%%%%%%%%%%%%%%%%%%%%%%%%%%%%%%%%%%%%%%%%%%%%%%%%%%%%%%%%%%%%%%%%%%%%%%%%%%%%%%%%%%%%%%%%%%%%%%%%%%%%%%%%%%%%%%%%%%%%%%%%%%%%%%%%%%

\section{Numerics} \label{Sec:SimulationResults}
In this section we present the performance of iterative projected gradient descent and consensus algorithm on a distributed optimization problem. A static distributed optimization problem is presented in Section~\ref{Sec:Numerics-Prob}. Synchronous algorithm convergence and optimality results are also exhibited. A dynamic distributed optimization problem and its convergence results is presented in Section~\ref{Sec:Numerics-D-Prob}. Effect of different server-server topologies is presented in Section~\ref{Sec:Numerics-TopEffects}.

\begin{figure}[h]
    \centering
    \includegraphics[height=2.0in]{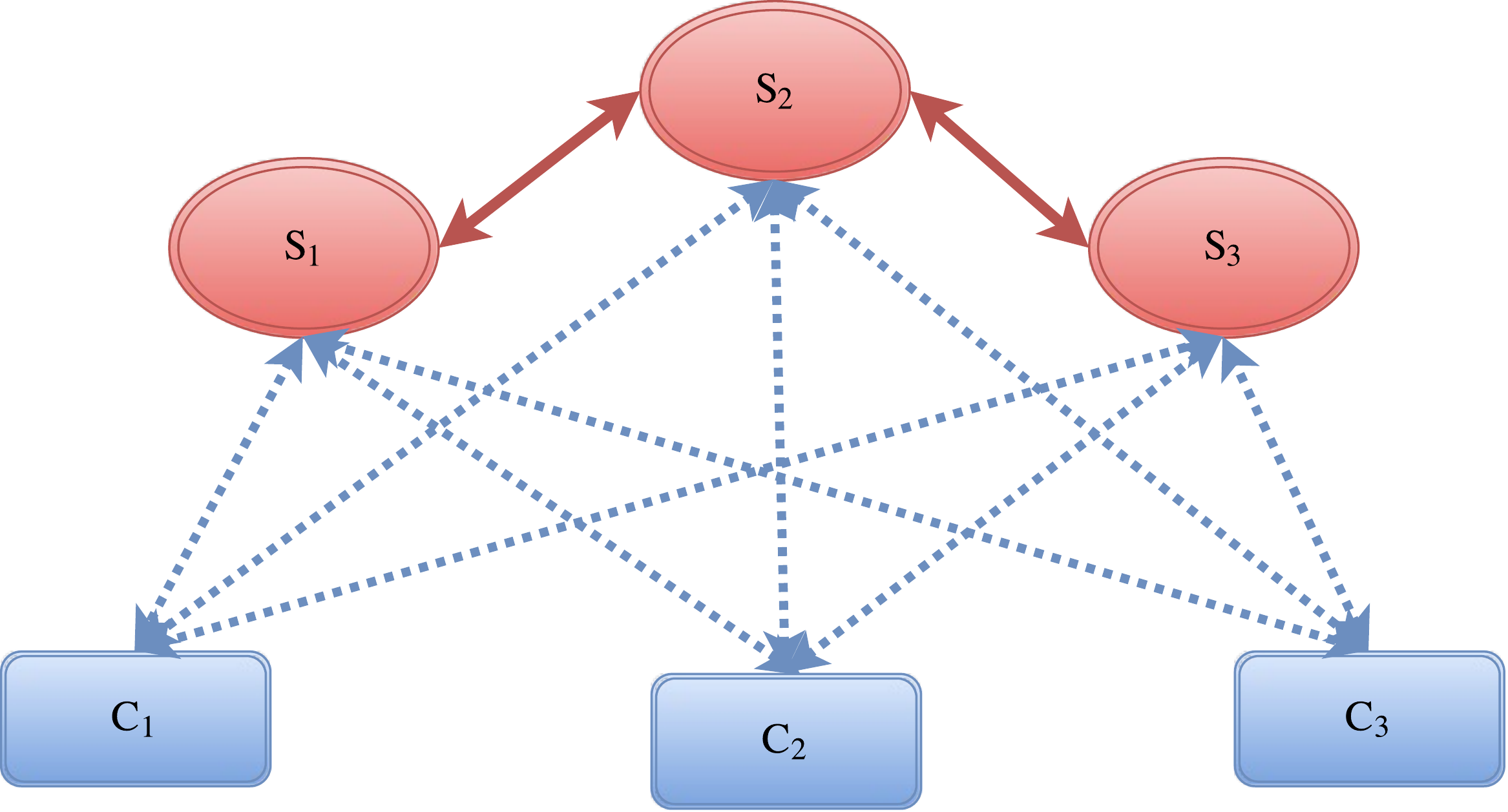}
    \caption{A schematic of distributed optimization problem: $S_i$ are Parameter Servers and $C_i$ are Clients. Red links represent the bidirectional flow of information between Servers and Blue links represents the downlink of states from Servers and uplink of gradients to Servers.}
    \label{Fig:Problem}
\end{figure}
\subsection{Distributed Optimization - A Static Case}  \label{Sec:Numerics-Prob}

In a static example we assume the server-server graph topology and server-client graph topology are both static. 

We consider a set of three Clients ($C_i$) each endowed with a scalar objective function ($f_i$) induced by their own dataset. The objective functions are given by, $f_1(x) = (x-1)^2$, $f_2(x) = (x-2)^2$ and $f_3(x) = (x-3)^2$. The scalar parameter belongs to the convex-compact decision set $\mathcal{X} = [-10, \; 10]$. We intend to find the minimizer of $f(x) = \sum_{i=1}^3 f_i (x)$, i.e.,
$$ \text{Find} \quad x^* \; = \; \underset{x \in \mathcal{X}}{\text{argmin}} \;  f(x). $$
Note that the functions have bounded gradients (over $\mathcal{X}$). $L_1 = 22$, $L_2 = 24$, $L_3 = 26$ and $\SB{L} = 72$. The gradients $g_h(x)$ for $h = 1, 2$ and $3$ are Lipschitz with constants being equal to 2, i.e. $N_1, N_2, N_3 = 2 \text{ and } \SB{N} = 6$. 

The Servers are connected in a static graph. Servers $S_1$ and $S_2$ are connected, and $S_2$ and $S_3$ are connected. The doubly stochastic weight matrix $B_k$ for this topology is a scrambling matrix. All clients are connected to all servers in a static graph (see Figure~\ref{Fig:Problem}). The weights are assumed to remain constant throughout the execution of the distributed optimization algorithm. The weight matrix is given by,
\begin{equation}
B_k = \begin{bmatrix}
0.8 & 0.2 & 0 \\
0.2 & 0.6 & 0.2 \\
0   & 0.2 & 0.8
\end{bmatrix}, \; \text{ and } \; W = \begin{bmatrix}
3 & -2 & -3 \\
-1 & 4 & -4 \\
-1 & -1 & 8
\end{bmatrix}.
\end{equation}
Note that $W$ satisfies the Symmetric Learning Condition (SLC) with $M = 1$ and the Bounded Update Condition (BUC) with $\bar{M} = 15$.

\begin{figure}[!b]
\begin{subfigure}{.5\textwidth}
  \centering
  \includegraphics[width=.95\linewidth]{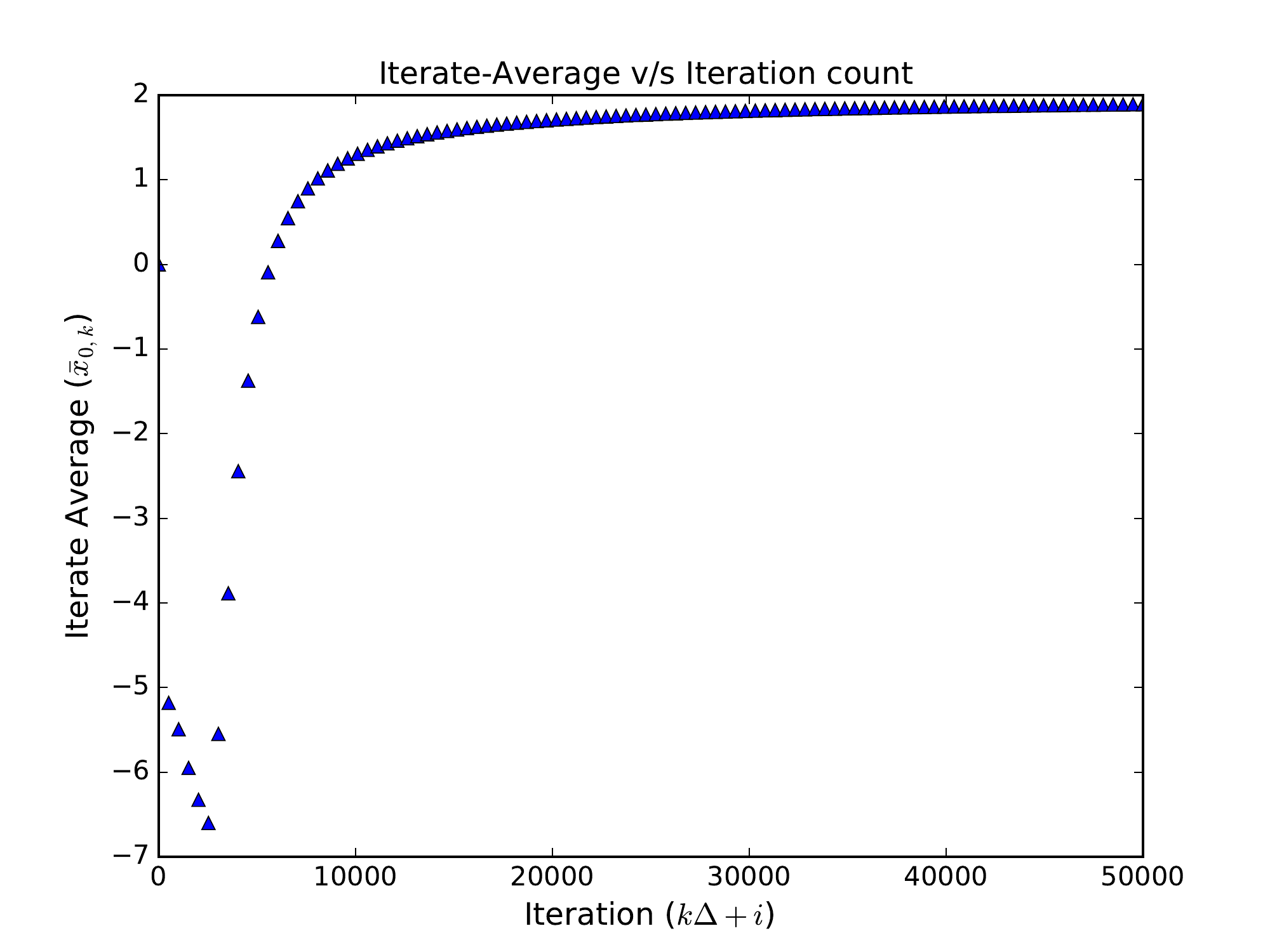}
  \caption{Iterate-average evolution over time.}
  \label{Fig:StatR-1}
\end{subfigure}%
\begin{subfigure}{.5\textwidth}
  \centering
  \includegraphics[width=.95\linewidth]{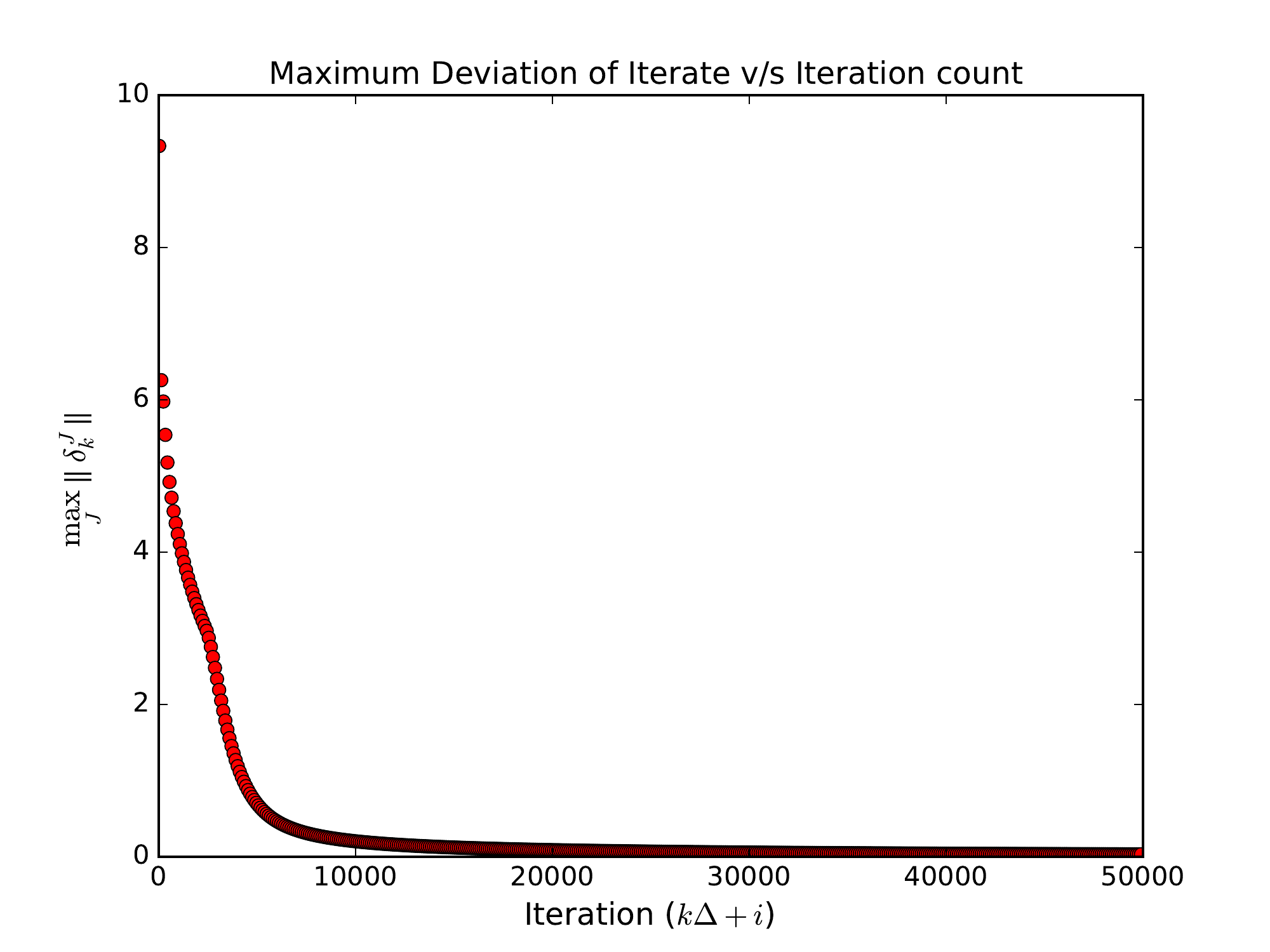}
  \caption{Evolution of maximum deviation of iterate from its average over time.}
  \label{Fig:StatR-2}
\end{subfigure} \\
\begin{subfigure}{.5\textwidth}
  \centering
  \includegraphics[width=.95\linewidth]{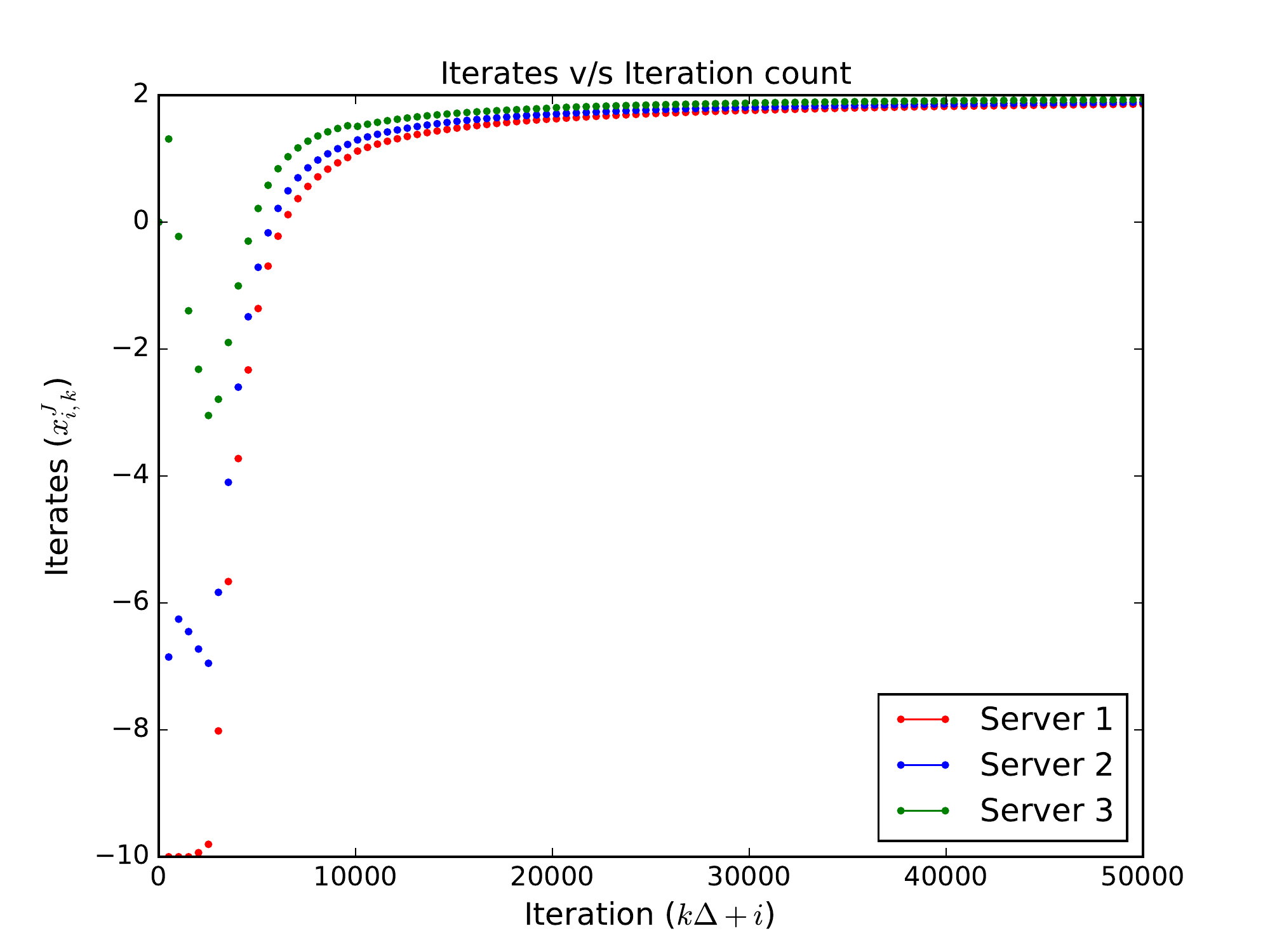}
  \caption{Evolution of iterates over time.}
  \label{Fig:StatR-3}
\end{subfigure}%
\begin{subfigure}{.5\textwidth}
  \centering
  \includegraphics[width=.95\linewidth]{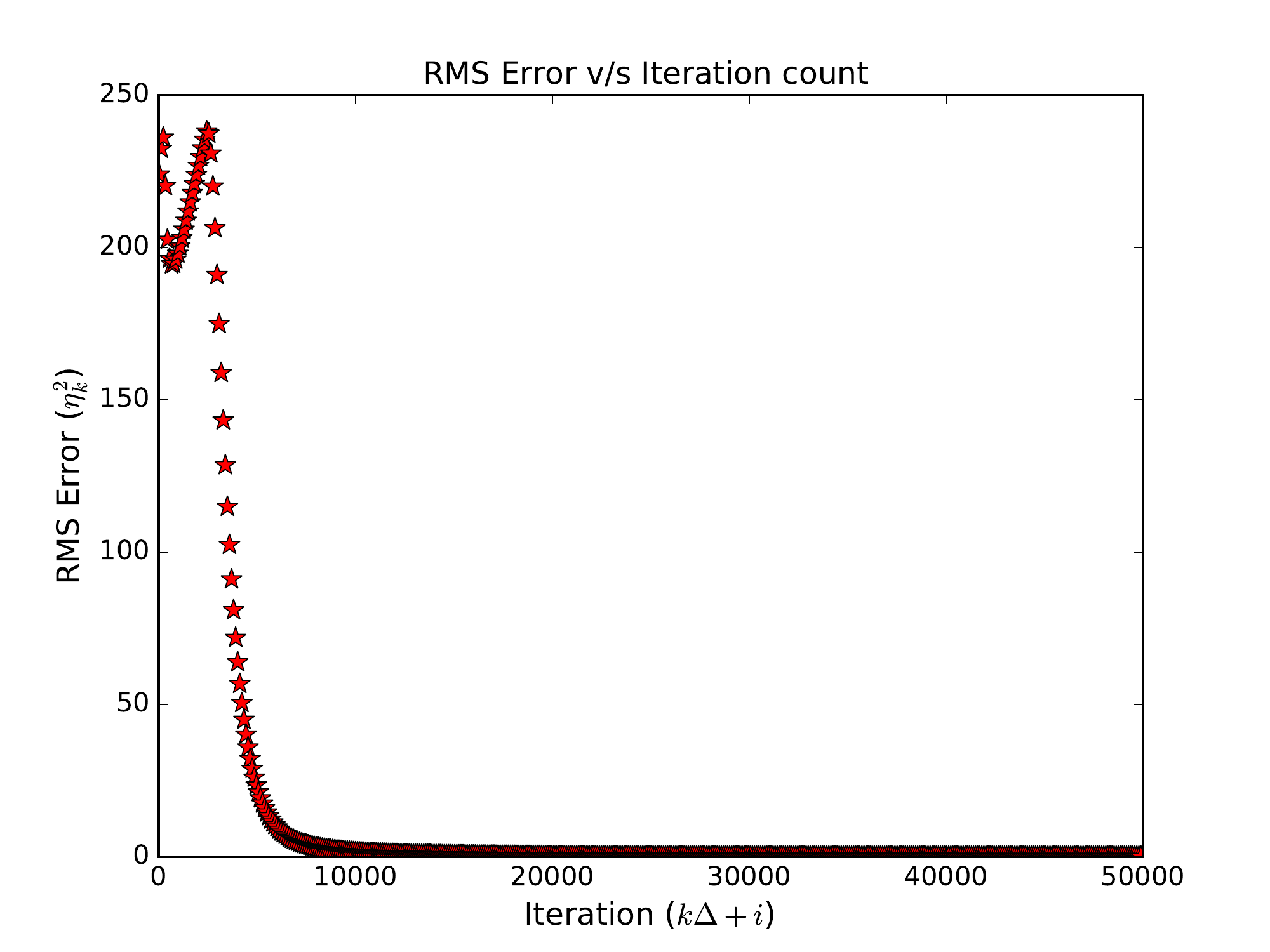}
  \caption{Evolution of RMS error $\eta_{i,k}^2$ over time.}
  \label{Fig:StatR-4}
\end{subfigure}
\caption{Static Problem}
\label{Fig:StatR}
\end{figure}

Interleaved Consensus and Projected Gradient Descent Algorithm described in Algorithm~\ref{Algo:ASLearn} and \ref{Algo:ASLearn2} is used to solve the distributed optimization problem described above. Consensus step is performed every five gradient descent steps ($\Delta = 5$). A decreasing learning rate is set $\alpha_k = 1/(k+0.0001)$. Note that this satisfies all the conditions on learning rate specified in Eq.~\ref{Eq:LearnStepCond}.

The server parameter ($x^J$) for all servers converges to $x^* = 2$. Figure~\ref{Fig:StatR-1} shows the convergence of iterate average to the optimum $x^* = 2$. The maximum deviation of iterate from its average over all servers is plotted with respect to time in Figure~\ref{Fig:StatR-2}. The plot shows that the server iterates converge to each other eventually as proved in Claim~\ref{Cl:Consensus}. The iterates ($x^J_{i,k}$) are plotted in Figure~\ref{Fig:StatR-3}. The figure shows convergence of iterates to the optimum. Figure~\ref{Fig:StatR-4} shows the decrease in RMS error ($\eta_k^2$). The RMS error decrease shows that \textit{all} servers move in the direction of the optimum. 

\begin{figure}[!b]
\begin{subfigure}{.5\textwidth}
  \centering
  \includegraphics[width=.95\linewidth]{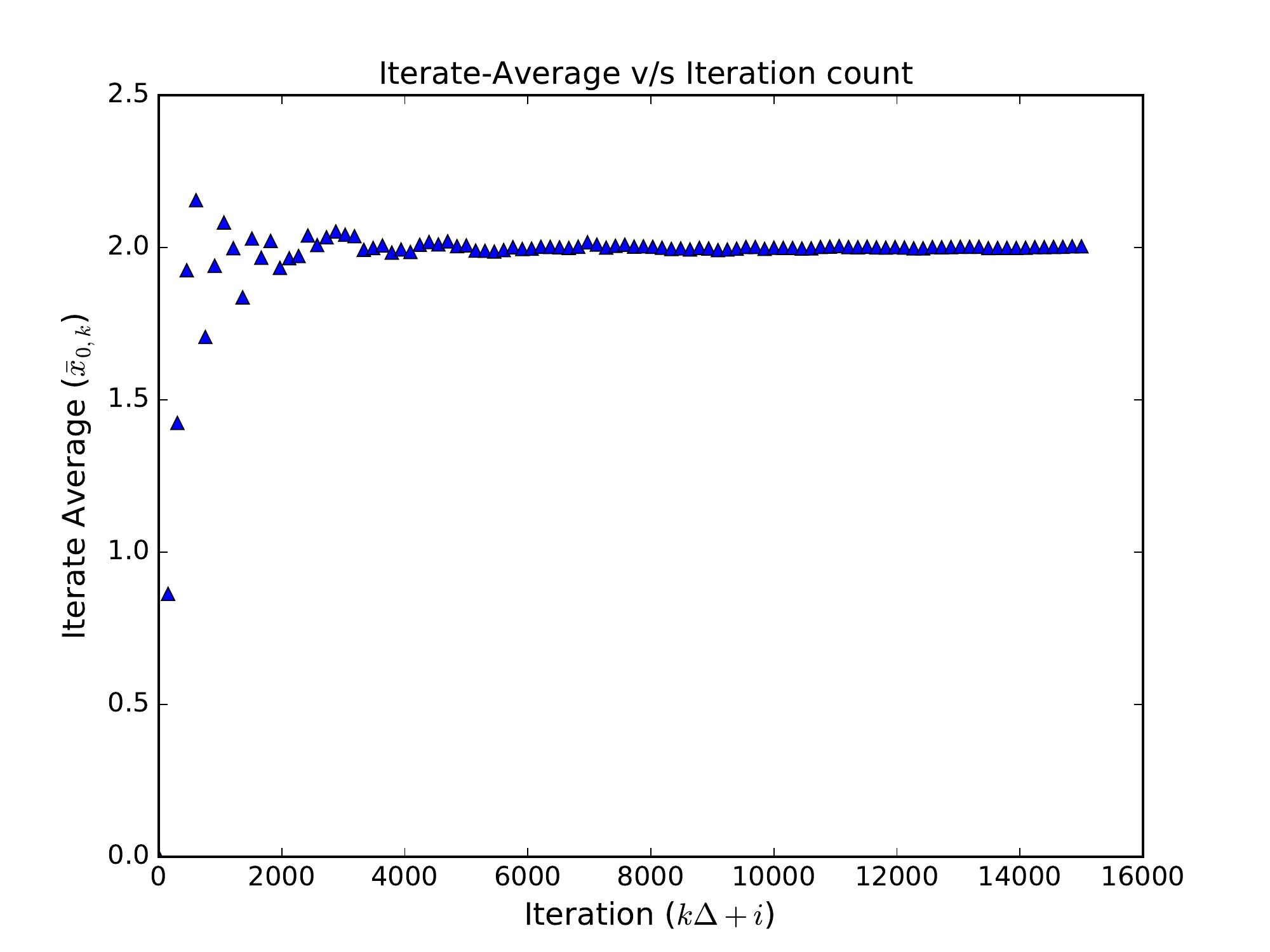}
  \caption{Iterate-average evolution over time.}
  \label{Fig:DynR-1}
\end{subfigure}%
\begin{subfigure}{.5\textwidth}
  \centering
  \includegraphics[width=.95\linewidth]{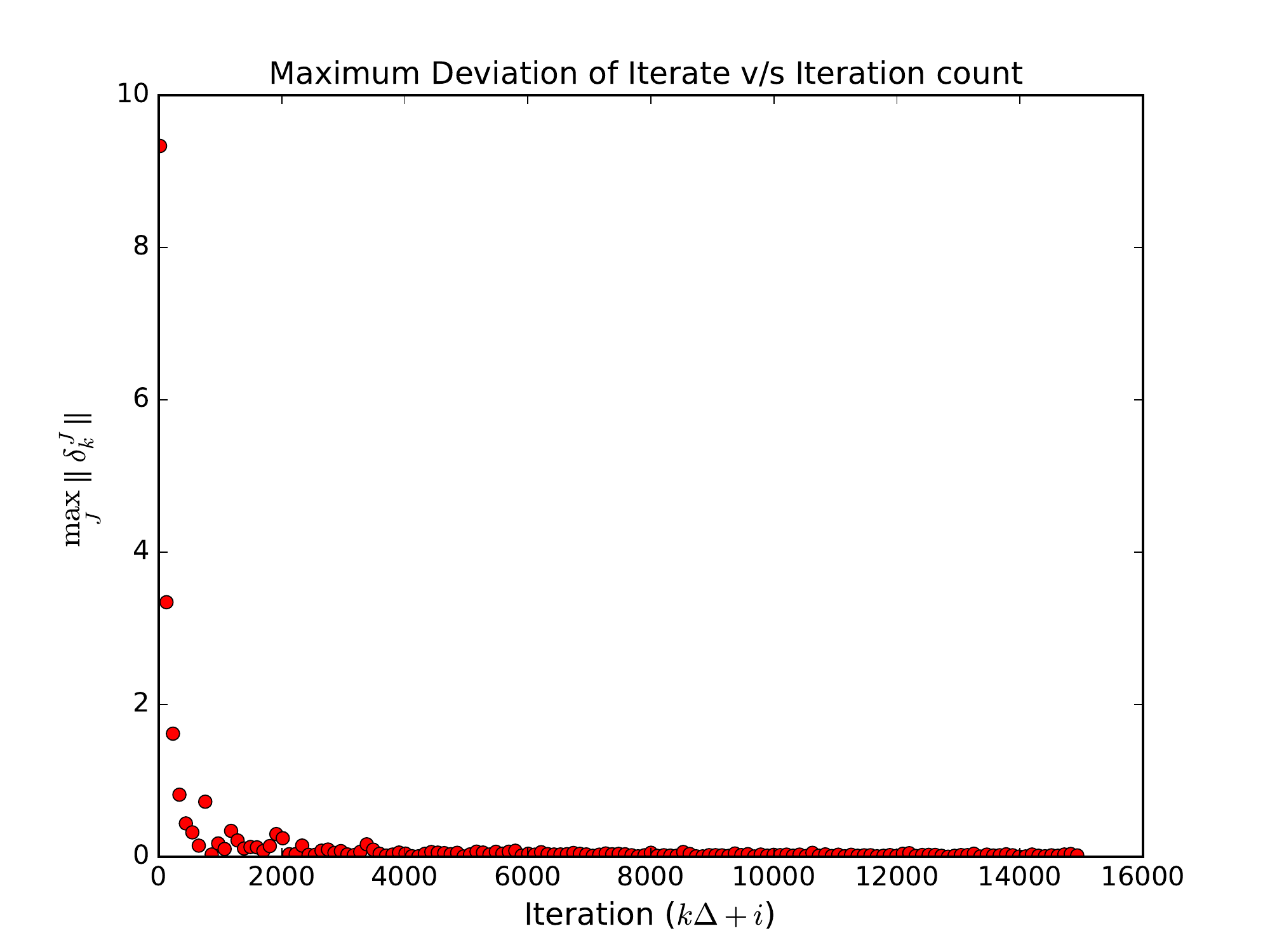}
  \caption{Evolution of maximum deviation of iterate from its average over time.}
  \label{Fig:DynR-2}
\end{subfigure} \\
\begin{subfigure}{.5\textwidth}
  \centering
  \includegraphics[width=.95\linewidth]{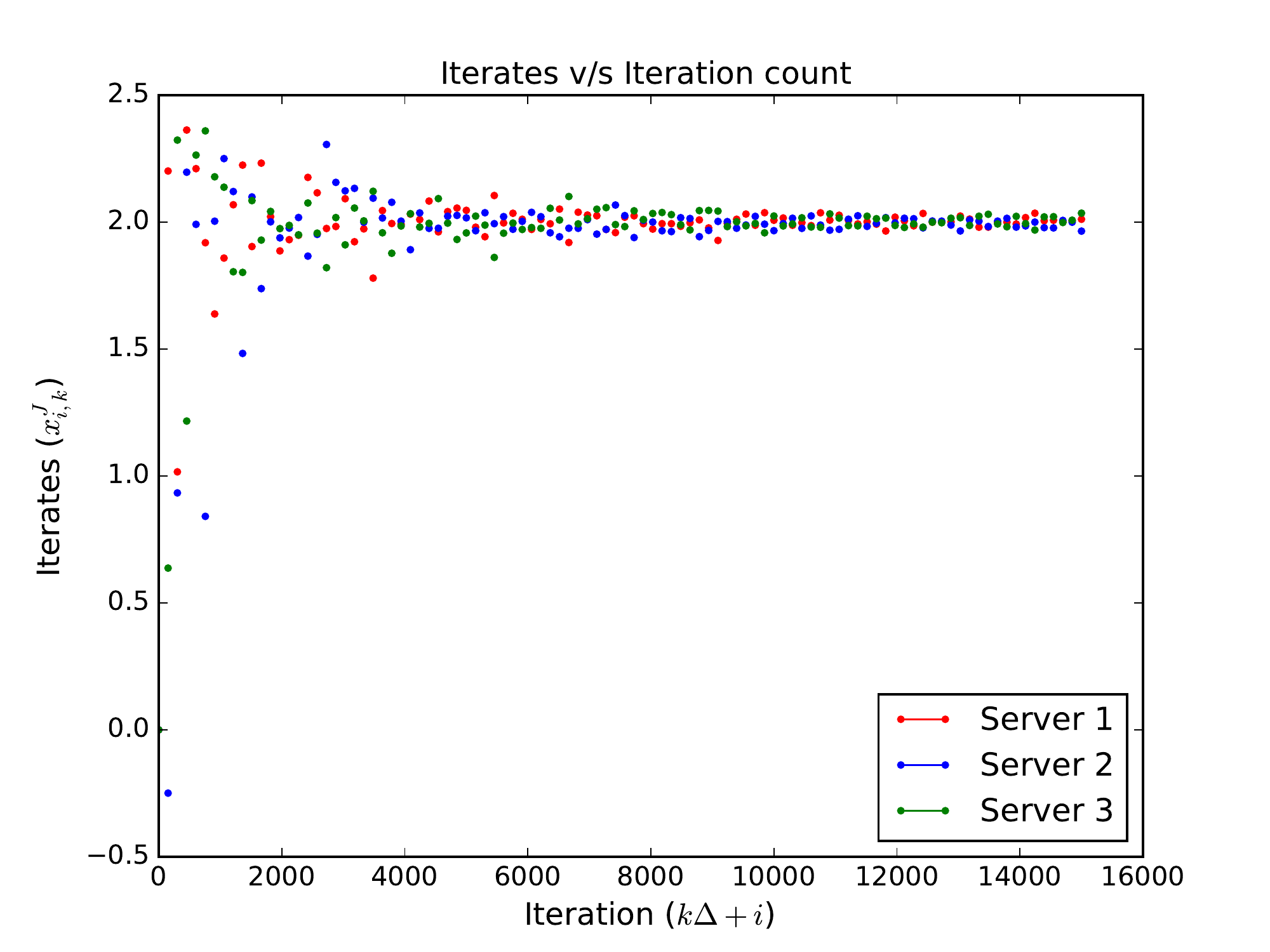}
  \caption{Evolution of iterates over time.}
  \label{Fig:DynR-3}
\end{subfigure}%
\begin{subfigure}{.5\textwidth}
  \centering
  \includegraphics[width=.95\linewidth]{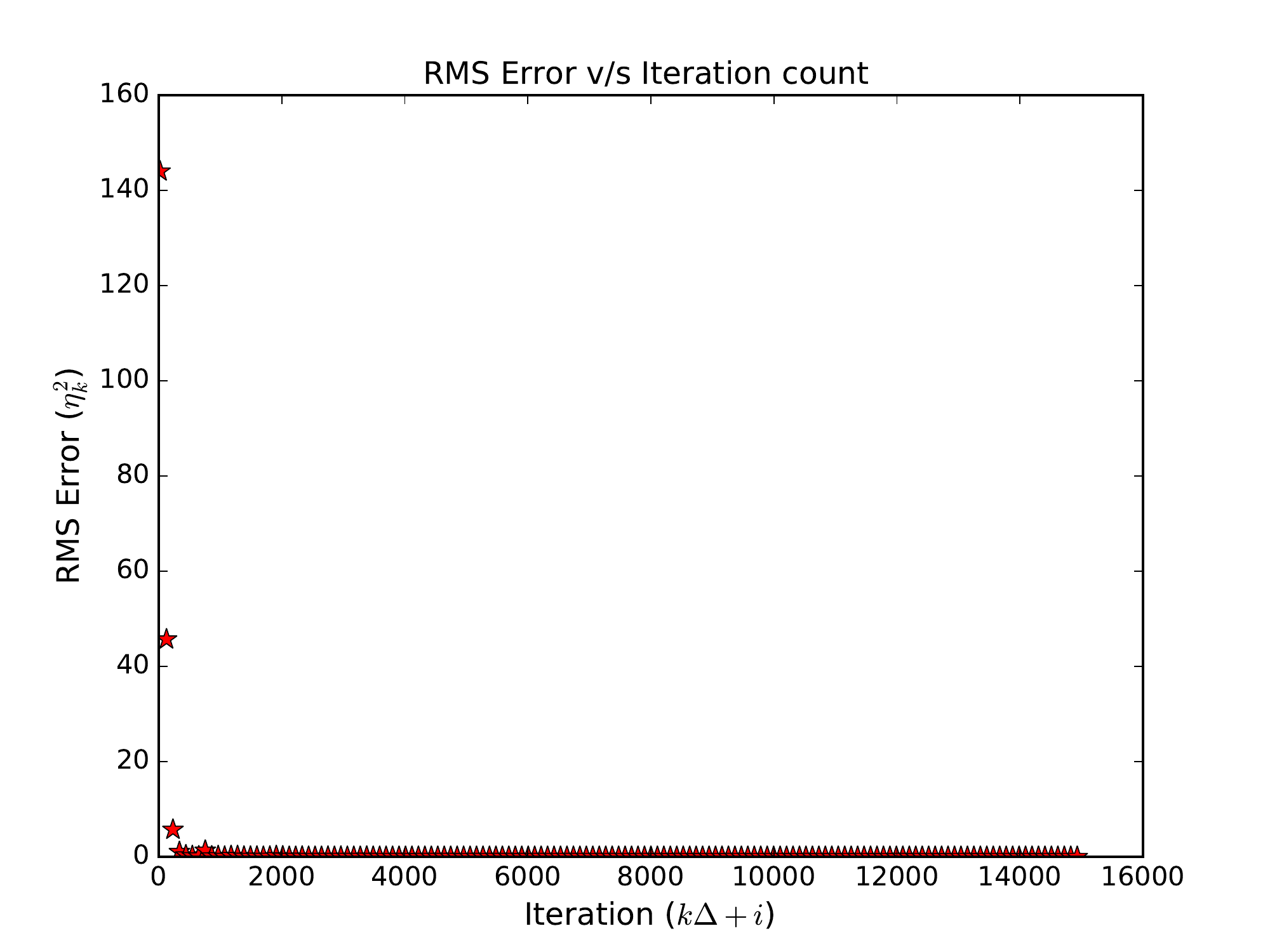}
  \caption{Evolution of RMS error $\eta_{i,k}^2$ over time.}
  \label{Fig:DynR-4}
\end{subfigure}
\caption{Dynamic Problem}
\label{Fig:DynR}
\end{figure}

\subsection{Distributed Optimization - A Dynamic Example}  \label{Sec:Numerics-D-Prob}
In the dynamic example we do not assyme anything about the server-client and server-server topologies apart from Assumptions~ \ref{Asmp:QConn} and \ref{Asmp:SerConn}.

We consider a set of three Clients ($C_i$) each endowed with a scalar objective function ($f_i$) induced by their own dataset. The objective functions are given by, $f_1(x) = (x-1)^2$, $f_2(x) = (x-2)^2$ and $f_3(x) = (x-3)^2$. The scalar parameter belongs to the convex-compact decision set $\mathcal{X} = [-10, \; 10]$. We intend to find the minimizer of $f(x) = \sum_{i=1}^3 f_i (x)$, i.e.,
$$ \text{Find} \quad x^* \; = \; \underset{x \in \mathcal{X}}{\text{argmin}} \;  f(x). $$
The gradient bounds and Lipschitz constants are same as the static case, i.e. $L_1 = 22$, $L_2 = 24$, $L_3 = 26$, $\SB{L} = 72$ and $N_1 = 2$, $N_2 = 2$, $N_3 = 2$, $\SB{N} = 6$. 

The Servers form a arbitrarily time-varying dynamic graph (connected). The doubly stochastic weight matrices $B_k$ for all possible connected topologies are given by,
\begin{equation}
{B_k}_1 = \begin{bmatrix}
0.8 & 0.2 & 0 \\
0.2 & 0.6 & 0.2 \\
0   & 0.2 & 0.8
\end{bmatrix}, \;
{B_k}_2 = \begin{bmatrix}
0.6 & 0.2 & 0.2 \\
0.2 & 0.8 & 0.0 \\
0.2 & 0.0 & 0.8
\end{bmatrix}, \text{ and }
{B_k}_3 = \begin{bmatrix}
0.8 & 0.0 & 0.2 \\
0.0 & 0.8 & 0.2 \\
0.2 & 0.2 & 0.6
\end{bmatrix}.
\end{equation}
Note that all three matrices are scrambling matrices and hence any choice of $B_k$ is acceptable.

All clients are connected to servers in a static graph. The weights are arbitrarily changed throughout the execution of the distributed optimization algorithm. The weight matrix $W$ is constructed to satisfy the Symmetric Learning Condition (SLC) with $M = 1$ and the Bounded Update Condition (BUC) with $\bar{M} = 15$. 

Interleaved Consensus and Projected Gradient Descent algorithm described in Algorithm~\ref{Algo:ASLearn} and \ref{Algo:ASLearn2} is used to solve the distributed optimization problem described above. Consensus step is performed every five gradient descent steps ($\Delta = 5$). A decreasing learning rate is set $\alpha_k = 1/(k+0.0001)$. Note that this satisfies all the conditions on learning rate specified in Eq.~\ref{Eq:LearnStepCond}. 

The server parameter ($x^J$) for all servers converges to $x^* = 2$. Figure~\ref{Fig:DynR-1} shows the convergence of iterate average to the optimum $x^* = 2$. Figures~\ref{Fig:DynR-2}, ~\ref{Fig:DynR-3} and ~\ref{Fig:DynR-4} depict the maximum deviation of iterate from the iterate-average, iterate convergence to each other and the decrease in RMS error ($\eta_k^2$) respectively. 

\begin{figure}[!b]
\begin{subfigure}{.5\textwidth}
  \centering
  \includegraphics[width=.65\linewidth]{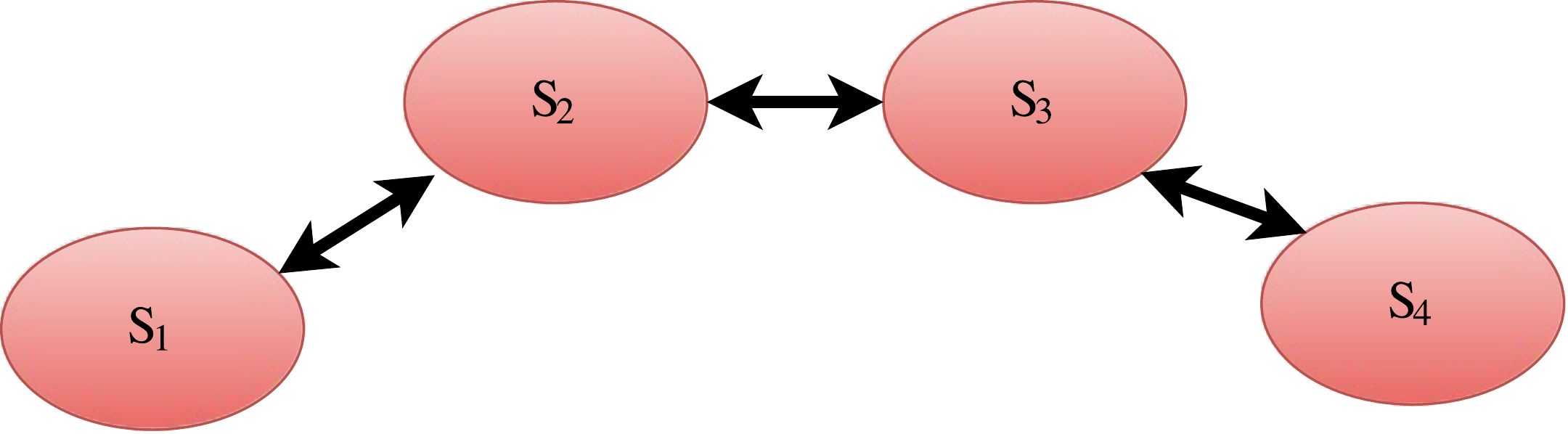}
  \caption{Line/Path Graph.}
  \label{Fig:T-1}
\end{subfigure}%
\begin{subfigure}{.5\textwidth}
  \centering
  \includegraphics[width=.45\linewidth]{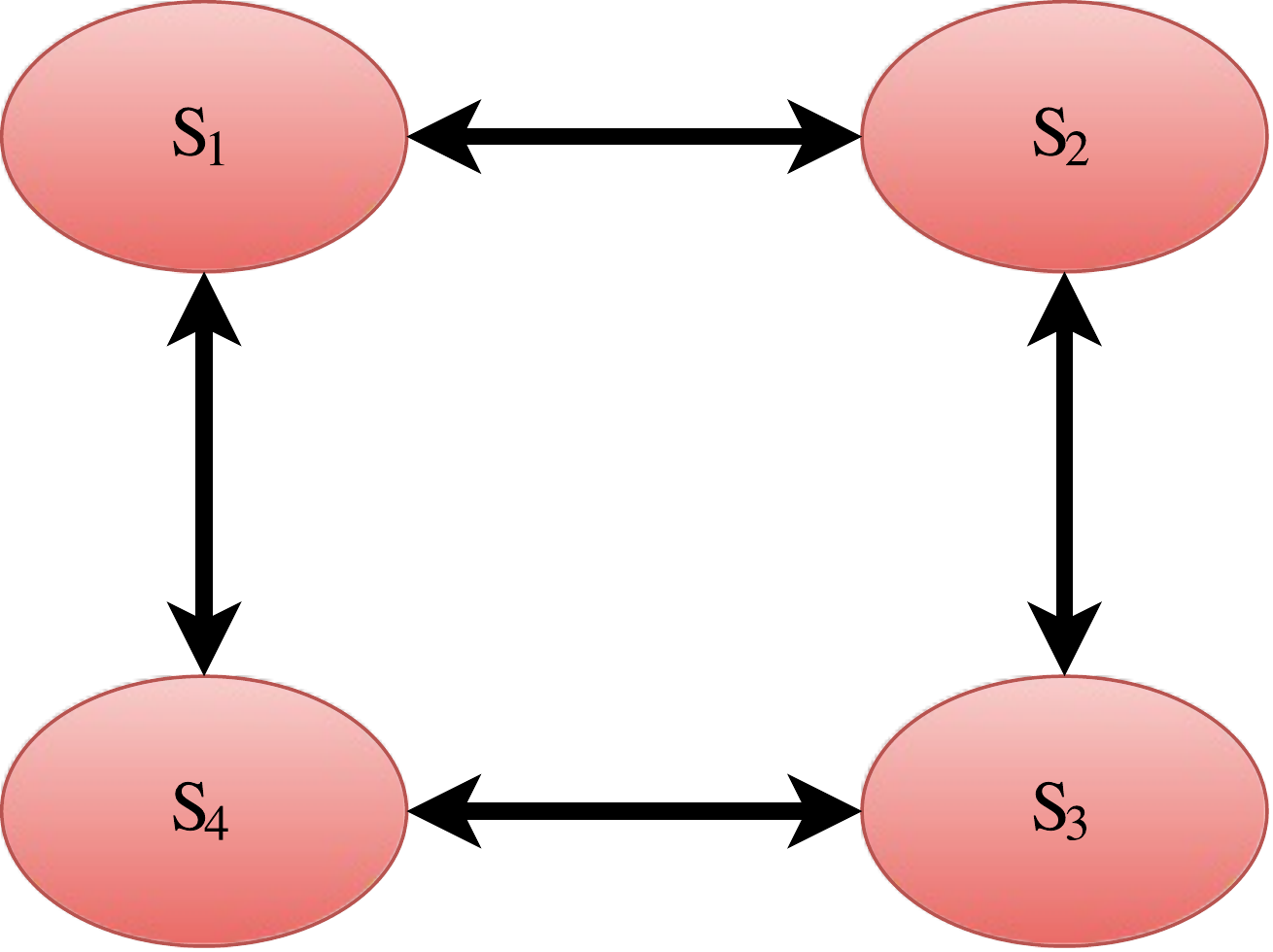}
  \caption{Cyclic Graph.}
  \label{Fig:T-2}
\end{subfigure} \\
\begin{subfigure}{.5\textwidth}
  \centering
  \includegraphics[width=.55\linewidth]{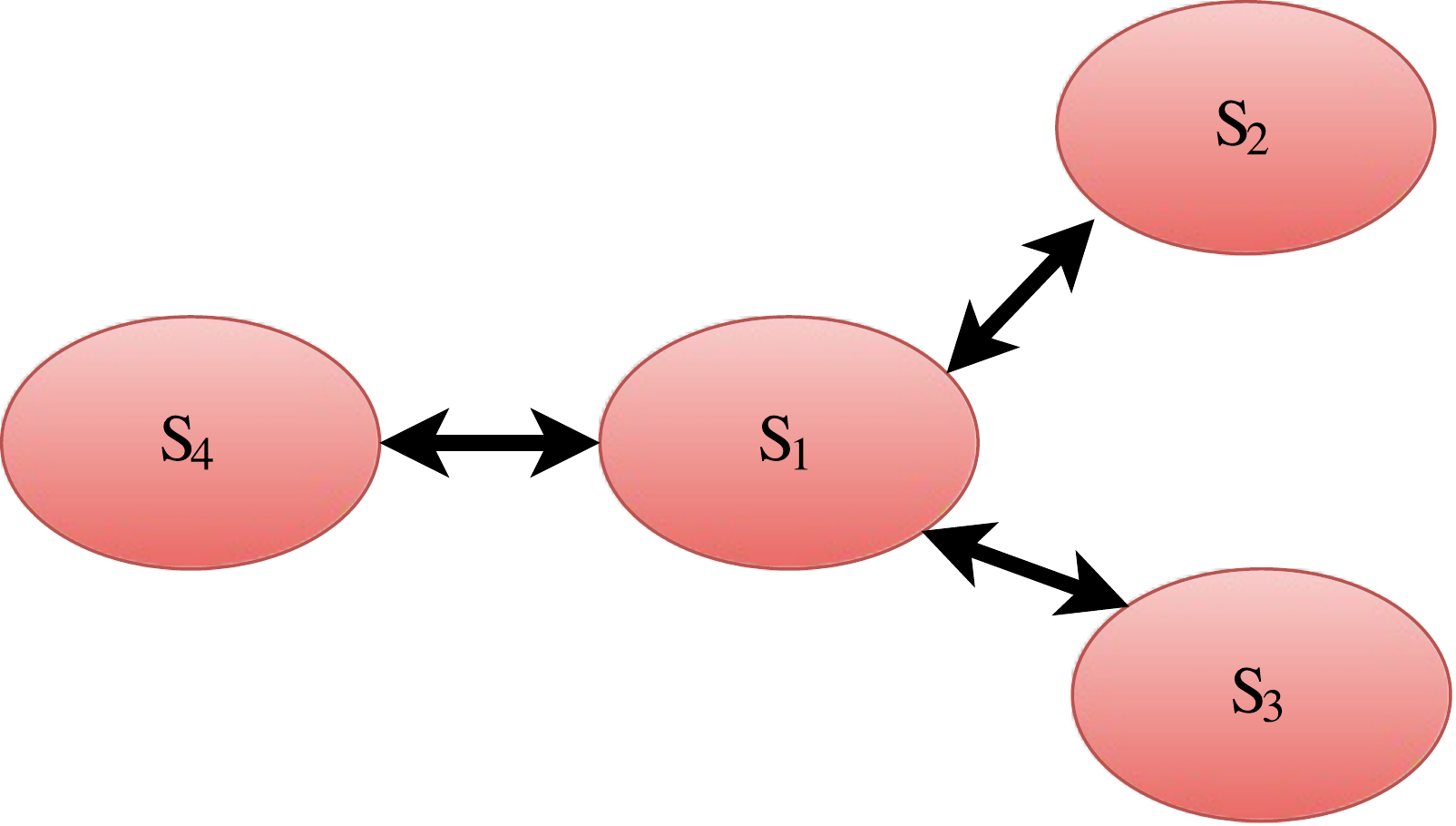}
  \caption{Star Graph.}
  \label{Fig:T-3}
\end{subfigure}%
\begin{subfigure}{.5\textwidth}
  \centering
  \includegraphics[width=.45\linewidth]{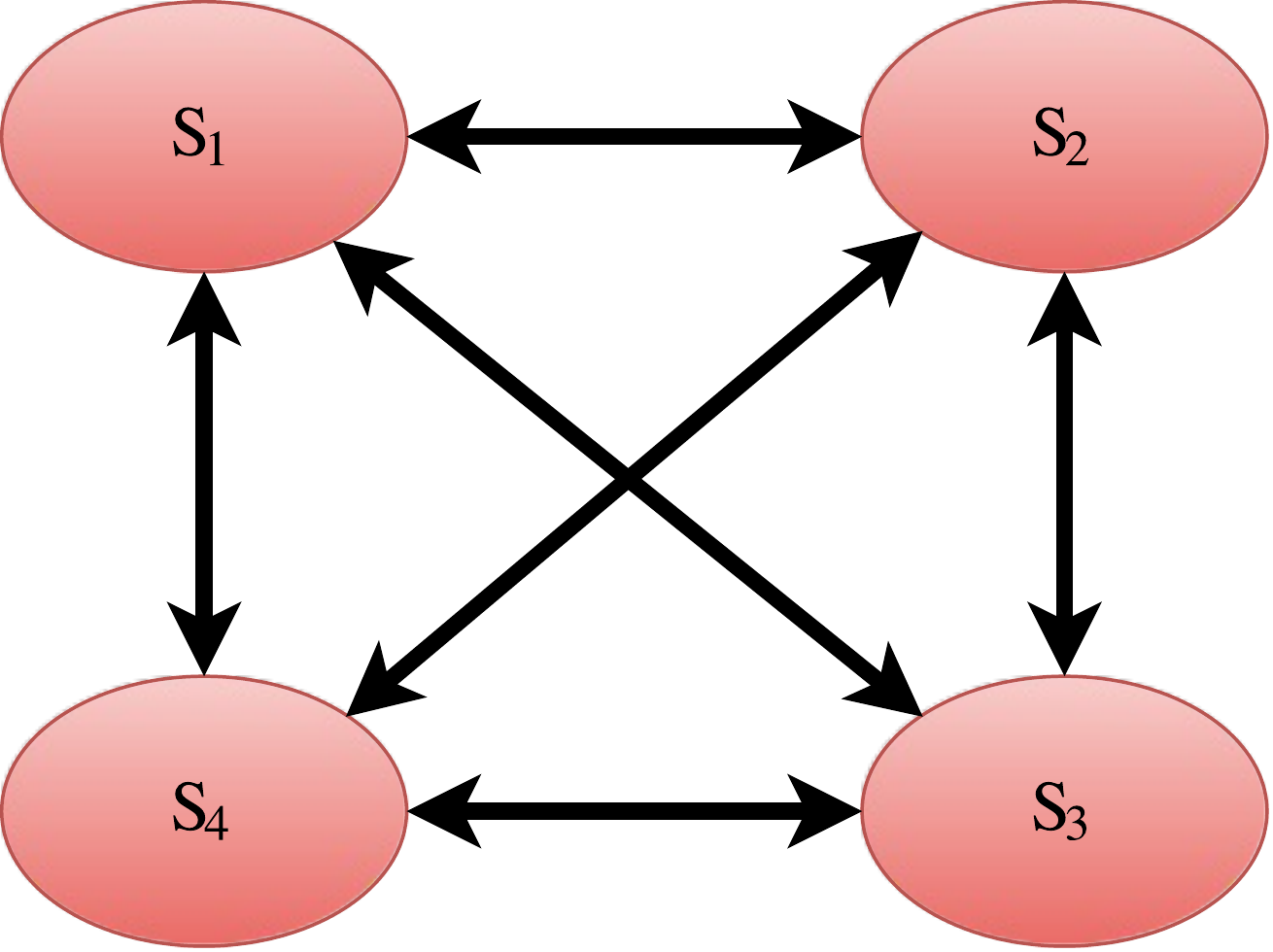}
  \caption{Complete Graph.}
  \label{Fig:T-4}
\end{subfigure}
\caption{Server-Server Graph Topologies}
\label{Fig:T}
\end{figure}

\subsection{Static and Dynamic Cases - Comparison}
The results are similar for the static and dynamic example presented in Sections~\ref{Sec:Numerics-Prob} and \ref{Sec:Numerics-D-Prob} respectively. Interleaved algorithm converges to the optimum in both cases. However, an interesting difference is that the iterates converge to the optimum much faster in the dynamic example. This can be seen clearly from the comparison of evolution of iterate-average (Figures~\ref{Fig:StatR-1} and \ref{Fig:DynR-1}, note the X-axes have different limits) and the RMS error plot (Figures~\ref{Fig:StatR-4} and \ref{Fig:DynR-4}, note the X-axes have different limits). 

The difference between static and dynamic solutions can be attributed to the fact that in the static case specific servers are forced in opposite direction (negative and positive weights). It takes several iterations ($\alpha_k$ has to be small enough) for the iterates to converge in response to the consensus step. However, in the dynamic case a server iterate may move closer or away from the optimum arbitrarily and hence it converges to the optimum relatively faster.

\subsection{Effect of Server Topologies} \label{Sec:Numerics-TopEffects}
We study the effect of different server topologies on the performance of Interleaved algorithm for a static distributed optimization problem. We consider four servers ($S_1$, $S_2$, $S_3$ and $S_4$) and three clients ($C_1$, $C_2$ and $C_3$). The objective functions for the three clients are same as seen in the static problem defined in Section~\ref{Sec:Numerics-Prob}. The optimum of $f(x)$ is $x^* = 2$. The weight matrix $W$ is given by,
\begin{equation*}
    W = \begin{bmatrix}
    4 & -1 & -2 \\
    -1 & 4 & -3 \\
    -1 & -1 & 8 \\
    -1 & -1 & -2
    \end{bmatrix}.
\end{equation*}
Note that this weight matrix satisfies SLC with $M = 1$ and BUC with $\bar{M} = 15$.

\begin{figure}[!b]
\begin{subfigure}{.48\textwidth}
  \centering
  \includegraphics[width=.95\linewidth]{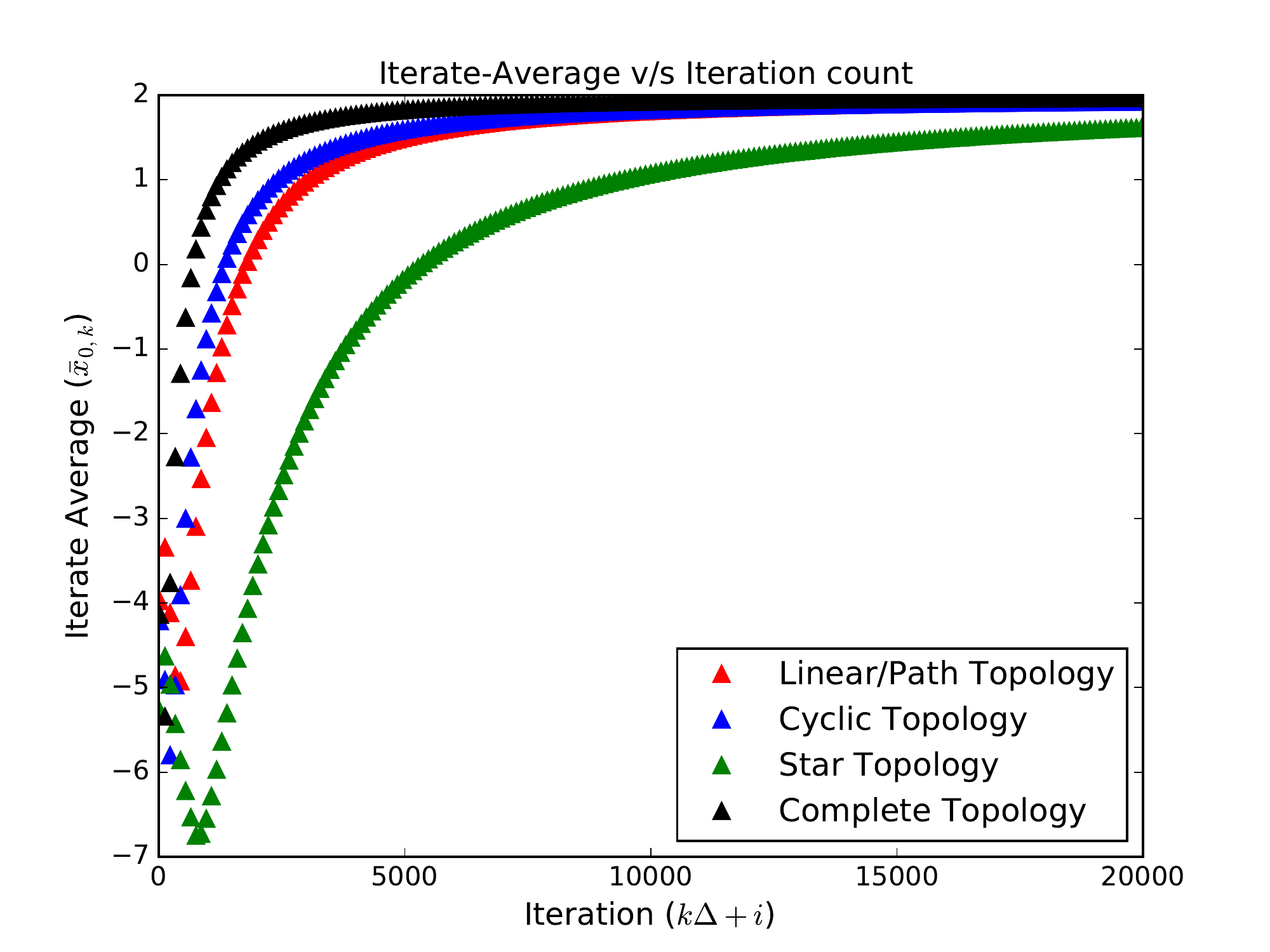}
  \caption{Iterate-average evolution over time with different server topologies.}
  \label{Fig:ST-1}
\end{subfigure} \hfill
\begin{subfigure}{.48\textwidth}
  \centering
  \includegraphics[width=.95\linewidth]{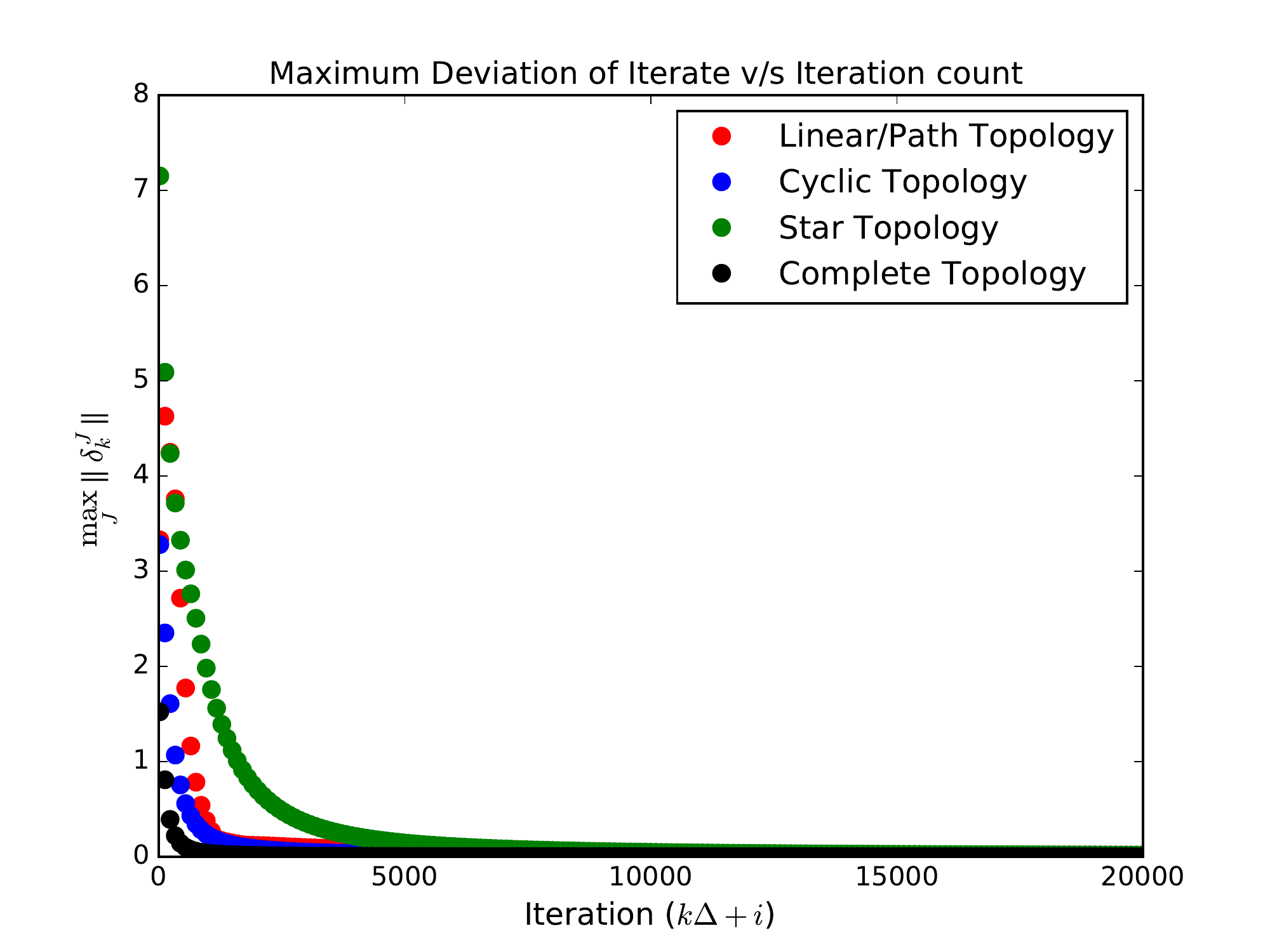}
  \caption{Evolution of maximum deviation of iterate from its average over time for different server topologies.}
  \label{Fig:ST-2}
\end{subfigure} \\
\centering
\begin{subfigure}{.52\textwidth}
  \centering
  \includegraphics[width=.9\linewidth]{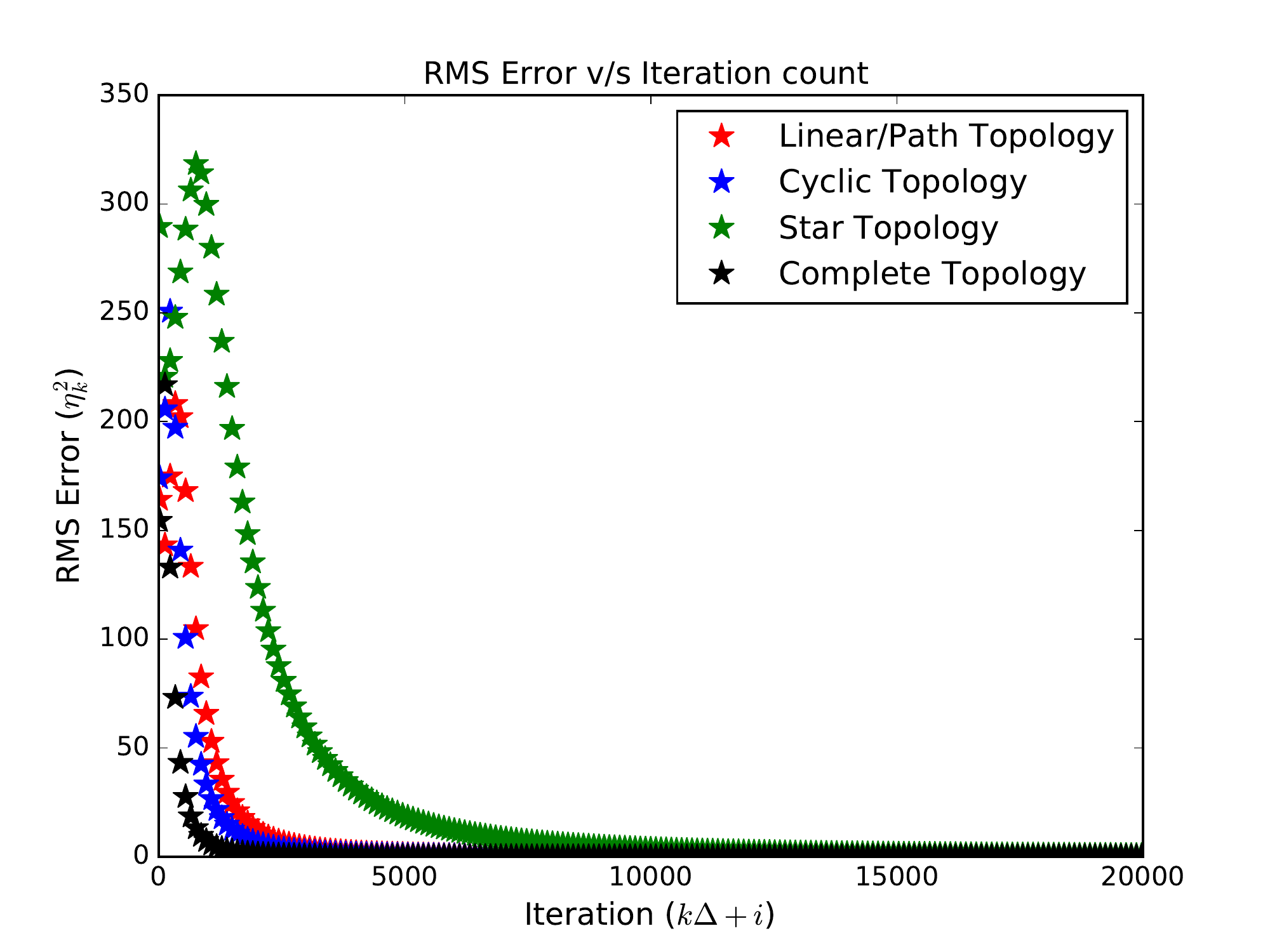}
  \caption{Evolution of RMS error time for different server topologies.}
  \label{Fig:ST-3}
\end{subfigure}
\caption{Effect of different Server-Server topologies (static case).}
\label{Fig:ST}
\end{figure}

We consider four different server topologies (constant over time) and corresponding doubly stochastic consensus weight matrices - Linear/Path graph (Figure~\ref{Fig:T-1}, ${B_k}_1$), Cyclic graph (Figure~\ref{Fig:T-2}, ${B_k}_2$), Star graph (Figure~\ref{Fig:T-3}, ${B_k}_3$) and Complete graph (Figure~\ref{Fig:T-4}, ${B_k}_4$),
\begin{align*}
&{B_k}_1 = \begin{bmatrix}
0.8 & 0.2 & 0 & 0.0\\
0.2 & 0.6 & 0.2 & 0.2\\
0.0   & 0.2 & 0.6 & 0.2 \\
0.0 & 0.0 & 0.2 & 0.8 \\
\end{bmatrix}, \;
{B_k}_2 = \begin{bmatrix}
0.6 & 0.2 & 0.0 & 0.2\\
0.2 & 0.6 & 0.2 & 0.0 \\
0.0 & 0.2 & 0.6 & 0.2 \\
0.2 & 0.0 & 0.2 & 0.6
\end{bmatrix}, \; \\ 
&{B_k}_3 = \begin{bmatrix}
0.4 & 0.2 & 0.2 & 0.2 \\
0.2 & 0.8 & 0.0 & 0.0 \\
0.2 & 0.0 & 0.8 & 0.0 \\
0.2 & 0.0 & 0.0 & 0.8
\end{bmatrix}, \; \text{and} \; {B_k}_4 = \begin{bmatrix}
0.4 & 0.2 & 0.2 & 0.2 \\
0.2 & 0.4 & 0.2 & 0.2\\
0.2 & 0.2 & 0.4 & 0.2 \\
0.2 & 0.2 & 0.2 & 0.4
\end{bmatrix}.
\end{align*}
Note that only Star and Complete graph topologies give rise to a scrambling transition matrix, $B_k$.

The simulation results in Figure~\ref{Fig:ST} show that the server-server communication topology has a significant impact on the speed of convergence. Star topology performs the worst in terms of speed of convergence (Figure~\ref{Fig:ST-1}) and reduction of RMS error (Figure~\ref{Fig:ST-3}). As expected, complete graph, server-server topology performs best, since it allows for fastest mixing of information. Linear and Cyclic graph both perform similar, however, Cyclic topology performs slightly better. This is also reasonable, since cyclic allows better mixing, (servers 1 and 4 are connected in cyclic graph, however, they are not connected in linear graph).

\section{Discussion}
We present a privacy promoting algorithm for distributed optimization in multi parameter server architecture. In Section~\ref{Sec:ConvergenceResults} correctness and optimality of the algorithm are proved. Simulation results presented in Section~\ref{Sec:SimulationResults} also show that iterates converge to the optimum. The speed of convergence depends on the communication topology between servers. We also show through simulation that the algorithm gets a significant speedup when the gradient weights are randomly changed at every iteration (as opposed to being constant). We reason that time-varying gradient weights ensures that no iterate gets pushed away from the optimum (negative weight) for extended periods of time. Hence, a iterate that gets kicked away from the equilibrium in an iteration, can easily get back towards optimum at the next instant when the weight becomes positive, resulting in faster convergence. 

We show convergence properties for this synchronous algorithm with a fixed $\Delta$. However, we conjecture that the algorithm will work as intended even when $\Delta$ changes during an execution, as long as it is frequent enough. This algorithm can also be tweaked to perform correctly in an asynchronous execution. Asynchronous algorithm along with its convergence analysis will be exhibited in a future technical report.   

Privacy enhancement comes from the random weights ($W$ matrix) used by the clients when the gradients are uploaded to parameter servers. Since, the weights are random, servers are oblivious to the true gradients. Enhanced privacy is achieved so long as the servers do not collaborate and share the gradients received from the clients directly. As the model is shared among servers, the parameter servers learn the correct model, however, cannot figure out specifics about the data improving privacy. We present a few intuitive arguments for privacy in function splitting approach (cf. \cite{gade16convsum}) in Section~\ref{Sec:PrivacyDiscussion} and draw parallels to motivate privacy in our algorithm. Detailed analysis of privacy guarantees will be presented in a future technical report. \\

\noindent \textbf{Optimization of Convex aggregate of Non-convex functions}

The analysis in Section~\ref{Sec:ConvergenceResults} motivates a solution for another interesting problem. Let us consider $C$ agents each endowed with a private objective function $f_i(x)$ that are not guaranteed to be convex, however the sum $f(x) = \sum_{i=1}^C f_i(x)$ is known to be convex. We refer to such an objective function ($f(x)$) as convex aggregate of non-convex functions. We intend to solve distributed learning problem (Problem~\ref{Prob:GlobalOpt}). We prove, in our report \cite{gade16convsum}, that coupled consensus and gradient descent algorithm presented in \cite{nedic2009distributed} can optimize convex sum of non-convex functions.

\section*{Acknowledgements}
Authors would like to thank and acknowledge Yuan-Ting Hu for stimulating discussions and valuable inputs.

\bibliography{Central.bib}
\bibliographystyle{ieeetr}

\newpage
\appendix
\section{Appendices}
\subsection{Alternate Proof for Eq.~\ref{Eq:ConsBoundktok1}} \label{sec:Apx0}
\begin{proof}
We begin with the consensus update equation in Eq.~\ref{Eq:ConsensusUpdateRelation1}, followed by subtracting vector $y \in \mathcal{X}$ on both sides.
\begin{align}
x^I_{k+1} &= \sum_{J=1}^S B_k[I,J] x^J_k       \qquad \qquad \qquad \ldots Eq.~\ref{Eq:ConsensusUpdateRelation1} \\
x^I_{k+1} - y &= \sum_{J=1}^S B_k[I,J] (x^J_k - y)  \qquad \ \quad \ldots y = \sum_{J=1}^S B_k[I,J] y \\
z^I_{k+1} &= \sum_{J=1}^S B_k[I,J] z^J_k \qquad \qquad \qquad \ldots z^I_{k+1} = x^I_{k+1} - y, \; \forall \ I = \{1, 2, \ldots, S\}
\end{align}
\noindent We now norm both sides of the quality and use the property that norm of sum is less than or equal to sum of norms.
\begin{align}
\|z^I_{k+1}\| = \|\sum_{J=1}^S B_k[I,J] z^J_k\| \leq \sum_{J=1}^S B_k[I,J] \|z^J_k\| \qquad \qquad \ldots B_k[I,J] \geq 0
\end{align}
\noindent Squaring both sides in the above inequality followed by algebraic expansion, we get, 
\begin{align}
\|z^I_{k+1}\|^2 &\leq \left(\sum_{J=1}^S B_k[I,J] \|z^J_k\|\right)^2 = \sum_{J=1}^S B_k[I,J]^2 \|z^J_k\|^2 + 2 \sum_{M \neq J} B_k[I,J] \ B_k[I,M] \ \|z^J_k\| \|z^M_k\| 
\end{align}
\noindent Now we use the property, $a^2 + b^2 \geq 2 a b$ for any $a, b$ in the latter term of the expansion above,
\begin{align}
\|z^I_{k+1}\|^2 &\leq \sum_{J=1}^S B_k[I,J]^2 \|z^J_k\|^2 + \sum_{M \neq J} B_k[I,J] \ B_k[I,M] \ \left(\|z^J_k\|^2 + \|z^M_k\|^2\right). 
\end{align}
Rearranging and using row stochasticity of $B_k$ matrix we get,
\begin{align}
\|z^I_{k+1}\|^2 &\leq \sum_{J=1}^S \left[ \|z^J_k\|^2 \left(B_k[I,J]^2 + \sum_{M \neq J} B_k[I,J] \ B_k[I,M] \right) \right]\\
&\leq \sum_{J=1}^S \left[ \|z^J_k\|^2 \left(B_k[I,J] \left( B_k[I,J] + \sum_{M \neq J} B_k[I,M] \right) \right) \right] \\
&\leq \sum_{J=1}^S B_k[I,J] \|z^J_k\|^2  \qquad \ldots B_k[I,J] + \sum_{M \neq J} B_k[I,M] = 1,\; \text{$B_k$ is row stochastic}
\end{align}
\noindent Summing the inequality over all servers, $I = 1, 2, \ldots, S$, 
\begin{align}
\sum_{I=1}^S \|z^I_{k+1}\|^2 &\leq \sum_{I=1}^S \sum_{J=1}^S B_k[I,J] \|z^J_k\|^2 = \sum_{J=1}^S \left(\|z^J_k\|^2 \left[ \sum_{I=1}^S B_k[I,J] \right] \right) \\
&\leq \sum_{J=1}^S \|z^J_k\|^2  \qquad \ldots \sum_{I=1}^S B_k[I,J] = 1,\; \text{$B_k$ is column stochastic}
\end{align}
\noindent This gives us Eq.~\ref{Eq:ConsBoundktok1},
\begin{align}
\sum_{J=1}^S \|x^J_{k+1} - y\|^2 \leq \sum_{J=1}^S \|x^J_k - y\|^2 
\end{align}
$\hfill \blacksquare$
\end{proof}

\subsection{Non-Negative Weight Matrix with Complete Server Graph} \label{sec:ApxA}
\subsubsection{Proof for Lemma~\ref{Lem:IterateConvNNWCSG}}
\begin{proof}
All servers agree to the state at the end of a cycle. This is a direct consequence of performing consensus after every $ \Delta$ projected gradient descent steps. 
\begin{equation}
x^I_{0,k} = x^J_{0,k} \qquad \forall \; I \neq J \text{and} \; \forall k 
\end{equation}

Using the non-expansivity property of the Projection operator, convexity of functions $f_i$, and gradient boundedness (Assumption~\ref{Asmp:SubBound}) we have,
\begin{align}
\| x^J_{i,k} - y \|^2 &= \| \mathcal{P}_\mathcal{X} [x^J_{i-1,k} - \alpha_k \sum_{h=1}^C W_{i-1,k}[J,h] \;  g_h(x^J_{i-1,k})] - y\|^2 \nonumber \\
&\leq \| x^J_{i-1,k} - \alpha_k \sum_{h=1}^C W_{i-1,k}[J,h] \  g_h(x^J_{i-1,k}) - y\|^2 \nonumber\\
&\leq \| x^J_{i-1,k} - y \|^2 +\alpha_k ^2 \left( \| \sum_{h=1}^C W_{i-1,k}[J,h] \  g_h(x^J_{i-1,k}) \|^2 \right) \nonumber \\
& \qquad \qquad \qquad \qquad \qquad \qquad - 2 \alpha_k \left(\sum_{h=1}^C W_{i-1,k}[J,h] \  g_h(x^J_{i-1,k}) \right)^T (x^J_{i-1,k} - y) \nonumber \\
&\leq \| x^J_{i-1,k} - y \|^2 + \alpha_k^2 \left( \sum_{h=1}^C W_{i-1,k}[J,h] \  L_h \right)^2 \nonumber \\
& \qquad \qquad \qquad \qquad \qquad \qquad - 2 \alpha_k \left(\sum_{h=1}^C W_{i-1,k}[J,h] \ \left(f_h(x^J_{i-1,k}) - f_h(y)\right) \right) \nonumber 
\end{align}

Adding the above inequality for all time instants within a cycle i.e. $i = 1, \; \cdots, \;  \Delta$, and cancelling telescoping terms;

\begin{align}
\| x^J_{ \Delta,k} - y \|^2 &\leq \| x^J_{0,k} - y \|^2 + \alpha_k^2 \sum_{i=1}^{ \Delta}\left(\sum_{h=1}^C W_{i-1,k}[J,h] \  L_h \right)^2 \nonumber \\ 
& \qquad - 2 \alpha_k \sum_{i=1}^{ \Delta}\left(\sum_{h=1}^C W_{i-1,k}[J,h] \ \left(f_h(x^J_{i-1,k}) - f_h(y)\right) \right)
\end{align}

Adding the above inequality for all servers, i.e. $J = 1, \; \cdots, \; S$;
\begin{align}
\sum_{J = 1}^S \| x^J_{ \Delta,k} - y \|^2 &\leq \sum_{J = 1}^S \| x^J_{0,k} - y \|^2 + \alpha_k^2 \sum_{J = 1}^S \sum_{i=1}^{ \Delta}\left( \sum_{h=1}^C W_{i-1,k}[J,h] \  L_h \right)^2 \nonumber \\ 
& \qquad - 2 \alpha_k \sum_{J = 1}^S \sum_{i=1}^{ \Delta}\left(\sum_{h=1}^C W_{i-1,k}[J,h] \ \left(f_h(x^J_{i-1,k}) - f_h(y)\right) \right)
\end{align}

We consider now the consensus step, use Jensen's inequality and non-negativity of norm to get,
\begin{align}
\| x^J_{0,k+1} - y \|^2 &= \| \frac{1}{S} \sum_{J = 1}^S x^J_{ \Delta,k} - y \|^2 \nonumber \\
&=  \| \frac{1}{S} \left( \sum_{J = 1}^S \left(x^J_{ \Delta,k} - y\right) \right) \|^2 \;
\leq \; \frac{1}{S} \sum_{J = 1}^S  \| x^J_{ \Delta,k} - y \|^2  {\; \leq \; \sum_{J = 1}^S  \| x^J_{ \Delta,k} - y \|^2} \label{Eq:ConsensusStep}
\end{align}

Combining the two inequalities presented above and strengthening the inequality we get,
\begin{align}
&\| x^J_{0,k+1} - y \|^2 \leq \frac{1}{S} \sum_{J = 1}^S  \| x^J_{ \Delta,k} - y \|^2 \leq \frac{1}{S} \sum_{J = 1}^S \| x^J_{0,k} - y \|^2 + \frac{1}{S} \alpha_k^2 \sum_{J = 1}^S \sum_{i=1}^{ \Delta}\left( \sum_{h=1}^C W_{i-1,k}[J,h] \  L_h\right)^2 \nonumber \\ 
& \qquad - 2 \frac{1}{S} \alpha_k \sum_{J = 1}^S \sum_{i=1}^{ \Delta}\left(\sum_{h=1}^C W_{i-1,k}[J,h] \ \left(f_h(x^J_{i-1,k}) - f_h(y)\right) \right) \nonumber \\
&\leq \frac{1}{S} \sum_{J = 1}^S \| x^J_{0,k} - y \|^2 - 2 \frac{M}{S} \alpha_k \left( f(x^J_{0,k}) - f(y)\right) + \frac{1}{S} \alpha_k^2 \sum_{J = 1}^S \sum_{i=1}^{ \Delta}\left(\sum_{h=1}^C W_{i-1,k}[J,h] \  L_h\right)^2 \nonumber \\ 
& \qquad - 2 \frac{1}{S} \alpha_k \sum_{J = 1}^S \sum_{i=1}^{ \Delta}\left(\sum_{h=1}^C W_{i-1,k}[J,h] \ \left(f_h(x^J_{i-1,k}) - f_h(x^J_{0,k})\right) \right) \ldots \text{Add-subtract $f_h(x_{0,k})$}  \nonumber \\
&\leq \frac{1}{S} \sum_{J = 1}^S \| x^J_{0,k} - y \|^2 - 2 \frac{M}{S} \alpha_k \left( f(x^J_{0,k}) - f(y)\right) + \frac{1}{S} \alpha_k^2 \sum_{J = 1}^S \sum_{i=1}^{ \Delta}\left( \sum_{h=1}^C W_{i-1,k}[J,h] \  L_h  \right)^2 \nonumber \\ 
& \qquad + 2 \frac{1}{S} \alpha_k \sum_{J = 1}^S \sum_{i=1}^{ \Delta}\left(\sum_{h=1}^C W_{i-1,k}[J,h] \ L_h \left( \| x^J_{i-1,k} - x^J_{0,k} \| \right) \right) \ldots \text{Strengthening the ineq.}  \nonumber  \\
&\leq \frac{1}{S} \sum_{J = 1}^S \| x^J_{0,k} - y \|^2 - 2 \frac{M}{S} \alpha_k \left( f(x^J_{0,k}) - f(y)\right) + \frac{1}{S} \alpha_k^2 \sum_{J = 1}^S \sum_{i=1}^{ \Delta}\left(\sum_{h=1}^C W_{i-1,k}[J,h] \  L_h  \right)^2 \nonumber \\ 
& \qquad + 2 \frac{1}{S} \alpha_k^2 \sum_{J = 1}^S \sum_{i=1}^{ \Delta}\left(\sum_{h=1}^C W_{i-1,k}[J,h] \ L_h \left(\sum_{t=1}^{i-1}\sum_{l = 1}^{C} W_{t-1,k}[J,l] L_l \right) \right) \ldots \text{Boundedness of gradient}    \nonumber\\
%&\leq \frac{1}{S} \sum_{J = 1}^S \| x^J_{0,k} - y \|^2 - 2 \frac{M}{S} \alpha_k \left( f(x^J_{0,k}) - f(y)\right) + \frac{1}{S} \alpha_k^2 \sum_{J = 1}^S \sum_{i=1}^{ \Delta}\left(\sum_{h=1}^C W_{i-1,k}[J,h] \  L_h \right)^2  \nonumber \\ 
%& \qquad + 2 \frac{1}{S} \alpha_k^2 \sum_{J = 1}^S \sum_{i=1}^{ \Delta}\left(\sum_{h=1}^C W_{i-1,k}[J,h] \ L_h \left(\sum_{t=1}^{i-1}\sum_{l = 1}^{C} W_{t-1,k}[J,l] L_l \right) \right)    \\
&\leq \frac{1}{S} \sum_{J = 1}^S \| x^J_{0,k} - y \|^2 - 2 \frac{M}{S} \alpha_k \left( f(x^J_{0,k}) - f(y)\right) \nonumber\\
& \qquad + \frac{1}{S} \alpha_k^2 \sum_{J = 1}^S \sum_{i=1}^{ \Delta}\left( \left(\sum_{h=1}^C W_{i-1,k}[J,h] \  L_h \right)^2  + 2 \left(\sum_{h=1}^C W_{i-1,k}[J,h] L_h \left( \sum_{t=1}^{i-1}\sum_{l = 1}^{C} W_{t-1,k}[J,l] L_l \right) \right) \right)  \nonumber   \\
&\leq \frac{1}{S} \sum_{J = 1}^S \| x^J_{0,k} - y \|^2 - 2 \frac{M}{S} \alpha_k \left( f(x^J_{0,k}) - f(y)\right) + \frac{1}{S} \alpha_k^2 C_0^2 \qquad  \ldots \text{Rearranging terms} \nonumber \\
&\leq \| x^J_{0,k} - y \|^2 - 2 \frac{M}{S} \alpha_k \left( f(x^J_{0,k}) - f(y)\right) + \frac{1}{S} \alpha_k^2 C_0^2   \qquad \ldots \text{$x^J_{0,k} = x^I_{0,k}$ for all $I,J$}
\label{Eq:ConvLem}
\end{align}
where, in the second to last step, using non-negativity of $W$ matrix we get, $$C_0^2 = \sum_{J = 1}^S \sum_{i=1}^{ \Delta}\left( \left(\sum_{h=1}^C W_{i-1,k}[J,h] \  L_h \right)^2  + 2 \left(\sum_{h=1}^C W_{i-1,k}[J,h] L_h \left( \sum_{t=1}^{i-1}\sum_{l = 1}^{C} W_{t-1,k}[J,l] L_l \right) \right) \right)$$ for some real number $C_0 > 0$.
$\hfill \blacksquare$
\end{proof}

\subsubsection{Proof for Theorem~\ref{Th:ConvNNWCSG}}
\begin{proof}
The proof closely follows the proof to Proposition 2.1 in \cite{nedic2001incremental}. We use the fact that once the iterate enters a certain level set it will not get too far away from the set. Let us fix a $\gamma > 0$ and let $k_0$ be such that $\gamma \geq \alpha_k C_0^2/(M)$ for all $k \geq k_0$. 

\begin{itemize}
\item[] Case 1: $f(x_{0,k}) > f^* + \gamma$

We have from Lemma~\ref{Lem:IterateConvNNWCSG}, for all $x^* \in \mathcal{X}^*$ and all $k \geq k_0$, 

\begin{align} 
\| x_{0,k+1} - x^* \|^2 &\leq \| x_{0,k} - x^* \|^2 - 2 \frac{M}{S} \alpha_k \left( f(x^J_{0,k}) - f(x^*)\right) + \frac{1}{S} \alpha_k^2 C_0^2 \nonumber \\
&\leq \| x_{0,k} - x^* \|^2 - 2 \frac{M}{S} \alpha_k \gamma + \frac{1}{S} \alpha_k^2 C_0^2 \nonumber \\
&\leq \| x_{0,k} - x^* \|^2 -  \frac{\alpha_k M}{S} (2 \gamma - \frac{\alpha_k C_0^2}{M}) \; \leq \; \| x_{0,k} - x^* \|^2 - \frac{M \alpha_k \gamma}{S}. \nonumber
\end{align}

The above inequality implies that,
\begin{equation}
dist(x_{0,k+1}, \mathcal{X}^*) \leq dist(x_{0,k}, \mathcal{X}^*) - \frac{M \alpha_k \gamma}{S}.
\label{Eq:ConvResult-1}
\end{equation}
Hence, for all $k \geq k_0$, the distance of the iterate from the optimal set $\mathcal{X}^*$ is decreasing due to assumptions on $\alpha_k$ (so long as $f(x_{0,k}) > f^* + \gamma$ remains valid). 
%So either the iterate enters the optimal set (and we are done) or the initial assumption $f(x_{0,k}) > f^* + \gamma$ gets violated (and we have Case 2 for this scenario).

\item[] Case 2: $f(x_{0,k}) \leq f^* + \gamma$

This scenario (Case 2) occurs for infinitely many $k$. This follows directly from Eq.~\ref{Eq:ConvResult-1}, and the assumptions on $\alpha_k$ \footnote{To elaborate further, whenever $f(x_{0,k}) > f^* + \gamma$ is satisfied, $x_{0,k}$ moves closer to $x^*$ and $f(x_{0,k})$ moves closer to $f^*$ (this follows from continuity of f). Hence, $f(x_{0,k}) \leq f^* + \gamma$ will happen infinitely often. Even if for certain iterations, the function value increases and Case 1 gets satisfied, the iterate will be pushed back so that Case 2 occurs.}. $x_{0,k}$ belongs to the level set $L_{\gamma} = \{ y \in \mathcal{X} | f(y) \leq f^* + \gamma \}$. The level set is bounded due to boundedness of $\mathcal{X}^*$ and we have,
$$ dist(x_{0,k}, \mathcal{X}^*) \leq d(\gamma) \triangleq \max_{y \in L_\gamma} dist(y,\mathcal{X}^*)  < \infty.$$
We get from Eq.~\ref{Eq:ConvLem}, $\| x_{0,k+1} - x_{0,k} \| \leq \alpha_k C_0/\sqrt{S}$, by substituting $y = x_{0,k}$. Further using triangle inequality as $\| x_{0,k+1} - x^* \| \leq \| x_{0,k+1} - x_{0,k} \| + \| x_{0,k} - x^* \|$, we have,
$$dist(x_{0,k+1}, \mathcal{X}^*) =  \| x_{0,k+1} - x^* \| \leq  \| x_{0,k} - x^* \| + \| x_{0,k+1} - x_{0,k} \| \leq dist(x_{0,k},\mathcal{X}^*) + \alpha_k C_0 / \sqrt{S}$$
$$dist(x_{0,k+1}, \mathcal{X}^*)  \leq d(\gamma)+ \alpha_k C_0 / \sqrt{S} \qquad \forall \; k \geq k_0.$$

Therefore, as $\alpha_k \rightarrow 0$, $$\limsup_{k \rightarrow \infty} dist(x_{0,k}, \mathcal{X}^*) \leq d(\gamma) \qquad \forall \; \gamma > 0.$$

Since, $f$ is continuous and its level sets are compact, we get $\lim_{\gamma \rightarrow 0} d(\gamma) = 0 $ which clearly implies, $$\lim_{k \rightarrow \infty} dist(x_{0,k}, \mathcal{X}^*) = 0$$  and consequently, $$\lim_{k \rightarrow \infty} f(x_{0,k}) = f^*.$$ 
\end{itemize} 
$\hfill \blacksquare$
\end{proof}

\subsection{Complete Server Graph} \label{sec:ApxB}
\subsubsection{Proof for Lemma~\ref{Lem:IterateConvCSG}}
\begin{proof}
We start with projection gradient descent method and establish error between the iterate and any point $y \in \mathcal{X}$,
\begin{align}
\| x^J_{i,k} - y \|^2 &= \| \mathcal{P}_\mathcal{X} [x^J_{i-1,k} - \alpha_k \sum_{h=1}^C W_{i-1,k}[J,h] \;  g_h(x^J_{i-1,k})] - y\|^2 \nonumber \\
&\leq \| x^J_{i-1,k} - \alpha_k \sum_{h=1}^C W_{i-1,k}[J,h] \  g_h(x^J_{i-1,k}) - y\|^2 \nonumber\\
&\leq \| x^J_{i-1,k} - y \|^2 +\alpha_k ^2 \left( \| \sum_{h=1}^C W_{i-1,k}[J,h] \  g_h(x^J_{i-1,k}) \|^2 \right)  \nonumber\\
& \qquad - 2 \alpha_k \left(\sum_{h=1}^C W_{i-1,k}[J,h] \  g_h(x^J_{i-1,k}) \right)^T (x^J_{i-1,k} - y)  \label{Eq:Projection} 
\end{align}
Adding the above inequality for all time instants within a cycle i.e. $i = 1, \; \cdots, \;  \Delta$, followed by a summation over all agents $J$ and the consensus step Eq.~\ref{Eq:ConsensusStep} ;
\begin{align}
\| x_{0,k+1} - y \|^2 \leq \sum_{J=1}^S\| x^J_{ \Delta,k} - y \|^2 &\leq \sum_{J=1}^S \| x^J_{0,k} - y \|^2 + \alpha_k^2 \sum_{J=1}^S \sum_{i=1}^{ \Delta}\left(\sum_{h=1}^C W_{i-1,k}[J,h] \  L_h \right)^2 \nonumber \\ 
& \hspace{-0.5in} \underbrace{- 2 \alpha_k \sum_{J=1}^S \sum_{i=1}^{ \Delta} \left[ \left(\sum_{h=1}^C W_{i-1,k}[J,h] \  g_h(x^J_{i-1,k}) \right)^T (x^J_{i-1,k} - y) \right]}_{{\textstyle \Lambda}} \label{Eq:Interim}
\end{align}
Typically one would use the gradient/gradient property to rewrite the third term ($\Lambda$) in Eq.~\ref{Eq:Interim} as a sum of functions. This is however not possible here, since the multipliers $W_{i-1,k} [J,h]$ could be negative. We will manipulate the third term and use bounds to obtain something involving addition/subtraction of function values (evaluated at iterate and $y$). 
We start by writing an expression for $x^J_{i,k}$,
$$ x^J_{i,k} = x^J_{i-1,k} - \alpha_k \left[ \sum_{h=1}^C W_{i-1,k} [J,h] g_h(x^J_{i-1,k})\right] + e^J_{i-1} $$
where $e^J_{i-1}$ is the difference term between a gradient descent and its projection at $i^{th}$ iteration. Recursively performing this unrolling operation we get,
$$ x^J_{i,k} = x^J_{0,k} - \alpha_k \sum_{t = 1}^{i} \left[ \sum_{h=1}^C W_{t-1,k} [J,h] g_h(x^J_{t-1,k})\right] + \sum_{t = 1}^{i} e^J_{t-1}.$$
Each $e^J_{i-1}$ has a finite bound and it can be shown by using the above expressions and triangle inequality,
\begin{align} 
\| e^J_{i-1} \| &= \| x^J_{i,k} - x^J_{i-1,k} + \alpha_k \sum_{h = 1}^C W_{i-1,k}[J, h] g_h(x^J_{i-1,k}) \| \nonumber\\
& \leq \|x^J_{i,k} - x^J_{i-1,k} \| +  \alpha_k \|\sum_{h = 1}^C W_{i-1,k}[J, h] g_h(x^J_{i-1,k}) \| \nonumber\\
& \leq 2 \alpha_k\| \sum_{h = 1}^C W_{i-1,k}[J, h] g_h(x^J_{i-1,k}) \| \nonumber \\
& \qquad \ldots \text{From  Eq.~\ref{Eq:Projection}, $\|x^J_{i,k} - x^J_{i-1,k} \| \leq  \alpha_k  \| \sum_{h = 1}^C W_{i-1,k}[J, h] g_h(x^J_{i-1,k}) \|$}
\end{align}
We hence arrive at the following expression, that will be used later in the analysis. Note the terms $\alpha_k D^J_{i-1,k}$ and $\alpha_k E^J_{i-1,k}$ will be used extensively later.
\begin{align}
 y - x^J_{i-1,k} &= y - x^J_{0,k} + \underbrace{\alpha_k \sum_{t = 1}^{i-1} \left[ \sum_{h=1}^C W_{t-1,k} [J,h] g_h(x^J_{t-1,k})\right]}_{\alpha_k D^J_{i-1,k}} - \underbrace{\sum_{t = 1}^{i-1} e^J_{t-1}}_{\alpha_k E^J_{i-1,k}} \nonumber \\ 
 &\triangleq y - x^J_{0,k} + \alpha_k ( D^J_{i-1,k} - E^J_{i-1,k}). \label{Eq:EstErrEqn}
 \end{align}
We obtain bounds on both $\sum_{J=1}^S D^J_{i-1,k}$ and $\sum_{J=1}^S E^J_{i-1,k}$ using the property $\| \sum (.) \| \leq \sum \| (.) \|$ successively,

\begin{align}
\sum_{J=1}^S \|E_{i-1,k} \| &= \sum_{J=1}^S \| \frac{1}{\alpha_k} \sum_{t = 1}^{i-1} e_{t-1}\| \leq  \frac{1}{\alpha_k} \sum_{J=1}^S \sum_{t = 1}^{i-1} \|e_{t-1}\| \leq \frac{1}{\alpha_k} \sum_{J=1}^S \sum_{t = 1}^{i-1} 2 \alpha_k\| \sum_{h = 1}^C W_{i-1,k}[J, h] g_h(x_{i-1,k}) \| \nonumber \\
&\leq 2 \sum_{J=1}^S \sum_{t = 1}^{i-1}  \sum_{h = 1}^C | W_{i-1,k}[J, h]| \ \|g_h(x_{i-1,k}) \| \leq 2 \sum_{J=1}^S \sum_{t = 1}^{ \Delta}  \sum_{h = 1}^C | W_{i-1,k}[J, h] | \ \|g_h(x_{i-1,k}) \| \nonumber \\
&\leq 2 \bar{M} \sum_{h=1}^{C} L_h = 2 \bar{M} \SB{L} \quad \ldots \text{Replace $i-1$ with $ \Delta$} \label{Eq:EBound3} \\
\sum_{J=1}^S\|D_{i-1,k}\| & = \| \sum_{J=1}^S \sum_{t = 1}^{i-1} \left[ \sum_{h=1}^C W_{t-1,k} [J,h] g_h(x_{t-1,k})\right] \| \leq \sum_{J=1}^S \sum_{t = 1}^{i-1} \left[ \sum_{h=1}^C | W_{t-1,k} [J,h] | \ \|g_h(x_{t-1,k})\|\right] \nonumber \\
& \leq \sum_{J=1}^S \sum_{t = 1}^{ \Delta} \left[ \sum_{h=1}^C | W_{t-1,k} [J,h] | \ \|g_h(x_{t-1,k})\|\right] \leq  2 \bar{M} \sum_{h=1}^{C} L_h = 2 \bar{M} \SB{L}\label{Eq:DBound3}
\end{align} 

We now perform similar unrolling for the gradient function. The gradient is Lipschitz bounded, hence we have, $\| g_h(x_{i-1,k}) - g_h(x_{0,k})\| \leq N_h \| x_{i-1,k} - x_{0,k}\|$. This gives us,
\begin{align} 
g_h(x_{i-1,k}) =  g_h(x_{0,k}) + l_{h,i-1} \quad \text{where,} \; \| l_{h,i-1} \| \leq N_h \| x_{i-1,k} - x_{0,k}\|, 
\label{Eq:UnrollGrad3}
\end{align}
where, the vector $l_{h,i-1}$ is such that its norm is less than or equal to $N_h \| x_{i-1,k} - x_{0,k}\|$. We further bound $\| x_{i-1,k} - x_{0,k} \|$ using Eq.~\ref{Eq:EstErrEqn} and using triangle inequality,
$$ \| x_{i-1,k} - x_{0,k} \| = \alpha_k \| D^J_{i-1,k} - E_{i-1,k}\| \leq \alpha_k (\| D^J_{i-1,k}\| + \|E_{i-1,k}\|) \leq 4 \alpha_k \bar{M} \SB{L}$$
$$ \| l_{h,i-1} \| \leq N_h \| x_{i-1,k} - x_{0,k} \| \leq 4 \alpha_k \bar{M} N_h \SB{L} $$
Now we move on to using these expressions to bound terms in $\Lambda$.
\begin{align*}
\Lambda &= 2 \alpha_k \sum_{J=1}^S \sum_{i=1}^{ \Delta} \left[ \left(\sum_{h=1}^C W_{i-1,k}[J,h] \ \left( g_h(x_{0,k}) + l_{h,i-1}\right) \right)^T (y - x^J_{i-1,k}) \right] \ldots \text{Unroll $g_h(x)$, Eq.~\ref{Eq:UnrollGrad3}} \nonumber\\
&= 2 \alpha_k \sum_{J=1}^S  \sum_{i=1}^{ \Delta} \left[ \left(\sum_{h=1}^C W_{i-1,k}[J,h] \  g_h(x_{0,k})\right) ^T (y - x_{0,k} + \alpha_k (D_{i-1,k} - E_{i-1,k}))  \right] \ldots \text{Eq.~\ref{Eq:EstErrEqn}} \nonumber\\ 
& \quad + 2 \alpha_k \sum_{J=1}^S  \sum_{i=1}^{ \Delta} \left[ \left(\sum_{h=1}^C W_{i-1,k}[J,h] \  l_{h,i-1} \right)^T (y - x_{0,k} + \alpha_k (D_{i-1,k} - E_{i-1,k}))  \right] \nonumber \\
&\leq 2 \alpha_k \left[- M (f(x_{0,k}) - f(y)) + \alpha_k \sum_{J=1}^S \sum_{i=1}^{ \Delta} \left[ \left(\sum_{h=1}^C W_{i-1,k}[J,h] \  g_h(x_{0,k})\right)^T (D_{i-1,k} - E_{i-1,k})  \right] \right]  \nonumber\\ 
& \quad + 2 \alpha_k \sum_{J=1}^S \sum_{i=1}^{ \Delta} \left[ \left(\sum_{h=1}^C W_{i-1,k}[J,h] \  l_{h,i-1} \right)^T (y - x_{0,k} + \alpha_k (D_{i-1,k} - E_{i-1,k}))  \right] \ldots \text{Subgrad. ineq.}  \nonumber \\
& \quad \text{Strengthening the inequality, } \nonumber \\
&\leq - 2 \alpha_k M (f(x_{0,k}) - f(y)) + 2 \alpha_k^2 \| \sum_{J=1}^S \sum_{i=1}^{ \Delta} \left[ \left(\sum_{h=1}^C W_{i-1,k}[J,h] \  g_h(x_{0,k})\right)^T (D_{i-1,k} - E_{i-1,k})  \right] \|  \nonumber \\
& \quad + 2 \alpha_k \| \sum_{J=1}^S \sum_{i=1}^{ \Delta} \left[ \left(\sum_{h=1}^C W_{i-1,k}[J,h] \  l_{h,i-1} \right)^T (y - x_{0,k} + \alpha_k (D_{i-1,k} - E_{i-1,k}))  \right]\|\nonumber  \\
& \quad \text{Using property of norms, $\| \sum (.) \|  \leq \sum \| (.) \|$,} \nonumber \\
&\leq - 2 \alpha_k M (f(x_{0,k}) - f(y)) + 2 \alpha_k^2 \sum_{J=1}^S  \sum_{i=1}^{ \Delta} \left[ \left(\sum_{h=1}^C| W_{i-1,k}[J,h]| \  \|g_h(x_{0,k}) \|\right) \|(D_{i-1,k} - E_{i-1,k})\|  \right]   \nonumber \\
& \quad + 2 \alpha_k \sum_{J=1}^S \sum_{i=1}^{ \Delta} \left[ \left(\sum_{h=1}^C | W_{i-1,k}[J,h] | \ \| l_{h,i-1} \| \right) \| (y - x_{0,k} + \alpha_k (D_{i-1,k} - E_{i-1,k}))\|  \right] \nonumber  \\
& \quad \text{Using the property, $\sum ab \leq \sum a \sum b$, if $a,b \geq 0$, and triangle inequality, } \nonumber \\
& \leq - 2 \alpha_k M (f(x_{0,k}) - f(y)) + 2 \alpha_k^2  \bar{M} (\sum_{h=1}^{C} L_h) \sum_{J=1}^S \sum_{i=1}^{ \Delta} \left[ \|(D_{i-1,k}\| + \| E_{i-1,k})\|  \right]    \nonumber \\
& \quad + 2 \alpha_k \sum_{J=1}^S \sum_{i=1}^{ \Delta} \left[ \left(\sum_{h=1}^C |W_{i-1,k}[J,h]| \ \| l_{h,i-1} \| \right) \left(\| (y - x_{0,k})\| + \alpha_k \| D^J_{i-1,k} \| + \alpha_k \|E^J_{i-1,k} \| \right)  \right] \nonumber  \\
& \leq - 2 \alpha_k M (f(x_{0,k}) - f(y)) + 8  \Delta \ \alpha_k^2  \bar{M}^2   (\sum_{h=1}^{C} L_h )^2 \ldots \text{Eq.~\ref{Eq:EBound3} and \ref{Eq:DBound3}}  \nonumber \\
& \quad + 2 \alpha_k^2 \ \| (y - x_{0,k})\| \ \sum_{J=1}^S \sum_{i=1}^{ \Delta} \left[ \left( \sum_{h=1}^C | W_{i-1,k}[J,h] | \ N_h (\| D^J_{i-1,k} \| + \| E^J_{i-1,k} \|) \right)  \right] \nonumber \\
& \quad + 2 \alpha_k^3 \sum_{J=1}^S \sum_{i=1}^{ \Delta} \left[ \left(\sum_{h=1}^C | W_{i-1,k}[J,h] | \ N_h (\| D^J_{i-1,k} \| + \| E^J_{i-1,k} \|)^2 \right) \right] \nonumber  \\
% & \leq 2 \alpha_k M (f(x_{0,k}) - f(y)) + 8  \Delta \ \alpha_k^2  \bar{M}^2   (\sum_{h=1}^{C} L_h )^2  \nonumber \\
% & \qquad + 2 \alpha_k^2 (\bar{M} \SB{L}') (4 \bar{M} \SB{L}') \nonumber \\
% & \qquad + 2 \alpha_k^3 \sum_{J=1}^S \sum_{i=1}^{ \Delta} \left[ \left(\sum_{h=1}^C \| W_{i-1,k}[J,h] \| \ N_h (\| D^J_{i-1,k} \| + \| E^J_{i-1,k} \|)^2 \right) \right] \nonumber  \\
&\leq - 2 \alpha_k M (f(x_{0,k}) - f(y)) + 8  \Delta \ \alpha_k^2  \bar{M}^2   (\sum_{h=1}^{C} L_h )^2  + 2 \alpha_k^2 \ \| (y - x_{0,k})\| \left[ \bar{M} \SB{N}  + 4 C  \Delta \bar{M} \SB{L} \right]  \nonumber \\
& \quad + 2 \alpha_k^3 \left[ \bar{M}\SB{N}  + C  \Delta \left( \sum_{J=1}^S \| D^J_{i-1,k} \| + \sum_{J=1}^S \| E^J_{i-1,k} \| \right)^2 \right] \nonumber  \\
&\leq - 2 \alpha_k M (f(x_{0,k}) - f(y)) + 8  \Delta \ \alpha_k^2  \bar{M}^2   (\sum_{h=1}^{C} L_h )^2  + 2 \alpha_k^2 \ \| (y - x_{0,k})\| \left[ \bar{M} \SB{N}  + 4 C  \Delta \bar{M} \SB{L} \right] \nonumber \\
& \quad + 2 \alpha_k^3 \left[ \bar{M}\SB{N}  + C  \Delta (4 \bar{M} \SB{L})^2\right] \nonumber  \\
& \leq - 2 \alpha_k M (f(x_{0,k}) - f(y)) + 2 \alpha_k^2 \left[ \bar{M} \SB{N}  + 4 C  \Delta \bar{M} \SB{L} \right] \ \| (y - x_{0,k})\|  + 8  \Delta \ \alpha_k^2  \bar{M}^2   (\sum_{h=1}^{C} L_h )^2 + \tilde{C}^2 \nonumber 
\end{align*}
where $\tilde{C}^2 = 2 \alpha_k^3 \left[ \bar{M}\SB{N}  + C  \Delta (4 \bar{M} \SB{L})^2\right]$. Substituting the above inequality in Eq.~\ref{Eq:Interim} and using the property $2\|x\|  \leq 1 + \|x\|^2$,
\begin{align}
\| x^J_{0,k+1} - y \|^2 \leq \| x^J_{ \Delta,k} - y \|^2 &\leq (1 + \frac{\alpha_k^2}{S} \bar{M} F) \| x^J_{0,k} - y \|^2 - 2  \frac{1}{S}  \alpha_k M (f(x_{0,k}) - f(y)) +  \frac{1}{S} \alpha_k^2 C_0^2
\label{Eq:Complete}
\end{align}
where, $$C_0^2 = \sum_{J=1}^S \sum_{i=1}^{ \Delta}\left(\sum_{h=1}^C  W_{i-1,k}[J,h] \ L_h \right)^2 + 8  \Delta \bar{M}^2   (\sum_{h=1}^{C} L_h )^2 + 2 \alpha_k \left[ \bar{M}\SB{N}  + C  \Delta (4 \bar{M} \SB{L})^2\right] + \bar{M}F $$ and $F = \left[  \SB{N}  + 4 C  \Delta \SB{L} \right]$. 
$\hfill \blacksquare$
\end{proof}

\subsubsection{Proof for Theorem~\ref{Th:ConvCSG}}
\begin{proof}
We will use the deterministic version of  Lemma~\ref{Lem:RobSiegConv} for proving convergence. We being by rewriting the relationship between iterates from Lemma~\ref{Lem:IterateConvCSG}, for $y = x^* \in \mathcal{X}^*$,
\begin{align}
\| x^J_{0,k+1} - x^*\|^2 &\leq (1 + \underbrace{\frac{\alpha_k^2}{S} \bar{M} F}_{q_k}) \| x^J_{0,k} - x^* \|^2 - 2  \frac{1}{S}  \alpha_k M (f(x_{0,k}) - f(x^*)) +  \underbrace{\frac{1}{S} \alpha_k^2 C_0^2}_{w_k}
\label{Eq:Complete2}
\end{align}
where, $$C_0^2 = \sum_{J=1}^S \sum_{i=1}^{ \Delta}\left(\sum_{h=1}^C W_{i-1,k}[J,h] \  L_h \right)^2 + 8  \Delta \bar{M}^2   (\sum_{h=1}^{C} L_h )^2 + 2 \alpha_k \left[ \bar{M}\SB{N}  + C  \Delta (4 \bar{M} \SB{L})^2\right] + \bar{M}F $$ and $F = \left[  \SB{N}  + 4 C  \Delta \SB{L} \right]$.
Clearly, the conditions required by Lemma~\ref{Lem:RobSiegConv} hold viz. $\sum_{k=0}^\infty q_k < \infty$ and $\sum_{k=0}^\infty w_k < \infty$ ($\sum_{k=1}^\infty \alpha_k^2 < \infty$). The rest of the proof is similar to the proof of Theorem~\ref{Th:ConvMain} and it follows directly that $x^J_{0,k}$ enters the optimal set.
$\hfill \blacksquare$
\end{proof}

\end{document}